\documentclass[11pt]{article}
\usepackage{graphicx} 
\usepackage{xcolor}
\definecolor{ForestGreen}{rgb}{0.1333,0.5451,0.1333}
\definecolor{DarkRed}{rgb}{0.65,0,0}
\definecolor{Red}{rgb}{1,0,0}
\usepackage[linktocpage=true,
pagebackref=true,colorlinks,
linkcolor=DarkRed,citecolor=ForestGreen,
bookmarks,bookmarksopen,bookmarksnumbered]
{hyperref}

\usepackage{enumerate}

\usepackage[margin=1in]{geometry}                		
\usepackage{amssymb}
\usepackage{amsthm}
\usepackage{amsmath}
\usepackage{thm-restate}

\usepackage{algpseudocode}
\usepackage{algorithm}
\usepackage{float}
\usepackage{tablefootnote}
\usepackage{longtable}

\usepackage{xspace}
\usepackage[colorinlistoftodos,textsize=tiny,textwidth=2cm,color=red!25!white,obeyFinal]{todonotes}
\usepackage{cleveref}

\newtheorem{theorem}{Theorem}[section]
\newtheorem{hypothesis}[theorem]{Hypothesis}
\newtheorem{lemma}[theorem]{Lemma}
\newtheorem{observation}[theorem]{Observation}

\newtheorem{definition}[theorem]{Definition}

\newtheorem{claim}[theorem]{Claim}

\newcommand{\poly}{\mathrm{poly}\xspace}
\newcommand{\polylog}{\mathrm{polylog}\xspace}
\newcommand{\vol}{\mathrm{vol}}
\newcommand{\CC}{\mathcal{C}}
\newcommand{\MM}{\mathcal{M}}
\renewcommand{\SS}{\mathcal{S}}
\newcommand{\OO}{\mathcal{O}}
\newcommand{\OTIL}{\tilde{\mathcal{O}}}
\newcommand{\IP}{\textbf{Input: }}
\newcommand{\OP}{\textbf{Output: }}
\newcommand{\IC}{{\textsc{ImproveCut}}}
\newcommand{\SC}{{\textsc{Sparsest Cut}}}
\newcommand{\VSC}{{\textsc{Vertex Sparsest Cut}}}
\newcommand{\SSVE}{{\textsc{Small Set Vertex Expansion}}}
\newcommand{\SSE}{{\textsc{Small Set Expansion}}}
\newcommand{\DSH}{{\textsc{Densest $k$-SubHypergraph}}}
\newcommand{\DKS}{{\textsc{Densest $k$-Subgraph}}}
\newcommand{\MMGP}{{\textsc{Min-Max Graph Partitioning}}}

\renewcommand{\AA}{\mathcal{A}}
\newcommand{\STE}{{\textsc{Small Set Terminal Expansion}}}

\newcommand{\eps}{\epsilon}

\newcommand{\Wahlstrom}{Wahlstr\"{o}m}

\newcommand{\defparproblem}[4]{
  \vspace{1mm}
\begin{center}
\noindent\fbox{

  \begin{minipage}{0.95\textwidth}
  \begin{tabular*}{\textwidth}{@{\extracolsep{\fill}}lr} \textsc{#1}  & {\bf{Parameter:}} #3 \\ \end{tabular*}
  {\bf{Input:}} #2  \\
  {\bf{Question:}} #4
  \end{minipage}
 
  }
  \end{center}
 }

\title{Approximating Small Sparse Cuts}
\author{Aditya Anand\thanks{University of Michigan, Ann Arbor, USA} \and Euiwoong Lee\thanks{University of Michigan, Ann Arbor, USA. Supported in part by NSF grant CCF-2236669 and Google} \and Jason Li\thanks{Carnegie Mellon University, Pittsburgh, USA} \and Thatchaphol Saranurak\thanks{University of Michigan, Ann Arbor, USA. Supported by NSF grant CCF-2238138}}

\date{}

\begin{document}

\maketitle
\pagenumbering{gobble}
\begin{abstract}
We study polynomial-time approximation algorithms for (edge/vertex) \SC{} and \SSE{} in terms of $k$, the number of edges or vertices cut in the optimal solution. Our main results are $\OO(\polylog\, k)$-approximation algorithms for various versions in this setting. 

Our techniques involve an extension of the notion of sample sets (Feige and Mahdian STOC'06), originally developed for small balanced cuts, to sparse cuts in general. We then show how to combine this notion of sample sets with two algorithms, one based on an existing framework of LP rounding and another new algorithm based on the cut-matching game, to get such approximation algorithms. Our cut-matching game algorithm can be viewed as a local version of the cut-matching game by Khandekar, Khot, Orecchia and Vishnoi
and certifies an expansion of every vertex set of size $s$ in $\OO(\log s)$ rounds. These techniques may be of independent interest. 

As corollaries of our results, we also obtain an $\OO(\log opt)$-approximation for min-max graph partitioning, where $opt$ is the min-max value of the optimal cut, and improve the bound on the size of multicut mimicking networks computable in polynomial time. 
\end{abstract}

\newpage

\tableofcontents
\newpage


\pagenumbering{arabic}
\section{Introduction}
Given a graph $G$ and a cut $(X, \overline{X})$, the {\em sparsity} of the cut $(X,\overline{X})$ is given as $\frac{|\delta_G(X)|}{\min(|X|, |\overline{X}|)}$, where $\delta_G(X)$ denotes the set of edges with exactly one endpoint in $X$.
\SC{} is the problem whose goal is to find a cut with minimum sparsity.
It is one of the most fundamental and well-studied problems in multiple areas of algorithms. 
In approximation algorithms, this problem has been the focus of the beautiful connection between algorithms, optimization, and geometry via metric embeddings~\cite{lr99, linial1995geometry, arora2005euclidean, arv09, khot2015unique, naor2018vertical}.
Through expander decomposition~\cite{kvv04,sw19}, 
finding sparse cuts is also one of the most fundamental building blocks in designing algorithms for numerous graph problems in many areas, including graph sketching/sparsification~\cite{ackqwz16,js18,cgpssw20}, Laplacian solvers~\cite{st04,ckpprsa17}, maximum flows~\cite{klos14}, \textsc{Unique Games}~\cite{trevisan05} and dynamic algorithms~\cite{ns17,wul17,nsw17}.
Over the years, there has been a lot of progress in approximating sparse cuts~\cite{lr99, arv09, fhl05, krvactual} culminating in the seminal paper of Arora, Rao and Vazirani~\cite{arv09}, which gave an $\OO(\sqrt{\log n})$-approximation for sparsest cut.

A closely related problem, which can be viewed as a generalization of \SC, is \SSE{}, introduced by Raghavendra and Steurer~\cite{rs10}.
 In \SSE{}, we are additionally given a size parameter $s$, and the goal is to find a cut $(Y,\overline{Y})$ with minimum sparsity satisfying $\min(|Y|, |\overline{Y}|) \leq s$.\footnote{For \SSE, we are typically interested in (bi-criteria) approximation algorithms which only violate the size constraint by a constant factor.}
 Closely related to the Unique Games Conjecture, it is one of the central problems in the hardness of approximation literature with applications to graph partitioning, continuous optimization, and quantum computing~\cite{raghavendra2010approximations, arora2015subexponential, minmax, raghavendra2012reductions,  barak2012hypercontractivity, louis2014approximation,bhattiprolu2021framework, ghoshal2021approximation, lee2022characterization}. Using the hierarchical tree decomposition of Racke~\cite{racke08} one obtains an $\OO(\log n)$-approximation for \SSE.

Suppose the optimal cut $(S, 
\overline{S})$ for \SC/\SSE{} in a graph cuts $k$ edges, while separating $|S| = s$ vertices from the rest of the graph (so that it is $\phi = \frac{k}{s}$-sparse). Most approximation algorithms studied for both these problems have an approximation ratio that is a function of $n$, the number of vertices. However, it is natural to expect that for an instance 
 we may have $k \ll n$. Thus a natural question is: can we obtain a \emph{parametric approximation} -  an approximation algorithm whose approximation ratio is a function of $k$ alone? 

 Clearly $k \leq n^2$, and such an approximation can be much better if the sparse cut cuts only a few edges.

 The notion of parametric approximation is closely related to that of Fixed Parameter Tractability (FPT), where given a parameter $\ell$, we seek exact algorithms that run in time $f(\ell) n^{\OO(1)}$. Besides the fact that both these notions seek to improve performance when the parameter $\ell$ is small, it is easy to see that if a minimization problem $P$ is fixed parameter tractable with respect to a parameter $\ell$, then it also admits an $g(\ell)$ approximation in polynomial time for some function $g$. On the other hand, parametric approximations have often been used as the starting point in the design of polynomial kernels -  see for instance the kernels for Chordal Vertex Deletion~\cite{jp18,agrawal18} and $\mathcal{F}$-minor deletion~\cite{fomin2016hitting}.

 The result of Cygan et al.~\cite{cygan20} for \textsc{Minimum Bisection} shows that \SSE{} admits an FPT algorithm running in time $2^{\OO(k\log k)} n^{\OO(1)}$. Combining this with the known $\OO(\log n)$ approximation for \SSE{} and the $\OO(\sqrt{\log n})$ approximation for \SC{} gives an $\OO(k \log k)$ approximation for \SSE{} and an $\OO(\sqrt{k\log k})$ approximation for \SC{} in polynomial time.

 This leads to the natural question: can we obtain an $\OO(\polylog(k))$ approximation for these problems in polynomial time? This question was also raised by \Wahlstrom~\cite{wahlstrom20,wahlstrom22}, who showed that this would have implications in designing quasipolynomial kernels for various cut problems. The work of Feige et al.~\cite{fhl05} shows that one can obtain an $\OO(\sqrt{\log k})$ approximation when $s = \Omega(n)$, but no such guarantee is known for general $s$.

 In this work, we systematically study parametric approximation algorithms for \SC{} and \SSE{} and their vertex analogs \VSC{} and \SSVE{}, whose approximation ratio is a function of $k$, the number of edges/vertices cut. It turns out that these four variants show very different behaviour in terms of algorithms and hardness of approximation in terms of $k$. We also study fast algorithms for some of these problems and obtain algorithms that run in almost-linear time.

\subsection{Our results}

\paragraph{Edge cuts.}We begin by considering approximation algorithms for \SC{} and \SSE{}. For the sake of clarity we define both these problems formally.

\defparproblem{\SC}{Graph $G$}{$\phi$}{Find a set of vertices $S'$ with $|S'| \leq \frac{n}{2}$  such that $|\delta_G(S')| \leq \phi|S'|$, or report that no such  set exists.}
\defparproblem{\SSE}{Graph $G$}{$\phi$,$s$}{Find a set of vertices $S'$ with $|S'| \leq s$ and $|\delta_G(S')| \leq \phi|S'|$, or report that no such  set exists.}

\begin{definition}[$(\alpha,\beta)$-approximation for \SSE]
We say that an algorithm is an $(\alpha, \beta)$-approximation for $\SSE$ if given parameters $\phi,s$, it outputs a set $S' \subseteq V(G)$ with $|S'| \leq \beta s$ and $|\delta_G(S')| \leq \alpha \phi |S'|$, or reports that no $S \subseteq V(G)$ satisfies $|S| \leq s$ and $|\delta_G(S)| \leq \phi |S|$. An $(\OO(\alpha), \OO(1))$-approximation will be simply called an $\OO(\alpha)$-approximation. 
\end{definition}

We show the first $\OO(\log k)$-approximation for \SSE. This improves upon the $\OO(\log n)$-approximation using~\cite{racke08} and answers the question of~\cite{wahlstrom20}, directly improving the result (discussed in detail in~\Cref{sec:mmn}).

\begin{theorem}\label{thm:logksse}
Let $k$ be the smallest cut-size $|\delta_G(S)|$ among all sets $S$ satisfying $|S| \leq s$ and $|\delta_G(S)| \leq \phi|S|$. Then $\SSE$ admits an $(\OO(\log k), \OO(1))$-approximation in polynomial time.
\end{theorem}

Note that in particular this implies an $\OO(\log k)$-approximation for \SC. Our next result gives an almost-linear algorithm for \SC{} using a new cut-matching game approach. We show an $\OO(\log^2 k)$-approximation algorithm for \SC{} which runs in time $m^{1 + o(1)}$ using a new cut-matching game approach to construct \emph{small set expanders}, that we discuss in~\Cref{sec:cutmatching}. These techniques may be of independent interest. 

\begin{theorem}\label{thm:cutmatchingedgelogk}
Let $k$ be the smallest cut-size $|\delta_G(S)|$ among all sets $S$ satisfying $|\delta_G(S)| \leq \phi|S|$. Then there is an $\OO(\log^2 k)$-approximation algorithm for \SC{} which runs in time $m^{1 + o(1)}$, where $m$ denotes the number of edges in the graph.
\end{theorem}

We leave as an open problem the question of if we can obtain an $\OO(\polylog(k))$-approximation for \SSE{} in almost-linear time: we remark that even in terms of $n$, to the best of our knowledge, the best known approximation ratio for \SSE{} in almost-linear time is $\OO(\log^4 n)$ by using dynamic programming on tree cut-sparsifiers obtained using~\cite{rst14}.

\paragraph{Vertex cuts.} 
We consider approximation algorithms for \VSC{} and \SSVE{}. Again we define these problems formally. Given a graph $G$, we denote a vertex cut by a triplet $(L,C,R)$ where they form a partition of $V(G)$, and after removing $C$ there are no edges between $L$ and $R$. The vertex sparsity of $(L, C, R)$ is defined to be $\frac{|C|}{|C| + \min(|L|, |R|)}$. 
\defparproblem{\VSC}{Graph $G$}{$\phi$}{Find a vertex cut $(L,C,R)$  with $|R| \geq |L|$ and $|C| \leq \phi|L \cup C|$, or report that no such cut exists.}
\defparproblem{\SSVE}{Graph $G$}{$\phi$,$s$}{Find a vertex cut $(L,C,R)$  with $|R| \geq |L|$, $|C| \leq \phi|L \cup C|$ and $|L| \leq s$ or report that no such cut exists.}

Like in the case of edge cuts, we investigate approximation in terms of $k$, where $k$ is the number of cut vertices in the optimal cut. We begin by observing that \SSVE{} is at least as hard to approximate as a well-known problem of \DKS: 
Given a graph and an integer $k \in \mathbb{N}$, \DKS{} asks to find a subset of $k$ vertices that induce the most number of edges.
Relationships between \DKS{} and problems related to \SSVE{} (e.g., \textsc{Minimum $k$-Union}, \textsc{Bipartite Small Set Vertex Expansion}) were discussed in previous works~\cite{hajiaghayi2006prize,
chuzhoy2015approximation, 
chlamtavc2017minimizing}.

\begin{theorem}\label{thm:DKS}
Suppose \DKS{} is $\alpha$ hard to approximate under some assumption $A$. Then there is no bi-criteria $(\OO(\sqrt{\alpha}), \OO(1))$-approximation for \SSVE{} under $A$.
\end{theorem}

Since \DKS{} is believed to be hard to approximate within a factor $n^{1/4-\eps}$ for any constant $\eps > 0$~\cite{bhaskara2010detecting, jones2023sum}, one already has a strong contrast between \SSE{} and \SSVE{} in terms of $n$. 

Surprisingly, a much stronger hardness can be proved in terms of $k$: we show that there is no $2^{o(k)}$-approximation even when we allow a parameterized running time of $f(k)\poly(n)$, assuming the Strongish Planted Clique Hypothesis (SPCH)~\cite{mrs20}. This is in contrast to the edge version \SSE, where we obtained an $\OO(\log k)$-approximation.

\begin{theorem}\label{thm:DSH_SSVE}
Assuming the SPCH, for $\alpha(k) = 2^{o(k)}$ and any functions $f(k), \beta(k)$, there is no $f(k) \cdot \poly(n)$-time $(\alpha(k), \beta(k))$-approximation algorithm for \SSVE{}. 
\end{theorem}

The $2^k$-approximation is tight; 
using the treewidth reduction technique of Marx et al.~\cite{msr13}, we give a matching upper-bound that gives a $2^k$-approximation in time $2^{2^{\OO(k^2)}} \poly(n)$, under a mild assumption which we call maximality. A vertex cut $(L,C,R)$ is \emph{maximal} if for all partitions of the connected components of $G \setminus C$ into $L',R'$ with $|L'| \leq |R'|$, we have $|L| \geq |L'|$. That is, $L$ and $R$ are chosen to obtain the sparsest cut possible after removing $C$.

\begin{theorem}\label{thm:fpt}
Suppose that there is a maximal vertex cut $(L,C,R)$ with $|C| = k$ that separates $s = |L| \leq |R|$ vertices. Then we can find a cut $(L', C', R')$, with $|C'| \leq k$ and $|L| \geq |L'| \geq \frac{|L|}{2^k}$ in time $2^{2^{\OO(k^2)}} n^{\OO(1)}$. In particular, there is a $(2^k, 1)$-approximation algorithm for \SSVE{} in time $2^{2^{\OO(k^2)}} n^{\OO(1)}$. 
\end{theorem}

We now turn our attention to \VSC, where we do not have the additional small set requirement.
First, we note that~\Cref{thm:fpt} also gives a $g(k)$-approximation for \VSC{} in (non-parameterized) polynomial time for some function $g$. To see this, consider the following algorithm. Whenever $k \geq \frac{\sqrt{\log \log n}}{C}$, for some large enough constant $C$, we use a standard poly-time $\OO(\sqrt{\log n})$-approximation (see for example~\cite{fhl05}), and otherwise we use the above algorithm. It follows that this is a polynomial time algorithm with approximation ratio $g(k) = 2^{{\OO(k^2)}}$.

The natural question is then if one can beat $2^{\poly(k)}$-approximation for \VSC. We (almost) show that this is indeed the case: we obtain an $\OO(\log k + \log \log n\phi)$-approximation in polynomial time, if the optimal cut has sparsity $\phi$ and cuts $k$ vertices.
The additional $\log \log n\phi$ term as compared to the edge version is due to the weaker guarantee on our technique of sample sets in the vertex version (see~\Cref{sec:sample_sets}). We show that at least by using sample sets, we cannot hope to improve this result (\Cref{subsec:sample_set_lower_bound}).

\begin{theorem}\label{thm:logkvertex}
Let $k$ be the minimum possible value of $|C|$ among all vertex cuts $(L,C,R)$ that are $\phi$-sparse. Then $\VSC$ admits an $\OO(\log k + \log \log n\phi)$-approximation in polynomial time.
\end{theorem}

Again, we give a fast algorithm using our cut-matching game: we obtain an $\OO(\log^2 k + \log^2 \log n\phi)$-approximation in time $m^{1 + o(1)}$.

\begin{theorem}\label{thm:cutmatchingvertexlogk}
Let $k$ be the minimum possible value of $|C|$ among all vertex cuts $(L,C,R)$ that are $\phi$-sparse. Then $\VSC$ admits an $\OO(\log^2 k + \log^2 \log n\phi)$-approximation in time $m^{1 + o(1)}$.
\end{theorem}

We note that the only regime in which these results are not necessarily an $\OO(\polylog(k))$-approximation is when $k \leq (\log n)^{o(1)}$.
\Cref{table:results} summarizes our results.

\begin{table}[!htbp]
\begin{center}
\begin{tabular}{|c|c|c|c|}
\hline
\textbf{Edge/Vertex} & \textbf{Approximation ratio}                                        & \textbf{Small set} & \textbf{Time} \\ \hline
Edge                 & $\OO(\log k)$                                           & yes                & $\poly(n)$          \\ \hline
Edge                 & $\OO(\log^2 k)$                                         & no                 & almost-linear \\ \hline
Vertex               & $\OO(\log k + \log\log n\phi)$                           & no                 & $\poly(n)$          \\ \hline
Vertex               & $\OO(\log^2 k + \log^2 \log n\phi)$                     & no                 & almost-linear \\ \hline
Vertex               & No $n^\epsilon$-approximation  for some $\epsilon > 0$ & yes                & $\poly(n)$          \\ \hline
Vertex               & No $2^{o(k)}$-approximation                            & yes                & $f(k)\poly(n)$           \\ \hline
Vertex               & $2^k$                                                 &  yes\tablefootnote{Under maximality.}                & $f(k)\poly(n)$          \\ \hline
\end{tabular}
\end{center}
\caption{Summary of our results}
\label{table:results}
\end{table}

\paragraph{Applications.} Next, we describe two applications of our result and techniques for \SSE{} to multicut mimicking networks and \MMGP. 

\paragraph{Multicut mimicking networks.}
We show an improvement in the size of the best multicut mimicking network that can be computed in polynomial time. This follows directly from the reduction in~\cite{wahlstrom22} to \SSE{} and our $\OO(\log k)$-approximation for it. We start by defining multicut mimicking networks.

\begin{definition}
Given a graph $G$ and a set of terminals $T \subseteq V(G)$, let $G'$ be another graph with $T \subseteq V(G')$. We say that $G'$ is a multicut mimicking network for $G$ if for every partition $\mathcal{T} = T_1 \cup T_2 \cup \ldots \cup T_s$  of $T$, the size of a minimum multiway cut for $\mathcal{T}$ is the same in $G$ and $G'$. 
\end{definition}

We show the following theorem using our $\OO(\log k)$-approximation for \SSE. Given a graph $G$ and a set of terminals $T \subseteq V(G)$, let $cap_G(T)$ denote the total degree of the vertices in $T$. 
\begin{theorem}\label{thm:mmn}
Given a terminal network $(G,T)$ with $cap_G(T) = k$, one can find in randomized polynomial time a multicut mimicking network of size $k^{\OO(\log^2 k)}$.
\end{theorem}

This improves the result in~\cite{wahlstrom22} which shows a bound of $k^{\OO(\log^3 k)}$. This in turn improves the best known kernels for many problems in polynomial time, which are based on multicut mimicking networks as considered in~\cite{wahlstrom22}. In particular, {{\sc Multiway Cut}}{} parameterized by the cut-size and {{\sc Multicut}}{} parameterized by the number of pairs and the cut-size together, admit $k^{\OO(\log^2 k)}$ size kernels in polynomial time.

\paragraph{\MMGP:} 

In \MMGP, we are given a graph $G$ and an integer $r$, and the goal is to partition the vertex set into $r$ parts of size $\frac{n}{r}$ each, while minimizing the maximum cut-size among all parts. 
Bansal et al.~\cite{minmax} obtain an $(\OO(\sqrt{\log n \log r}), \OO(1))$-approximation algorithm.
Extending our notion of sample sets for sparse cuts to a weighted setting, combined with LP-rounding techniques from~\cite{hc16} and~\cite{minmax} we obtain the following theorem.

\begin{theorem}
Min-max graph partitioning admits an $(\OO(\log opt), \OO(1))$-approximation algorithm, where $opt$ is the optimal min-max value.
\end{theorem}

\subsection{Technical highlight}
We use two main techniques to obtain our results. The first is that of sample sets for sparse cuts discussed in \Cref{sec:sample_sets}. Feige and Mahdian~\cite{fm06} introduced the concept of sample sets, which are a small set of terminals which represent every set with a (vertex or edge) cut of size at most $k$ in proportion to its size, upto an additive deviation of $\epsilon n$, where $n$ is the number of vertices in the graph. They used it to obtain a $2^{\OO(k)}\poly(n)$ time algorithm for finding balanced separators of size $k$. We extend this concept to sparse cuts, and show that there is a small set of terminals that represents all sparse cuts.  We show essentially tight bounds for such sample sets, and then combine this with two other techniques --- (a) existing LP-based methods, and (b) our new cut-matching game --- to obtain our approximation algorithms. Essentially, sample sets for sparse cuts enable us to convert $\OO(\log s)$-approximation algorithms to $\OO(\log k)$-approximation.

The second main technique we develop is the cut-matching game for obtaining small set expanders, introduced in~\Cref{sec:cutmatching}. The cut-matching game~\cite{krvactual,krv07,osvv08} is a very flexible framework that yields one of the fastest approximations to \SC, and hence is used extensively in designing fast graph algorithms. The cut-matching game approach of~\cite{krv07} shows that there is a cut player strategy so that for any matching player, we can construct an expander in $\OO(\log n)$ rounds. We show that there is a cut player strategy to construct $s$-small set expanders in $\OO(\log s)$ rounds: these are graphs in which every set with size at most $s$ is expanding. We show how to achieve this in time $\tilde{\mathcal{O}}(m)$, where $m$ is the number of edges. Together with the almost-linear time max-flow algorithm of~\cite{chen22} which we use for the matching player strategy, this forms the basis of all our almost-linear time algorithms. 

\subsection{Related work and connections to FPT algorithms}

\paragraph{Vertex cuts.}
Apart from a number of aforementioned works on approximating the edge version of \SSE{} and \SC, there is a large literature on approximating \SSVE{} and \VSC. For \VSC, Feige et al.~\cite{fhl05} used the ARV algorithm to obtain an $\OO(\sqrt{\log n})$-approximation. As we discussed, \SSVE{} is much harder to approximate. Louis and Makarychev~\cite{louis2014approximation} obtained  $\tilde{\OO}(\frac{n}{s} \sqrt{\log n})$ and  
$\tilde{\OO}( 
\frac{n}{s} (\sqrt{\psi \log d} + \psi))$ approximations where $\psi$ is the vertex conductance of the optimal cut and $d$ is the maximum degree. 

\paragraph{Approximation algorithms with respect to other parameters.} Approximation algorithms where the approximation ratio is a function of a parameter rather than the input size have been studied before for various graph partitioning problems. s
For instance, Garg et al. showed an $\OO(\log t)$-approximation algorithm for \textsc{Multicut}

where $t$ is the number of terminal pairs that need to be separated~\cite{gvy96}. Even et al.~\cite{ess98} studied the approximation of  \textsc{Feedback Vertex/Arc Sets} in directed graphs, and obtained $\OO(\log opt \log \log opt)$-approximations, where $opt$ is the size of the optimal feedback set. Even et al.~\cite{enrs99} obtained $\OO(\log opt)$-approximations for many graph partitioning problems using spreading metrics, where $opt$ is the size of the optimal cut. Calinescu et al.~\cite{ckr05} obtained an $\OO(\log k)$-approximation for the \textsc{$0$-Extension} problem with $k$ terminals. Lee~\cite{lee17} obtained an $\OO(\log k)$-approximation for \textsc{$k$-Vertex Separator} and other related problems.
As previously mentioned, Feige et al.~\cite{fhl05} gave an $\OO(\sqrt{\log k})$-approximation for \VSC{} when the optimal solution is balanced, which was an important ingredient in approximating treewidth in polynomial time. 

There are previous results for \SSE{} whose approximation guarantee depends on $s$, the size of the set that we seek to find. Bansal et al.~\cite{minmax} gave an approximation algorithm for \SSE{} in terms of both $n$ and $s$: their algorithm has an approximation ratio of $\OO(\sqrt{\log n \log {\frac{n}{s}}})$.

Hasan and Chwa~\cite{hc16} obtained an $\OO(\log s)$-approximation for \SSE{}. 
Finally, the conductance of the optimal cut $\psi \in (0, 1)$ is another important parameter related to the performance of many algorithms for \SC{} and \SSE; the well-known Cheeger inequality gives an $\OO(\sqrt{\psi^{-1}})$-approximation algorithm. For \SSE, Raghavendra et al. \cite{raghavendra2010approximations} gave an $\OO(\sqrt{ \psi^{-1} \log (n/s) })$-approximation.

\paragraph{Connection to FPT algorithms.}Cygan et al.~\cite{cygan20} showed that \textsc{Minimum Bisection} admits an FPT algorithm that runs in time $2^{\OO(k\log k)}\poly(n)$. In fact, given a size parameter $s < \frac{n}{2}$, their algorithm can find a cut with minimum number of cut-edges among all cuts $(X,\overline{X})$ with $|X| = s$. This essentially means we can solve \SC{} and \SSE{} exactly when $k \leq \frac{\log n}{C\log \log n}$ for some constant $C$. But when $k \geq \frac{\log n}{C \log \log n}$, $\SC$ admits an $\OO(\sqrt{\log n}) = \OO(\sqrt{k \log k})$-approximation, and \SSE{} admits an $\OO(\log n) = \OO(k \log k)$-approximation. Thus FPT algorithms imply some $f(k)$-approximations for the edge versions in polynomial time, but our algorithms significantly improve upon these guarantees.

\section{Overview}
In this section, we describe on a high level our ideas for obtaining approximation algorithms in terms of $k$. For simplicity, we will focus on edge cuts, and only give an overview of two results: how we obtain an $\OO(\log k)$-approximation for \SSE{}, and an $\OO(\log^2 k)$-approximation for \SC{} in almost-linear time. 

\subsection{Approximating \SSE.} In \SSE, given that there exists a set $S$ with $|S| = s$ and $|\delta_G(S)| = k$ (so that $S$ is $\phi = \frac{k}{s}$-sparse), we wish to find a set which has size at most $\OO(s)$ and has sparsity $\phi' = \OO(\phi \log k)$.

First, we may assume $s \gg k \log k$, otherwise, the $\OO(\log s)$-approximation of~\cite{hc16} is already an $\OO(\log k)$-approximation. 

Is there a way to still use (an extension of) an $\OO(\log s)$-approximation algorithm even when $s \gg k \log k$?
Our basic approach is to carefully choose a {\em sample set} $T \subseteq V$ that sparsifies the vertex set in the following sense. 
For $S \subseteq V$, define the {\em terminal sparsity} of $S$ as $\phi_T(S) := \frac{|\delta_G(S)|}{|S \cap T|}$; intuitively, we pretend that the size of the set $S$ is $|S \cap T|$ instead of $|S|$.
Ideally, 
\begin{enumerate}[(i)]
    \item if every $S' \subseteq V$ satisfies $|S' \cap T| \approx |S'| \cdot \frac{|T|}{n}$ (which implies $\phi_T(S') \approx \phi(S') \cdot \frac{n}{|T|}$), and
    \item the $O(\log s)$-approximation algorithm for $\phi(.)$ can be strengthened to an $O(\log t)$-approximation for $\phi_T(.)$ with $t = |S \cap T|$,
\end{enumerate}
one can expect to find a set $S'$ with $|S' \cap T| \leq \OO(t)$ and $\phi_T(S') \leq \OO(\log t \cdot \phi_T(S))$, which implies that $|S'| \leq \OO(s)$ and $\phi(S') \leq \OO(\log t \cdot \phi(S))$. If we could set $t$ to be sufficiently small so that $O(\log t) = O(\log k)$, we will achieve our goal.

Of course, (i) is impossible to have for every $S \subseteq V$ (e.g., when $S \subseteq T$), but one can hope to satisfy it for every sparse cut. The following definition requires that every $S$ sparse in either $\phi(S)$ or $\phi_T(S)$ satisfies $|T \cap S| \approx |S|\cdot\frac{|T|}{n}$; in other words, any set $S$ sparse with respect to $\phi$ should be sparse with respect to $\phi_T$, and vice versa.

\begin{definition}[Simplification of~\Cref{def:sample_set}]
Given a graph $G = (V(G), E(G))$ and parameters $\epsilon, \phi'$, an $(\epsilon, \phi')$ edge sample set for $G$ is a non-empty set of terminals $T \subseteq V(G)$ which has the following property. Consider a set of vertices $W \subseteq V(G)$
which satisfies either (a) $\phi(W) \leq \phi'$ or (b) $\phi_T(W) \leq \frac{n}{|T|}\frac{\phi'}{10}$. Then we have $|\frac{n}{|T|}|W \cap T| - |W|| \leq \epsilon|W|$.
\end{definition}

This definition is a refined notion of the previously defined sample sets (e.g.,~\cite{fm06}), which allows an additive $\epsilon n$ error and preserves only balanced cuts that have $\Omega(n)$ vertices on each side. Based on Steiner decomposition, we prove in Lemma~\ref{lemma:sample_set} that there exists a sample set of size $\min\{n, \Theta(\frac{n \phi'}{\eps^2})\}$. Compute such a sample set $T$ with $\phi' = \OO(\phi \log k) = \OO(\frac{k \log k}{s})$ as our target sparsity and $\eps = 1/10$. Since we assumed $s \gg k \log k$, we must have that $T$ is of size $\Theta(\frac{n \phi'}{\eps^2}) < n$.

Since our target set $S$ is $\phi$-sparse and therefore $\phi'$-sparse, we must have $t := |S \cap T| = \Theta\left(\frac{|T|}{n}|S|\right) =  \Theta(k\log k)$, so that the terminal sparsity of $S$ is $\OO(\frac{1}{\log k})$. 
If the condition (ii) above is true, we are able to find a set $S'$ with $|S' \cap T| = \OO(t)$ and $\phi_T(S') = \OO(\log k \cdot \frac{1}{\log k}) = \OO(1)$. Notice that $\frac{n}{10|T|} \phi' = \Theta(1)$.  By carefully choosing constants one can ensure that $S'$ is $\frac{n}{10|T|}\phi'$-terminal sparse, and hence must satisfy the sample condition above. It implies that $|S'| = \Theta(\frac{n}{|T|}|S' \cap T|) = \OO(s)$. Since $S'$ is $\OO(1)$-terminal sparse, we have 
$\phi(S') = \frac{|\delta_G(S')|}{|S'|} = \OO\left(\frac{|T|}{n}\frac{|\delta_G(S')|}{|S' \cap T|}\right) =  \OO \left(\frac{\phi'|\delta_G(S')|}{|S' \cap T|}\right) = \OO(\phi') = \OO(\phi \log k)$, and it follows that we obtain an $\OO(\log k)$-approximation for the sparsity. 

Thus, it suffices to achieve (ii) above, which can be considered as a  ``terminal version'' of \SSE: Given a set $S$ with $|\delta_G(S)| = k$, $|S \cap T| = t = \Theta(k \log k)$ such that $S$ is $\psi = \OO(\frac{1}{\log k})$-terminal sparse, can we find a set with at most $\OO(t)$ terminals that is $\psi' = \OO(\psi \log t)$-terminal sparse?
Note that this is a generalization of the original \SSE{} problem when $T = V(G)$. 

Hasan and Chwa~\cite{hc16} handled this special case by giving an $\OO(\log s)$-approximation by extending the techniques of~\cite{minmax}. We show that the techniques from~\cite{hc16} and ~\cite{minmax} extend to the more general version with terminals, and obtain an $\OO(\log t)$-approximation, giving an $\OO(\log k)$-approximation for \SSE.

\subsection{Cut-matching game for small set expanders.} 
Our $\OO(\log^2 k)$-approximation algorithm for \SC{} (Theorem~\ref{thm:cutmatchingedgelogk}) is based on our new cut-matching game for constructing small set expanders. Given a graph $G$, the cut-matching game frameworks of~\cite{krvactual,krv07} start with an empty graph $H$ with the same vertex set. Let $\phi$ be a sparsity parameter. In every round, each player takes the following action. 

\begin{itemize}
    \item The cut player finds a bisection $(B, \overline{B})$ and gives it to the matching player. 

    \item The matching player either finds a $\phi$-sparse cut in $G$ (so we are done), or finds a perfect matching $M$ between $B$ and $\overline{B}$ so that $\phi \cdot M \preceq^{\text{flow}} G$, which means that $G$ (with unit capacities) admits a flow simultaneously sending amount $\phi$ between the endpoints of every $e \in M$. The matching $M$ is added to $H$. 
\end{itemize}

Khandekar, Khot, Orecchia and Vishnoi~\cite{krv07} proved that there is a cut player strategy that, against any matching player always returning a matching, ensures after $\OO(\log n)$ rounds that $H$ is an $\Omega(1)$-expander.\footnote{A graph $H$ is ($\Omega(1)$)-expanding when every $S \subseteq V$ with $|S| \leq n/2$ has $\phi(S) \geq \Omega(1)$. Note that this is with respect to sparsity, not conductance.}
This implies that $G$ is an $\Omega(\frac{\phi}{\log n})$-expander as $\phi H \preceq^{\text{flow}} \OO(\log n)G$. Therefore, by executing the strategy of both players, one can either obtain a $\phi$-sparse cut in $G$ or certify that $G$ is a $\Omega(\frac{\phi}{\log n})$-expander. This gives an $\OO(\log n)$-approximation for \SC. However, we note that the cut player strategy of~\cite{krv07}  is not a poly-time strategy, hence this does not directly give us a polynomial time approximation algorithm.

\paragraph{Cut player strategy for small set expanders.}
We show that in a similar framework, there is a cut player strategy that, against any matching player always returning a matching, can ensure in $\OO(\log s)$ rounds that $H$ is a $s$-small set $\Omega(\frac{1}{\log s})$-expander; every $S \subseteq V$ with $|S| \leq s$ has $\phi_H(S) \geq \Omega(\frac{1}{\log s})$. Furthermore, we implement the strategy in near-linear time. Given exponential time, we can even ensure that $H$ is a $s$-small set $\Omega(1)$-expander, thus generalizing the result of~\cite{krv07}.

The basic idea is the following. Let $S$ be a set of at most $s$ vertices in $H$. If we can somehow force the cut player, within $\OO(\log s)$ rounds, to output a bisection $(X, \overline{X})$ which is ``sufficiently unbalanced'' with respect to $S$, we would be done. Concretely, suppose that $||X \cap S| - |\overline{X} \cap S|| \geq \epsilon |S|$ for some constant $\epsilon$, then in this round, at least $\epsilon|S|$ edges must be added across the cut $(S, \overline{S})$ in $H$ via the matching between $X$ and $Y$, and hence $S$ becomes expanding.\footnote{A subset $S \subseteq V$ is ($\Omega(1)$)-expanding when $\phi(S) \geq \Omega(1)$. Note that this is with respect to sparsity, not conductance.}

This motivates our approach. First, we start by finding a $\Omega(1)$ balanced cut $(W, \overline{W})$ in $H$ that is $\frac{1}{C}$-sparse for some large constant $C \gg 1$. If no such cut exists, we show that we can be done using just one more round --- we skip this detail in this overview. Next, we use our key subroutine called \IC, which given a cut $(W, \overline{W})$ in $H$ that is $\frac{1}{C}$-sparse, uses a single max-flow call to obtain another balanced cut $(Q, \overline{Q})$ such that for every subset $S \subseteq V(H)$ that is $\frac{1}{K}$-sparse for some constant $K \gg C$, we have $|\delta_{H[S]}(Q \cap S)| \leq \frac{|S|}{D}$, for some large constant $D$. In other words, the number of edges cut \emph{inside} every significantly sparse set $S$ is at most $\frac{|S|}{D}$.  

The cut player then simply outputs the cut $(X = Q, \overline{X} = \overline{Q})$. (Let us assume that $(Q, \overline{Q})$ is an exact bisection in this overview; the actual framework is slightly more involved.)

Fix a set $S$ of size at most $s$. After some number of rounds, if $S$ is already $\frac{1}{K}$-expanding, we are already done and hence we assume that this is not the case. Now for the cut $(Q, \overline{Q})$ found by the cut player in the current round, if it turns out that $||S \cap Q| - |S \cap \overline{Q}|| \geq \Omega(|S|)$, the matching player is forced to add $\Omega(|S|)$ edges across $(S, V(H) \setminus S)$ in this round, and hence $S$ becomes expanding. Now suppose this does not happen. Hence we may now assume that the cut $(Q, \overline{Q})$ is balanced with respect to $S$.

Also \IC{} guarantees that $|\delta_{H[S]}(Q \cap S)| \leq \frac{|S|}{D}$. Since, $|Q \cap S| = |\overline{Q} \cap S| = \Omega(|S|)$, the sparsity of the cut $Q \cap S$ \emph{inside} $H[S]$ must be $\OO(\frac{1}{D})$. 
Thus the cut $(Q, \overline{Q})$ cuts the set $S$ in a balanced and sparse way. If $S$ was the whole graph, this is exactly what the cut player in previous cut-matching games does; finding a sparse balanced cut and giving it to the matching player!

We show that if this happens for more than $\OO(\log s)$ rounds, then there must exist a round in which the matching player must add $\Omega(|S|)$ edges across the cut $(S, V(H) \setminus S)$ (see~\Cref{lemma:potential}). This analysis is a local version of the potential analysis of~\cite{krv07}. It follows that after $\OO(\log s)$ rounds, there must be some round in which the matching player added $\Omega(|S|)$ edges across the cut $(S, V(H) \setminus S)$ (or find a sparse cut in $G$), and thus we obtain our result.

We remark that the key idea behind the \IC{} subroutine has been used in previous work on related problems. Lang and Rao~\cite{lang2004flow} and Andersen and Lang~\cite{andersen2008algorithm} gave algorithms to improve quotient cuts using maximum flows. Similar ideas were used in local flow algorithms~\cite{orecchia2014flow,henzinger2020local}. Saranurak and Wang~\cite{sw19} used a similar algorithm for expander pruning.

\paragraph{Near-linear time implementation:} We describe the main technical ingredients in making the above algorithm run in time $\tilde{\mathcal{O}}(m)$. First is the task of finding an $\Omega(1)$ balanced cut in $H$ that is $\frac{1}{C}$-sparse, or to certify that there is no such cut. While we cannot hope to solve this problem exactly in polynomial time, it is possible to check if there is an $\Omega(1)$ balanced cut in $H$ that is $\frac{1}{C}$-sparse, or certify that every balanced cut in $H$ is $\Omega(\frac{1}{\log s})$-expanding. But even though this looks like an $\OO(\log s)$-approximation for \SC{} by itself, we can indeed solve this problem as follows. Since our cut-matching game only runs for $\OO(\log s)$ rounds, the maximum degree $\Delta$ in $H$ is only $\OO(\log s)$. We further ensure that the graph $H$ is regular. Then it is enough to find a balanced cut $(W, \overline{W})$ with \emph{conductance} $\frac{|\delta_G(W)|}{\vol(W)} \leq \frac{1}{C\Delta}$, or certify that every balanced cut has conductance at least $\Omega(\frac{1}{C^2\Delta^2})$.

But this is exactly a ``Cheeger-type'' approximation for balanced low conductance cuts. We use the algorithm of~\cite{ov11} to accomplish this: while their result is slightly different from the notion described above, it will be enough for us to obtain our result (see \Cref{thm:bal_sep}).

Once we have the cut player strategy for small set expanders, it is easy to show that one can use a matching player similar to that of the standard cut-matching game~\cite{krvactual} as discussed above and obtain the following result: Given a graph $G$ and a  parameter $\phi'$, either find a $\phi'$-sparse cut, or certify that the graph is a $s$-small set $\Omega(\frac{\phi'}{\log^2 s})$-expander in time $m^{1 + o(1)}$. This gives an $\OO(\log^2 s)$-approximation for \SC. 

It is not difficult to show that this extends to terminal sparsity as well: Given a graph $G$ and a  parameter $\phi'$, either find a $\phi'$-terminal sparse cut, or certify that every set with at most $s$ terminals is $\Omega(\frac{\phi'}{\log^2 s})$-expanding. Armed with this result for the terminal version and the technique of sample sets used similarly as described for \SSE, we can obtain an $\OO(\log^2 k)$-approximation. We show that this entire procedure can be accomplished in time $m^{1 + o(1)}$, and hence we obtain our algorithm for \SC.

\paragraph{Organization of the paper:} We start with preliminaries and definitions in~\Cref{sec:prelim}. In ~\Cref{sec:sample_sets}, we define our new notion of sample sets and give both upper and lower bounds for sample sets for both edge and vertex cuts.~\Cref{sec:sselp} describes how to combine sample sets with LP rounding algorithms to obtain approximation algorithms. In~\Cref{sec:cutmatching} we describe our new cut-matching game for small set expanders, and describes how to use sample sets together with this new cut-matching game to obtain fast parametric approximation algorithms.~\Cref{sec:ssve} proves both our hardness and algorithmic results for \SSVE{}. \Cref{sec:applications} discusses the applications of our results and techniques to \MMGP{} and Multicut Mimicking Networks. In~\Cref{sec:open} we state some open problems that arise naturally from this work.~\Cref{sec:appendix} consists of our generalizations of some existing LP rounding algorithms, which enable us to combine them with the technique of sample sets in~\Cref{sec:sselp}.

\section{Preliminaries and notation}\label{sec:prelim}

We start with basic graph terminology and define the problems which we will consider throughout this paper. All graphs $G$ are connected and undirected unless otherwise stated.

\paragraph{Edge cuts.} We will denote an (edge) cut in a graph $G$ by $(U, \overline{U})$ where $U \subseteq V(G)$ and $\overline{U} = V(G) \setminus U$. For the sake of simplicity, we sometimes also simply refer to this cut as $U$. The set of cut edges is denoted by $\delta_G(U)$. We say that the cut $(U, \overline{U})$ is \emph{$b$-balanced with respect to a set $X \subseteq V(G)$} if min$(|U \cap X|, |\overline{U} \cap X|) \geq b|X|$. A cut is \emph{$b$-balanced} if it is $b$-balanced with respect to $V(G)$. The sparsity of a cut $(U, \overline{U})$ is defined as $\frac{|\delta_{G}(U)|}{\min(|U|, |\overline{U}|)}$. We will denote the sparsity of the cut $(U, \overline{U})$ by $\phi^G(U)$, and omit the superscript when the context is clear. A cut is $\phi'$-sparse if its sparsity is at most $\phi'$, and $\phi'$-expanding otherwise. We say that $G$ is a $\phi'$-expander if every cut is $\phi'$-expanding. We call a cut $(U, \overline{U})$ $\phi'$-sparse with respect to a set $X$ if $|\delta_{G[X]}(U)| \leq \phi'$min$(|U \cap X|, |\overline{U} \cap {X}|)$, where $G[X]$ is the graph induced on the vertex set $X$. 

The \emph{volume} of a set $X$, denoted by $\vol_G(X)$, is the sum of the degrees of the vertices inside $X$. The \emph{conductance} of a cut $(X, \overline{X})$ is defined as $\frac{|\delta_G(X)|}{\min\{\vol_G(X), \vol_G(\overline{X})\}}$.

Given a set of terminals $X \subseteq V(G)$, we define the \emph{terminal-sparsity} of a cut $(U,\overline{U})$ as $\phi^G_X(U) = \frac{|\delta_G(U)|}{\min(|U\; \cap \; X|, |\overline{U}\; \cap \; X|)}$. Notice that this is different from sparsity with respect to $X$, since we count every edge in $\delta_G(U)$ as opposed to only counting the edges $\delta_{G[X]}(U)$. Similarly, we say that $(U,\overline{U})$ is $\phi'$-terminal sparse if $\phi^G_X(U) \leq \phi'$, and $\phi'$-terminal expanding otherwise. We say that $G$ is a $\phi'$-terminal expander if every cut is $\phi'$-terminal expanding. We omit the superscript $G$ when the graph is clear from the context.

\paragraph{Vertex cuts.} We denote a vertex cut by a partition $(L, C, R)$ of $V(G)$ such that there are no edges between $L$ and $R$ after removing $C$. The sparsity of the cut is given as $\phi_G(L,C,R) = \frac{|C|}{\min(|L \;\cup\; C|, |R\; \cup \;C|)}$.

Given a set of terminals $X$, the (vertex) terminal-sparsity of the cut $(L, C, R)$ is defined as $\phi^G_X(L,C,R) = \frac{|C|}{\min(|(L\; \cup\; C)\; \cap\; X|, |(R\; \cup\; C)\; \cap\; X|)}$. Again we omit the superscript when the graph is clear from the context.  For a set $L$, we will denote by $N_G(L)$ the neighbours of the vertices in $L$ which are not in $L$. We will also talk about the vertex sparsity of a set $L$: this will be defined as $\phi^G(L) = \frac{|N_G(L)|}{|L \;\cup\; N_G(L)|}$. Similarly, given a set of terminals we define the vertex terminal sparsity of the set $L$ as $\phi^G_T(L) = \frac{|N_G(L)|}{|(L \;\cup N_G(L))\; \cap T|}$. We say that $L$ is $\phi'$-(vertex) terminal sparse if $\phi^G_T(L) \leq \phi'$, and $\phi'$-(vertex) terminal expanding otherwise. $G$ is a $\phi'$-(vertex) terminal expander, if every set $L \subseteq V(G)$ is $\phi'$-(vertex) terminal expanding. We also remark that for the sake of simplicity, our notation for vertex and edge sparsity is similar, but the difference shall be clear from the context.

\paragraph{Flow embeddings.} Given a flow $f$ on an undirected weighted graph $G$, with capacity $c(e)$ for each edge $e \in E(G)$, the congestion of the flow $f$ is defined as $\max_{e} \frac{f(e)}{c(e)}$. If $G$ is unweighted, we have $c(e) = 1$ for each edge $e$, so that the congestion is $\max_e f(e)$. For graphs $H$ and $G$, we write $H \preceq^{\text{flow}} cG$, or $H$ flow-embeds in $G$ with congestion $c$, if one unit of demand can be routed between the vertices corresponding to every edge of $H$ simultaneously in $G$, with congestion $c$.

\paragraph{Algorithms and runtime.} Given a connected graph $G$ with $n$ vertices and $m$ edges, we say that an algorithm $A$ runs in near-linear time if it runs in time $\tilde{\mathcal{O}}(m)$ where $\tilde{\mathcal{O}}$ hides polylogarithmic factors. We say that $A$ runs in almost-linear time if it runs in time $m^{1 + o(1)}$.

\section{Sample sets for sparse cuts}\label{sec:sample_sets}
In this section, we extend the technique of sample sets in~\cite{fm06}  to the setting of sparse cuts. We begin by defining the notion of sample sets for sparse cuts that we will need. Our notion of sample sets will be a small set of terminals which represent every sparse cut in proportion to its size in the actual graph.

\begin{definition}\label{def:sample_set}
Given a graph $G = (V(G), E(G))$ and parameters $\epsilon, \phi'$, an $(\epsilon, \phi')$ edge (vertex) sample set for $G$ is a non-empty set of terminals $T \subseteq V(G)$ which has the following property. Consider a  set of vertices $W \subseteq V(G)$ which is either (a) a connected component after removing an edge (vertex) cut (b) the union of all connected components except one after removing an edge (vertex) cut. Further suppose that $W$ satisfies either (a) $\phi(W) \leq \phi'$ or (b) $\phi_T(W) \leq \frac{n}{|T|}\frac{\phi'}{10}$. Then we have $|\frac{n}{|T|}|W \cap T| - |W|| \leq \epsilon|W|$.
\end{definition}

We remark that the number $10$ is arbitrary and chosen large enough for clarity and ease of exposition. 

\paragraph{Comparison with sample sets in~\cite{fm06}.}~\cite{fm06} defined an $(\epsilon,k)$ sample: this was a set of terminals $T \subseteq V(G)$, such that for every set $W$ that is either a connected component or the union of all but one connected components after removing an edge/vertex cut, we have $|\frac{n}{|T|}|W \cap T| - |W|| \leq \epsilon n$. They obtained an edge sample set with size $\OO(\frac{k}{\epsilon^2})$ and a vertex sample set of size $\OO(\frac{k}{\epsilon^2} \log \frac{1}{\epsilon})$. However, this sample set is clearly useful only when $W$ itself is large (at least $\epsilon n$). Thus, one can view our notion of sample sets as a strengthening: we preserve all sparse cuts within $\epsilon$ multiplicative deviation, instead of just preserving large sets with a small cut-size.   

\paragraph{Sample sets using random sampling.}Before we proceed further, we note that one can easily obtain a sample set of size $(\min\{\Theta(\frac{n\phi'}{\epsilon^2} \log n), n\})$ with high probability using random sampling for both edge and vertex cuts. We describe this on a high level for edge cuts. Pick a random subset $T \subseteq V(G)$ of $\frac{Cn\phi'}{\epsilon^2}\log n$ terminals. If $\frac{Cn\phi'}{\epsilon^2}\log n > n$, we can simply return the whole vertex set. Consider a $\phi'$-sparse set $W$ with $\delta_G(W) = k'$. The expected number of terminals from $W$ is then $\frac{C\phi'|W|}{\epsilon^2} \log n \geq \frac{Ck'\log n}{\epsilon^2}$, where the inequality follows from the fact that $W$ is $\phi'$-sparse. By a Chernoff bound, the number of terminals is within an $\epsilon$ multiplicative error with probability $e^{-Ck'\log n}$. However, there are only $n^{\OO(k')}$ sets to consider, and applying a union bound on them and on all possible values of $k'$, every such $W$ must satisfy the sample condition with probability at least $1 - \frac{1}{n}$. Similarly, we can show that every set $W$ with $\phi_T(W) \leq \frac{n}{|T|}\frac{\phi'}{10}$ must satisfy the sample condition with probability at least $1 - \frac{1}{n}$. Applying a union bound then gives us our result.

However, we show that one can significantly improve these results, both for edge and vertex cuts, using techniques from~\cite{fm06}.
The rest of this section is organized into three subsections. In the first, we show an upper bound for the size of sample sets for edge cuts. In the second, we show an upper bound for vertex cuts. Finally, we show a matching lower bound for vertex cuts using a connection to Ramsey numbers.

\subsection{Edge cuts -- upper bound}
The next result is our improved upper bound for edge cuts.
\begin{lemma}\label{lemma:sample_set}
For every $\epsilon \in (0,1)$, and any sparsity parameter $\phi'$ there is an $(\epsilon,\phi')$ edge sample set for every graph $G$ of size $\min\{\Theta(n\frac{\phi'}{\epsilon^2}),n\}$. Further, such a set can be computed by a deterministic algorithm which runs in near-linear time.
\end{lemma}
In fact, our result will be stronger than the notion defined in~\Cref{def:sample_set} in that it will hold for any $\phi'$-sparse set $W$, not just a connected component or the union of all connected components except one after removing an edge cut.

\begin{proof}
Our proof is similar to the construction of the deterministic sample set in~\cite{fm06} for edge separators. We will use the notion of Steiner-$t$-decomposition introduced in~\cite{fm06}.
\begin{lemma}[Lemma 5.2,~\cite{fm06}]\label{lemma:feige}
Given any graph $G$ and a value $1 \leq t \leq n$, one can compute in near-linear time a partition $V_1, V_2 \ldots V_{\ell}$ of $V(G)$ and a partition of the edge set $E_1, E_2 \ldots E_{\ell}, E'$ such that
\begin{enumerate}
    \item $G[E_i]$ is connected for all $i \in [\ell]$.
    \item $V(G[E_i]) = V_i$ or $V(G[E_i]) =  V_i \cup \{u\}$ for some vertex $u$ for each $i \in [\ell]$.
    \item $|V_i| \leq 2t$ for all $i \in [\ell]$.
    \item $|V_i| \geq t$ for all $i > 1$.
\end{enumerate}
\end{lemma}

We remark here that Lemma 5.2 of~\cite{fm06} only claims a polynomial run-time, but it is quite easy to make this work in time $\Tilde{\OO}(m)$. For a formal proof, we refer to the proof of Claim 3.3 from Fomin et al.~\cite{flspw18}, which is a generalization of~\Cref{lemma:feige}.

We will refer to each vertex set $V_i$  a \emph{bag} in the decomposition.
Compute a Steiner-$t$-decomposition with $t = \frac{\epsilon}{100\phi'}$.
Clearly, $\frac{100n\phi'}{\epsilon} \geq \ell \geq \frac{50n\phi'}{\epsilon}$. We now construct a sample set $T$ as follows: arbitrarily pick $\lfloor \frac{|V_i|}{t\epsilon} \rfloor$ vertices from the bag $V_i$ for each $i \in \{2,3 \ldots l\}$. Note that we may assume that $V_i$ has at least $\lfloor \frac{|V_i|}{t\epsilon}| \rfloor$ vertices, for if not, we must have $\phi' \geq \frac{\epsilon^2}{200}$, and in that case, we can simply return the whole vertex set as our sample set.

\begin{lemma}
The set $T$ is a $(4\epsilon, \phi')$ sample set for all $\epsilon \in (0,\frac{1}{100})$. 
\end{lemma}

\begin{proof}

 We first show that every bag is represented proportional to its size in the terminal set.
 Observe that from every bag $V_i$, $i > 1$, we pick $\frac{|V_i|}{t\epsilon} \geq \lfloor\frac{|V_i|}{t\epsilon}\rfloor \geq \frac{|V_i|}{t\epsilon} - 1$ vertices. Since $t \leq |V_i|$ for all $i > 1$, this means that $|V_i \cap T| \in [\frac{|V_i|}{t\epsilon}-\epsilon\frac{|V_i|}{t\epsilon}, \frac{|V_i|}{t\epsilon}]$ for all $i > 1$.  Adding, and noting that $|V_1| \leq 2t$ we get $|T| \in [\frac{(n - 2t)}{t\epsilon}(1 - \epsilon), \frac{n}{t\epsilon}]$. This means that $\frac{n}{|T|}|V_i \cap T| \in [|V_i|{(1 - \epsilon)}, \frac{|V_i|}{1 - \epsilon}\frac{n}{n-2t}]$ for all $i > 1$. We know that $t = \frac{\epsilon}{100 \phi'} \leq \frac{\epsilon}{4} n$, which in turn gives ${n-2t} \geq (1 - \frac{\epsilon}{2})n$, and hence we have $|T| \in [(1 - \frac{\epsilon}{2})(1 - \epsilon)\frac{n}{t\epsilon}, \frac{n}{t\epsilon}]$ or equivalently, $|T| \in [(1 - 2\epsilon)\frac{n}{t\epsilon}, \frac{n}{t\epsilon}]$ and $\frac{n}{|T|}|V_i \cap T| \in [|V_i|{(1 - \epsilon)}, \frac{|V_i|}{1 - 2\epsilon}]$ which implies that $\frac{n}{|T|}|V_i \cap T| \in [|V_i|{(1 - \epsilon)}, |V_i|(1 + 3\epsilon)]$ for small enough $\epsilon$, for all $i > 1$.

 Next, we show that every set with $\phi(W) \leq \phi'$ satisfies the sample condition. Fix such a set $W$.

 Observe that the number of edges cut, $|\delta_G(W)|$ is at most $\phi'|W|$. Mark a bag $V_i$ of the decomposition as ``bad'' if there is a cut edge which belongs to the set $E_i$.
 Additionally, we always mark the bag $V_0$ as bad. Clearly, the number of bad bags is at most $\phi'|W| + 1$. The total number of vertices in these bags is at most $2t\phi'|W| + t\leq \frac{\epsilon|W|}{4}$ (note that $t \leq \frac{\epsilon|W|}{100}$ since $|W| \geq \frac{|\delta_G(W)|}{\phi'} \geq \frac{1}{\phi'}$). Thus at least $(1 - \frac{\epsilon}{4})|W|$ vertices of $W$ are present in good bags. Observe that by the definition of a good bag, if $W$ intersects a good bag $V_i$, it must include the whole bag $V_i$, since there is a way to reach every vertex of the bag using the edges $E_i$, and none of the edges in $E_i$ are cut. 

 Let $G'$ be the set of vertices of the good bags contained in $W$. From the preceeding argument, we know that $|G' \cap W| \geq (1 - \frac{\epsilon}{4})|W|$. Since each good bag $V_i$ has at least $\frac{|V_i|}{t \epsilon}(1 - \epsilon)$ vertices of $T$, we must have $|G' \cap T| \geq (1 - \epsilon)(1 - \frac{\epsilon }{4})\frac{|W|}{t\epsilon} \geq (1 - 2\epsilon)\frac{|W|}{t\epsilon}$.
This gives $|W \cap T| \geq |G' \cap T| \geq (1 - 2\epsilon)\frac{|W|}{t\epsilon}$. Noting that $|T| \in [1-2\epsilon, 1]\frac{n}{t\epsilon}$, this gives that $\frac{n}{|T|}|W \cap T| \geq (1 - 2\epsilon)|W|$.

We now show the proof of the other direction that $W$ does not have too many terminals. As before, the number of bad bags is at most $\phi|W| + 1$ and the total number of vertices in bad bags is at most $ \frac{\epsilon|W|}{4}$. Let $B$ be the union of vertices in all the bad bags. Observe that by construction we must have $|B \cap T| \leq \frac{|B|}{t\epsilon}$. Noting that $|T| \geq (1- 2\epsilon)\frac{n}{t\epsilon}$, this gives $\frac{n}{|T|}|B \cap T| \leq \frac{|B|}{1- 2\epsilon} \leq \frac{\epsilon|W|}{4(1 - 2\epsilon)} \leq \frac{\epsilon|W|}{2}$ for small enough $\epsilon$. By construction, we must have $|G' \cap T| \leq \frac{|W|}{t\epsilon}$ which gives $\frac{n}{|T|}|G' \cap T| \leq \frac{|W|}{1 - 2\epsilon} \leq |W|(1 + 3\epsilon)$. Thus, $\frac{n}{|T|}|W \cap T| \leq \frac{n}{|T|}(|G' \cap T| + |B \cap T|) \leq |W|(1 + 3\epsilon) + \frac{\epsilon|W|}{2} \leq |W|(1 + 4\epsilon)$. This completes the proof.

Next, we show that every set with $\phi_T(W) \leq \frac{\epsilon^2}{200}$ satisfies the sample condition (Note that $\frac{\epsilon^2}{200} > \frac{n}{10|T|}\phi'$). We will do this by showing that $\phi_T(W) \leq \frac{\epsilon^2}{200}$ implies $\phi(W) \leq \phi'$ -  since we already showed that any set with $\phi(W) \leq \phi'$ satisfies the sample condition, we would be done. Again, we have $\delta_G(W) \leq \phi_T(W)|W \cap T| \leq \frac{\epsilon^2}{200}|W \cap T|$, and this upper bounds the number of bad bags. Define $G'$ and $B$ as before. Recall that from bag $V_i$ we pick at most $\frac{|V_i|}{t\epsilon} \leq \frac{2}{\epsilon}$ terminals. This means that the total number of terminals in bad bags $|T \cap B| \leq \frac{2}{\epsilon} \cdot \frac{\epsilon^2}{200}|W \cap T| \leq \epsilon|W \cap T|$. Thus we must have $|W \cap T \cap B| \leq \epsilon|W \cap T|$. This implies that $|G' \cap T| \geq (1 - \epsilon)|W \cap T|$. Since $G'$ is a collection of bags, as in the previous analysis, we must have $|G' \cap T| \in [\frac{|G'|}{t\epsilon}(1 - \epsilon),\frac{|G'|}{t\epsilon}]$. This gives 
$$|W| \geq |G'| \geq t\epsilon|G' \cap T| \geq t\epsilon (1 - \epsilon)|W \cap T| \geq t\epsilon(1 - \epsilon)\frac{200}{\epsilon^2}\delta_G(W) \geq \frac{2(1 - \epsilon)}{\phi'}\delta_G(W) \geq \frac{1}{\phi'}\delta_G(W).$$

But this means that $W$ is $\phi'$-sparse in $G$, and the previous analysis now directly proves the result. \end{proof}


Since $T$ clearly has size $\Theta(\frac{n\phi'}{\epsilon^2})$ and the algorithm runs in near-linear time, we are done.\end{proof}

\subsection{Vertex cuts -- upper bound}
We next show our result for vertex cuts. We will use the following standard Chernoff bounds.
\begin{lemma}[Chernoff Bounds]
Let $X = X_1 + X_2 \ldots X_n$ be the sum of $n$ identically distributed indicator random variables with mean $p$, and suppose that for any subset $S \subseteq [n]$, we have $\Pr[\bigwedge\limits_{i \in S} X_i = 1] \leq p^{|S|}$. Let $E[X] = \mu = np$.

Then we have
\begin{itemize}
\item $\Pr[X > (1 + \delta)\mu] \leq e^{\frac{-\delta^2}{2 + \delta}\mu}$ for all $\delta > 0$. In particular, for $\delta \in (0,1)$, we have $\Pr[X > (1 + \delta)\mu] \leq e^{\frac{-\delta^2}{3}\mu}$ and for $\delta > 1$ we have $\Pr[X > (1 + \delta)\mu] \leq e^{\frac{-\delta}{3}\mu}$.
\item $\Pr[X < (1- \delta)\mu] \leq e^{\frac{-\delta^2 \mu}{3}}$ for $\delta \in (0,1)$.
\end{itemize}

\end{lemma}

We need the following well-known result about set families with small VC-dimension.
\begin{lemma}[Sauer-Shelah Lemma,~\cite{sauer72,shelah72}]
Let $\SS$ be a family of sets on a universe $U$ with VC-dimension $d$, and let $T \subseteq U$ be given. Then the size of the set family $r = \{S \cap T \mid S \in \SS\}$ is at most $\sum_{i = 0}^{d-1} \binom{|T|}{i} \leq d\binom{|T|}{d}$.

\end{lemma}

\begin{lemma}\label{lemma:sampleset_vertex}
For every $\epsilon \in (0,1)$ and $\phi' < \frac{1}{2}$, there is an $(\epsilon,\phi')$ vertex sample set for every graph $G$ of size $\min\{\Theta(n\frac{\phi'}{\epsilon^2}\log(\frac{n\phi'}{\epsilon^3})),n\}$. Further, there is a randomized algorithm that computes such a set with constant probability in near-linear time.
\end{lemma}

\begin{proof}
The proof follows the VC-dimension union bound recipe for $\epsilon$-nets and is along the lines of Lemma 2.6 from~\cite{fm06}. We also refer to the proof of Theorem 5.3.4 of~\cite{Har08}. Before we proceed further with the proof, we note that by Lemma 3.4 of~\cite{fm06}, the family of sets $W \subseteq V(G)$ with $|N_G(W)| = k'$ which is either a connected component or the union of all connected components except one after removing $N_G(W)$, has VC-dimension at most $dk'$, for some constant $d$. Also, recall that if a set $W$ is $\phi'$-sparse, we have $\frac{|N_G(W)|}{|W \cup N_G(W)|} \leq \phi'$, but since $\phi' < \frac{1}{2}$, we also have $|N_G(W)| \leq |W|$ so that in fact $|N_G(W)| \leq 2\phi'|W|$.

Pick a random subset $T \subseteq V(G)$ of terminals with $|T| = C \cdot \frac{n \phi'}{\epsilon^2}{\log \frac{n\phi'}{\epsilon^3}}$ for some large constant $C$ that shall be chosen later. If $C \cdot \frac{n \phi'}{\epsilon^2}{\log \frac{n\phi'}{\epsilon^3}} > n$ we can simply return the whole vertex set as the terminal set, so we assume that this does not happen. We will now show that $T$ is a $(20\epsilon, \phi')$-sample set.

We define families of sets $\SS_1^{k'}$ and $\SS_2^{k'}$ for each $k' \in \{1,2\ldots n\}$ as follows. Let $\SS_1^{k'}$ be the set of all sets $W$ which are (a) $\phi'$-sparse, (b) are either a connected component or the union of all connected components except one after removing a vertex cut and (c) satisfy $|N_G(W)| = k'$. Let $\SS_2^{k'}$ be the set of all sets $W$ which are (i) $\frac{\epsilon^2}{10C\log(\frac{n\phi}{\epsilon^3})}$-terminal sparse, (ii) are either a connected component or the union of all connected components except one after removing a vertex cut and (iii) satisfy $|N_G(W)| = k'$.

If $T$ is not a $(20\epsilon, \phi')$ sample set, one of the following two cases must happen. The first is that there exists a $k' \in [n]$ and a set $W \in \SS_1^{k'}$ that violates the sample condition, so that $|\frac{n}{|T|}|W \cap T| - |W|| \geq 20\epsilon|W|$. The second case is that there exists a  $k'\in [n]$ and a set $W \in \SS_2^{k'}$ that violates the sample condition. We will bound the probability of both these ``bad'' events. Henceforth, we fix a value $k'$ and omit the superscripts and simply denote the set families by $\SS_1$ and $\SS_2$. First, let $E$ be the event that there exists a set $W \in \SS_1$ which violates the sample condition: $|\frac{n}{|T|}|W \cap T| - |W|| \geq 20\epsilon|W|$. Now pick another set of random terminals $T'$ with $|T'| = \frac{|T|(1-\epsilon)}{\epsilon}$. Let $F_X$ be the event that for a fixed set $X$ we have  $|X \cap T'| \in (1-\epsilon, 1+ \epsilon) \frac{|X|}{n} |T'|$. Note that $\mathbb{E}[|X \cap T'|] = \frac{|X|}{n} |T'|$, so by a Chernoff bound we must have $\Pr[F_X] \geq 1 - 2e^{-\frac{\epsilon^2}{3}\frac{|X|}{n}|T'|}$ for any fixed set $X$. For any $W \in \SS_1$, since $W$ is $\phi'$-sparse, $\frac{|W|}{n}|T'| \geq {|W|\phi'}{\frac{C(1-\epsilon)}{\epsilon^3}}\log(\frac{n\phi'}{\epsilon^3}) \geq \frac{Ck'}{4\epsilon^3}\log(\frac{n\phi'}{\epsilon^3})$. It follows that we must have $\Pr[F_W] \geq \frac{1}{2}$ for a large enough choice of $C$.

Now let $E'$ be the event that there exists a set $W \in \SS_1$ with $N_G(W) = k'$ which (i) violates the sample condition and (ii) also satisfies $F_W$. Clearly, $\Pr[E'] \geq \Pr[E'|E]\Pr[E]$, and thus $\Pr[E] \leq \frac{\Pr[E']}{\Pr[E'|E]}$. The denominator is at least $\Pr[F_W]$ for any $W$ that violates the sample condition, and thus $\Pr[E] \leq 2\Pr[E']$.

We now bound $\Pr[E']$. For the sake of analysis, imagine that we pick $X = T \cup T'$ together, and then randomly select a subset of $X$ of size $|T|$ and assign it as $T$, and the rest as $T'$. We can write

$$ \Pr[E'] = \sum_{X} \Pr[E'|X]\Pr[X] $$

Henceforth, we fix $X$ and bound $\Pr[E'|X]$. Define $R = \{W \cap X \mid W \in \SS_1\}$. By the Sauer-Shelah Lemma, $|R| \leq dk' \binom{|X|}{dk'}$. Define $E''$ to be the event that there exists an $r \in R$ that satisfies $r \cap T \notin (1-10\epsilon, 1+10\epsilon)\ell\epsilon$ and $r \cap T' \in (1-\epsilon, 1+\epsilon)\ell$ for some $|T| = \frac{Cn\phi'}{\epsilon^3}\log\frac{n\phi'}{\epsilon^3} \geq \ell \geq \frac{Ck'}{4\epsilon^3}\log\frac{n\phi'}{\epsilon^3}$. Observe that $E'$ implies $E''$, since we can set $\ell = {|W|\frac{|T'|}{n}}$ in $E''$. Thus it is enough to bound $\Pr[E''|X]$.

Fix  $r \in R$ and a value of $\ell$. First, observe that we may assume $|r| \geq (1 - 5\epsilon)\ell$. For if not, we must have $|T' \cap r| \leq |r| \leq (1 - 5\epsilon)\ell$ and hence $\Pr[E''] = 0$. 
Now there are two cases. First suppose that $|r| \geq (1+5\epsilon)\ell$. Then $\mu = \mathbb{E}[|T' \cap r|] = \frac{|T'|}{|T| + |T'|}|r| = (1- \epsilon)|r| \geq (1+ 3\epsilon)\ell$ (for small enough $\epsilon$). But in order to have $\Pr[E'']$ non-zero, we must have $|T' \cap r| \in (1-\epsilon, 1+ \epsilon)\ell$. But this means that $|T' \cap r|$ deviates from its mean $\mu$ by a multiplicative $\epsilon$ (for otherwise $E''$ is impossible).  Using the standard Chernoff Bound this probability is at most $e^{-\frac{\epsilon^2}{3}\ell} \leq e^{-{C}{k'\log \frac{n\phi'}{\epsilon^3}}}$, where the last inequality follows from the fact that $\ell \geq \frac{Ck'}{4\epsilon^3}\log \frac{n\phi'}{\epsilon^3}$.

Alternatively, suppose that  $|r| \in (1-5\epsilon, 1+5\epsilon)\ell$. This means that $\mathbb{E}[|T \cap r|] = \frac{|T|}{|T| +|T'|} |r| \in (1-5\epsilon, 1+5\epsilon) \epsilon \ell $. 
But we have $|T \cap r| \notin  (1-10\epsilon, 1+ 10\epsilon) \ell\epsilon$, hence we must have a multiplicative deviation of $2\epsilon$ from the mean $\mu$. 
For a multiplicative deviation of $2\epsilon$ from this mean, using the Chernoff bound, the probability is at most $2e^{\frac{-\epsilon^3}{3} \ell}$. Since $\ell \geq \frac{Ck'}{4\epsilon^3} \log(\frac{n\phi'}{\epsilon^3})$  this probability is at most $2e^{-\frac{C}{12}k'\log{\frac{n\phi'}{\epsilon^3}}}$.

Since the VC-dimension is at most $dk'$, by the Sauer-Shelah Lemma, $|R|$ is at most $dk' {|X| \choose dk'}$. Using the fact that ${|X| \choose dk'} \leq (\frac{|X|e}{dk'})^{dk'}$, we obtain that this is at most $(\frac{n\phi'}{\epsilon^3})^{10dk'}$. Together, after applying a union bound on these sets, and on the $\frac{Cn\phi'}{\epsilon^3}\log{\frac{n\phi'}{\epsilon^3}}$ choices of $\ell$, the failure probability is at most $10\frac{Cn\phi'}{\epsilon^3}\log{\frac{n\phi'}{\epsilon^3}}e^{-\frac{C}{12}k' \log \frac{n\phi'}{\epsilon^3}} (\frac{n\phi'}{\epsilon^3})^{10dk'} < \frac{1}{100k'^2}$ for a large enough choice of $C$.

 For the other case, let $E_2$ be the event that there is a set $U \in \SS_2^{k'}$ which is not $\phi'$-sparse. We will bound the probability of the event $E_2$. Recall that we already showed the sample set guarantee for all $\phi'$-sparse sets, so if $E_2$ does not occur, every such set $U$ will satisfy the sample set guarantee. First, we start by noting that for any such $U$ which is $\frac{\epsilon^2}{10C\log(\frac{n\phi'}{\epsilon^3})}$-terminal sparse, we must have $|(U \cup N_G(U)) \cap T| \geq \frac{10Ck'}{\epsilon^2} \log(\frac{n\phi'}{\epsilon^3})$, which in turn implies $|U \cap T| \geq \frac{9Ck'}{\epsilon^2} \log(\frac{n\phi'}{\epsilon^3})$ since $|N_G(U)| = k'$.
 
 Now suppose that $U$ is not $\phi'$-sparse. Again, we pick a second set of terminals $T'$ as before. After picking the set $T'$, we have $\mathbb{E}[U \cap T'] = \frac{C\phi'|U|}{\epsilon^3} \log(\frac{n\phi'}{\epsilon^3}) \leq \frac{Ck'}{\epsilon^3} \log(\frac{n\phi'}{\epsilon^3})$ where the inequality holds since we assumed that $U$ is not $\phi'$-sparse. This in turn means that with probability at least $\frac{1}{2}$ we have $|U \cap T'| \leq \frac{2Ck'}{\epsilon^3} \log(\frac{n\phi'}{\epsilon^3})$. As in the previous analysis, we have $\Pr[E_2] \leq 2\Pr[E_2']$, where $E_2'$ is the event that there exists such a set $U$  which satisfies $|U \cap T| \geq 9\frac{Ck'}{\epsilon^2} \log(\frac{n\phi'}{\epsilon^3})$ and $|U \cap T'| \leq \frac{2Ck'}{\epsilon^3} \log(\frac{n\phi'}{\epsilon^3})$.

We now bound $\Pr[E_2']$. We analyze this in a similar way, imagine picking $X = T \cup T'$ together and randomly deciding which elements form $T$ and which form $T'$. We will bound $\Pr[E_2'|X]$ for a fixed $X$. Define $R_2 = \{W \cap X \mid W \in \SS_2\}$.  By the Sauer-Shelah Lemma, $|R_2| \leq dk' {|X| \choose dk'}$. Define $E_2''$ to be the event that there exists an $r \in R_2$ that satisfies $r \cap T \geq 9\frac{Ck'}{\epsilon^2}\log(\frac{n\phi'}{\epsilon^3}) $ and $|r \cap T'| \leq \frac{2Ck'}{\epsilon^3} \log(\frac{n\phi'}{\epsilon^3})$. As in the previous case, we have $\Pr[E_2'] \leq \Pr[E_2'']$, so it is enough to bound $\Pr[E_2'']$.

Again, there are two cases. First, suppose that $|r| \geq \frac{5Ck'}{\epsilon^3} \log(\frac{n\phi'}{\epsilon^3})$. Then $\mathbb{E}[|T' \cap r|] = \frac{|T'|}{|X|}|r| = (1-\epsilon)|r| \geq \frac{5Ck'}{2\epsilon^3} \log(\frac{n\phi'}{\epsilon^3})$. It follows that if $E_2''$ occurs then $|T' \cap r|$ deviates from its mean by more than a multiplicative factor of  $\frac{1}{10} \geq \epsilon$, and hence this happens with probability at most $2e^{-\frac{\epsilon^2}{3}\frac{Ck'}{\epsilon^3}\log(\frac{n\phi'}{\epsilon^3})}$.

Otherwise, we must have $|r| \leq \frac{5Ck'}{\epsilon^3} \log(\frac{n\phi'}{\epsilon^3})$. Now we obtain $\mathbb{E}[|T \cap W|] \leq \frac{5Ck'}{\epsilon^2} \log(\frac{n\phi'}{\epsilon^3})$. However, for $E_2''$ to occur we must have $|T \cap r| \geq \frac{9Ck'}{\epsilon^2} \log \frac{n\phi'}{\epsilon^3}$. This implies that $|T \cap r|$ deviates from its expectation by a multiplicative factor $\delta > 1$. By the Chernoff bound, we have that this happens with probability at most $e^{-\frac{\delta\mu}{3}} \leq e^{-\frac{4Ck'}{3\epsilon^2} \log(\frac{n\phi'}{\epsilon^3})}$.

 By the Sauer-Shelah Lemma  $|R_2|$ is at most $dk' {|X| \choose dk'}$. Using the fact that ${|X| \choose dk'} \leq (\frac{|X|e}{dk'})^{dk'}$, we obtain that this is at most $(\frac{n\phi'}{\epsilon^3})^{10dk'}$. Together, after applying a union bound on these sets, the failure probability is at most $10e^{-\frac{C}{3}k' \log \frac{n\phi'}{\epsilon^3}} (\frac{n\phi'}{\epsilon^3})^{10dk'} < \frac{1}{100k'^2}$ for a large enough choice of $C$.
 
Finally, combining both the cases and applying a union bound on all sizes $k'$ between $1$ and $n$, we obtain that the probability of failure is at most $\frac{1}{10}$, and hence $T$ is a $(20\epsilon, \phi')$ sample set with probability at least $\frac{9}{10}$.
\end{proof}

\subsection{Vertex cuts -- lower bound}\label{subsec:sample_set_lower_bound}
A natural question is if the $\log(\frac{n\phi'}{\epsilon^3})$ factor in the size of the sample set is avoidable for vertex cuts, just as in the case of the edge version. We show that this is not the case. The next lemma shows that this result for vertex sample sets is essentially tight using a connection with Ramsey numbers.
\begin{lemma}
There exists a family of graphs $\mathcal{F}$ and parameters $\phi$ such that for any constant $\epsilon \in (0,1)$ any $(\epsilon, \phi)$ vertex sample set $T$ must satisfy $T = \Omega(N\phi \log({N\phi}))$ for an $N$-vertex graph from $\mathcal{F}$.
\end{lemma}

\begin{proof}
Consider the complete graph $K_n$ on $n$ vertices. Now construct the vertex-edge incidence graph of this graph. Formally, consider the bipartite graph $H$ with vertex set $L \cup R$. The set $R$ contains a vertex $x_w$ for every vertex $w \in V(K_n)$, and the set $L$ contains a vertex $y_{uv}$ all pairs of vertices $u,v \in V(K_n)$. For every vertex $y_{uv}$, add the edges $(x_u, y_{uv})$ and $(x_v, y_{uv})$ to $E(H)$. Let $k = \lfloor \frac{\log n}{2} \rfloor$, $\phi = \frac{k}{{k \choose 2} + k} = \frac{2}{k+1}$ and choose $s = {k \choose 2}$. We will let $R(k,k)$ denote the (Ramsey) number such that any graph on $R(k,k)$ vertices must have either a clique of size $k $ or an independent set of size $k$. 
It is well known that $R(k,k) \leq 4^k$ while a recent breakthrough lowered this upper bound to $(4-\epsilon)^k$ for some absolute constant $\epsilon > 0$~\cite{CGMS23}. 

Let $T$ be a $(\epsilon, \phi)$ sample set for $H$. Notice that $H$ has $N = {n \choose 2} + n$ vertices. We will show that $|T| = t \geq \frac{1}{2}{{n \choose 2}}$. Clearly $\frac{1}{2}{n \choose 2} =  \Omega(N\phi \log (N\phi))$ since $\phi = \OO(\frac{1}{\log n})$.
For the sake of contradiction suppose that $t \leq \frac{1}{2}{n \choose 2}$.  In $K_n$, $T$ corresponds to a set of vertices and edges. Let $G'$ be the resulting graph after removing from $K_n$ the ``edges'' of $T$. Since $G'$ has $n \geq 4^k \geq R(k,k)$ vertices, it follows that there is a set of $k$ vertices in $G'$ that form either an independent set of size $k$ or a clique of size $k$. In $H$, this corresponds to a set $Y \subseteq R$ of $k$ right vertices. Observe that after removing $Y$, the ${k \choose 2}$ left vertices (call them $U$) corresponding to the edges of $K_n$ between vertices of $Y$, are disconnected from the rest of the graph. Also note that $H \setminus (U \cup Y)$ is connected. Thus $U$ is the union of all connected components except one after removing $Y$. Further, $U$ has sparsity $\frac{k}{{k \choose 2} + k} = \phi$. It follows that $U$ must satisfy $|\frac{N}{|T|}{|U \cap T|} - |U| | \leq \epsilon|U|$. 

There are two cases. First, suppose that $Y$ corresponds to a clique in $G'$. This means that no vertex of $U$ is in $T$. But since $|\frac{N}{|T|}{|U \cap T|} - |U| | \leq \epsilon|U|$ this cannot happen for any $T$ with $|T| > 0$ and hence we obtain a contradiction. So now suppose that $Y$ corresponds to an independent set in $G'$. But this means every vertex of $U$ is in $T$, and hence $|U \cap T| = |U|$. Since $|T| \leq \frac{1}{2}{n \choose 2} \leq \frac{N}{2}$, $U$ cannot satisfy $|\frac{N}{|T|}{|U \cap T|} - |U| | \leq \epsilon|U|$, and we have a contradiction again.    
\end{proof}

\section{\SSE~and \VSC~via LP rounding}\label{sec:sselp}

\label{section:sse}

In this section, we will use an extension of an existing LP-based algorithm combined with sample sets for sparse cuts in order to obtain an $\OO(\log k)$-approximation algorithm for small set expansion and an $\OO(\log k + \log \log n\phi)$-approximation algorithm for \VSC.

\subsection{\SSE}

We start by defining another problem which we call \STE.
\defparproblem{\STE}{Graph $G$ and set of terminals $T \subseteq V(G)$}{$\phi,s$}{Find a set of vertices $S'$ with $|S' \cap T| \leq s$, and $|\delta_G(S')| \leq \phi|S' \cap T|$, or report that no such set exists.}

\begin{definition}[$(\alpha,\beta)$-approximation]
We say that an algorithm is an $(\alpha, \beta)$-approximation for $\STE$ if given parameters $\phi,s$ it outputs a set $S' \subseteq V(G)$ with  $|S' \cap T| \leq \beta s$ and $|\delta_G(S')| \leq \alpha \phi |S' \cap T|$, or outputs that there is no set $S \subseteq V(G)$ with $|S \cap T| \leq s$ and $|\delta_G(S)| \leq \phi |S \cap T|$.

\end{definition}
When $T = V(G)$ we get the \SSE{} problem. The following result from~\cite{hc16} gave an $\OO(\log s)$-approximation for \SSE.

\begin{theorem}[see~\cite{hc16}]\label{thm:logs}
\SSE{} admits an $(\OO(\log s), \OO(1))$-approximation.
\end{theorem}

As described in the introduction, our goal will be to use an ``$\OO(\log s)$-type''-approximation algorithm for the terminal version, so we need an algorithm for \STE{}. The ideas remain similar to \SSE. We follow the ideas of~\cite{hc16} with minor changes to obtain this. For completeness' sake, we give a simple and concise algorithm for \STE\; based on ideas similar to that in~\cite{hc16}.

This is deferred to the Appendix (\Cref{thm:sse}). In fact,~\Cref{thm:sse} is a stronger result, and the approximation algorithm for \STE{} will follow directly as a corollary.

\begin{lemma}[Follows from~\Cref{thm:sse}]\label{cor:sse}
\STE{} admits an \\$(\OO(\log s), \OO(1))$-approximation.
\end{lemma}

Henceforth, we will assume that there indeed exists a set of vertices $S$  with $|S \cap T| =  s$ and $\delta_G(S) =  \phi |S|$ (if not, it is easy to see that our algorithm will detect that this is the case). Let us denote by $k = \phi s$ the cut-size of $S$.

We now combine the notion of sample sets with~\Cref{cor:sse} to obtain an $\OO(\log k)$-approximation for \SSE{} in polynomial time.

\begin{proof}[Proof of~\Cref{thm:logksse}]

We proceed as follows. Choose an arbitrary constant $\epsilon \in (0,\frac{1}{100})$ and construct an $(\epsilon, \phi')$ sample set $T$ of size $\min\{n,\frac{dn\phi'}{\epsilon^2}\}$ where $d$ is the constant hidden in the big-$\Theta$ notation of \Cref{lemma:sample_set} and $\phi' = \frac{ck \log k}{s} = c\phi \log k$, for some large enough constant $c$ which we will choose later. Henceforth we assume that $\frac{dn\phi'}{\epsilon^2} < n$ since otherwise we would have $s = \OO({k\log k})$ and we can simply set the terminal set to be the whole vertex set to obtain an $\OO(\log s) = \OO(\log k)$-approximation algorithm from~\Cref{thm:logs}.

Now recall that we are promised that there exists a set $S$ of size $s$ with expansion $\phi$. The sample set guarantee means that $\frac{n}{|T|}|S \cap T| = \Theta(|S|)$. This in turn implies that $|S \cap T| = \Theta(|S|\frac{|T|}{n})$ which gives $|S \cap T| = \Theta(|S|{dc\phi\log k}) = \Theta(dck\log k) =  s'$ (say).

Run the algorithm $A$ for \STE{} from~\Cref{thm:logs} with the parameter $s'$.

Notice that $|S \cap T| = \Theta(dck\log k)$, hence $\phi_T(S) \leq \OO(\frac{1}{dc\log k})$. Then the algorithm must return a set $U$ that is $\OO(\log s')\phi_T(S) = \OO(\log k \cdot \frac{1}{dc\log k}) = \OO(\frac{1}{dc})$-terminal sparse, and $|U \cap T| \leq \OO(s')$. Choose $c$ large enough such that $\phi_T(U) \leq \frac{\epsilon^2}{100d} \leq \frac{n}{10|T|}\phi'$.

The sample set guarantee from \Cref{def:sample_set} implies that $|U| \in[1-\epsilon, 1+ \epsilon]\frac{n}{|T|} |U \cap T|$.  Since $\frac{n}{|T|} = \OO(\frac{s}{k \log k})$ and $|U \cap T| \leq \OO(s') = \OO(k \log k)$, we must have $|U| \leq \OO(s)$. Also, $U$ is $(\frac{\epsilon^2}{100d})$-terminal sparse, so by the sample set guarantee, $U$ must be $\OO(\frac{\epsilon^2}{100d}\frac{|T|}{n})$-sparse in $G$. But $\frac{|T|}{n} = \OO(\phi') = \OO(d\phi \log k)$ and hence it must be the case that $U$ is $\OO(\phi \log k)$-sparse in $G$.\end{proof}

\subsection{\VSC}
In this section we obtain an $\OO(\log k + \log \log n\phi)$-approximation for \VSC, proving~\Cref{thm:logkvertex}. Again, we start with the following theorem which approximates terminal sparsity. The proof again uses techniques from~\cite{hc16} and~\cite{minmax}, except that now there is no small set guarantee: the result is only for \VSC. We remark that our algorithm is also similar to the algorithm of~\cite{lee17}.
\begin{restatable}{theorem}{thmlpvertex}\label{thm:lpvertex}
Given a graph $G$ with a set of terminals $T$, suppose that there is a vertex cut $(L, C, R)$ with $|R \cap T| \geq |L \cap T| = s$, $|C| = k$ and terminal sparsity $\phi = \frac{k}{s + |C \cap T|}$. Further suppose that $s \geq k \log s$ (so that $\phi \leq \OO(\frac{1}{\log s})$). Then there is a polynomial time algorithm that finds a cut $(L',C',R')$ with terminal sparsity $\phi' = \frac{|C'|}{\min\{|(L' \cup C') \cap T|, |(R' \cup C') \cap T|\}} \leq \OO(\phi \log s)$. 
\end{restatable}
 We postpone the proof of this theorem to the appendix. We are now ready to prove ~\Cref{thm:logkvertex}.

\begin{proof}[Proof of~\Cref{thm:logkvertex}]

We proceed as follows. Suppose there exists a vertex cut $(L,C,R)$ with $|C| = k$, and $|L \cup C| = \frac{k}{\phi} = s + k$, so that $|L| = s$. Let $\phi' = \lambda\phi(\log k + \log \log {n\phi})$ where $\lambda$ is a large constant that will be chosen later. Choose $\epsilon = \frac{1}{100}$ and construct a $(\epsilon, \phi')$ vertex sample set $T$ of size $\min\{n,\frac{Dn\phi'}{\epsilon^2}\log(\frac{n\phi'}{\epsilon^3})\}$ where $\phi' = \lambda\phi(\log k + \log \log {n\phi})$ and $D$ is the constant hidden in the big-Theta notation of~\Cref{lemma:sampleset_vertex}. Henceforth, we assume $\frac{Dn\phi'}{\epsilon^2}\log(\frac{n\phi'}{\epsilon^3}) < n$, for if not, then it follows that $s = \OO(k(\log k + \log \log {n\phi})\log({n\phi'}))$ - then applying the result of~\Cref{thm:lpvertex} with the terminal set $T = V(G)$ clearly gives us an $\OO(\log k + \log \log {n\phi})$-approximation.

Observe that by the guarantee of the sample sets, we must have $|L \cap T| = \Theta(\lambda Dk(\log k + \log \log {n\phi})\log {n\phi'})$ (we omit the dependence on $\epsilon$ henceforth). It follows that $L$ is $\OO(\frac{1}{\lambda D(\log k + \log \log {n\phi})\log n\phi'})$-terminal sparse. Now we run the algorithm from~\Cref{thm:lpvertex} to obtain another vertex cut $(L',C',R')$ with $|R' \cap T| \geq |L' \cap T|$ which is $\psi$-terminal sparse where $$\psi = \OO(\frac{1}{\lambda D(\log k + \log \log {n\phi})\log n\phi'} \log |L \cap T|) = \OO(\frac{1}{\lambda D \log n\phi'}).$$ Choose $\lambda$ large enough so that $\psi \leq \frac{\epsilon^2}{10D \log \frac{n\phi'}{\epsilon^3}} \leq \frac{n}{10|T|}\phi'$. Let $s' = |L' \cap T|$ and $s'' = \frac{n}{|T|}s'$. 

There are two cases. Fix a constant $\alpha > 1$ which will be chosen later. First, suppose no component in $G \setminus C'$ has size $>n - \frac{s''}{\alpha}$. Then it is clear that we can partition $V(G) \setminus C'$ into $A \cup B$ such that both $A$ and $B$ have $\Omega(s'')$ vertices.

Now suppose there is indeed a component $X$ which has size $> n - \frac{s''}{\alpha}$. We will obtain a contradiction. Clearly, it must be the case that $V(G) \setminus (C' \cup X)$ has size $< \frac{s''}{\alpha}$. But then $V(G) \setminus (C' \cup X)$ is the set of all but one component after removing a vertex cut. Also clearly, either $L \subseteq V(G) \setminus (C' \cup X)$ or $R \subseteq V(G) \setminus (C' \cup X)$, and hence $V(G) \setminus (C' \cup X)$ contains at least $s'$ terminals. Then since $V(G) \setminus (C' \cup X)$ is $\frac{n}{10|T|}\phi'$-terminal sparse, $V(G) \setminus (C' \cup X)$ must satisfy the sample condition. But this means $|V(G) \setminus (C' \cup X) \cap T| \leq 2 \frac{|T|}{n}\frac{s''}{\alpha} \leq \frac{2s'}{\alpha}$, which is a contradiction for $\alpha > 2$.

Thus we must in fact have that every component in $G \setminus C'$ has size $\leq n - \frac{s''}{\alpha}$, and hence we obtain a partition of $V(G) \setminus C'$ into $A \cup B$ such that both $A$ and $B$ have $\Omega(s'')$ vertices. Since $(L',C',R')$ was $\frac{n}{10|T|} \phi'$-terminal sparse, we must have $|C'| \leq \frac{n}{10|T|}\phi'|(L' \cup C') \cap T| \leq \OO(\frac{n}{10|T|}\phi'|L' \cap T|) = \OO(\frac{n}{|T|}\phi's') = \OO(s''\phi')$. It follows that the vertex cut $(A,C,B)$ must be $\OO(\phi')$-sparse. Hence we obtain an $\OO(\frac{\phi'}{\phi}) = \OO(\log k + \log \log {n\phi})$-approximation.\end{proof}

\section{\SC{} via a cut-matching game for small set expanders}\label{sec:cutmatching}

In this section, we develop a cut-matching game for small set expanders. This will be our main ingredient to obtain an $\OO(\log^2 k)$-approximation for \SC{} and an $\OO(\log^2 k + \log^2 \log n\phi)$-approximation for \VSC{} in time $m^{1 + o(1)}$, proving theorems~\ref{thm:cutmatchingedgelogk} and~\ref{thm:cutmatchingvertexlogk}. We will use the analysis of~\cite{krv07}. On a high level, the framework of~\cite{krv07} constructs an expander using a two-player cut-matching game in $\OO(\log n)$ rounds. We give a cut matching game that constructs \emph{$s$-small set expanders}: these are graphs where every set of size at most $s$ is expanding. We show that there is a cut-matching game that can construct such expanders in only $\OO(\log s)$ rounds. 

We begin by reviewing the cut-matching framework of~\cite{krv07}. The cut-matching game is a two-player game, with a cut player $\CC$ and a matching player $\MM$. The cut and matching players take alternate turns. The game proceeds in rounds, where each round is a turn of the cut player followed by the matching player. Initially, the game starts with an edgeless graph $H$ on $n$ vertices. In each round, the cut player $\CC$ outputs a bisection $(B, \overline{B})$ with $|B| = |\overline{B}| = \frac{n}{2}$. The matching player then chooses a perfect matching from $B$ to $\overline{B}$, which is then added to the graph $H$. We have the following theorem from~\cite{krv07}. 

\begin{theorem}[Follows from~\cite{krv07}]\label{thm:kkov}
There is a (exponential time) cut player strategy so that for any matching player strategy, the graph $H$ is an $\Omega(1)$-expander after $\OO(\log n)$ rounds.

\end{theorem}

On a high level, this strategy is quite simple: in each round, the cut player $\CC$ finds a cut $(A, \overline{A})$ with $\frac{n}{4} \leq |A| \leq |\overline{A}|$ in $H$ with expansion (sparsity) at most $\frac{1}{100}$. This cut is then extended to an arbitrary bisection $(B, \overline{B})$ such that $A \subseteq B$.

Now we briefly recall how the cut-matching game is used to approximate \SC. As discussed in the Overview, this essentially involves implementing a matching player strategy. Throughout, we assume that a parameter $\phi'$ is given. The right value of $\phi'$ will be found by binary search. 
The matching player $\MM$ then, in each round, given a bisection $(B, \overline{B})$ by the cut player, outputs either (i) a $\phi'$-sparse cut in $G$, or (ii) a perfect matching $M$ matching vertices of $B$ to those of $V(H)\setminus B$ such that $M \preceq^{flow} \frac{G}{\phi'}$ (one unit of demand across each edge of $M$ can be routed simultaneously in $G$ with congestion $\frac{1}{\phi'}$). The matching $M$ is added to $E(H)$.

It follows that if the matching player did not return a $\phi'$-sparse cut in any of these rounds, $H \preceq^{flow} \OO(\frac{\log n}{\phi'}) G$. Since $H$ must be a $\Omega(1)$-expander using~\Cref{thm:kkov}, we obtain that $G$ must be a $\Omega(\frac{\phi'}{\log n})$-expander. This gives an $\OO(\log n)$-approximation.

We use a modified cut-matching game to get an $\OO(\log^2 s)$-approximation for (terminal) sparse cuts. In contrast to the~\cite{krv07} cut-matching game which guarantees that the graph $H$ is expanding after $\OO(\log n)$ rounds, we give a cut player strategy to construct a \emph{small set expander} that ensures that every \emph{set of size at most s} is expanding after $\OO(\log s)$ rounds. We remark that there is a slight technical difference in our cut-matching game. In each round, the cut player, instead of returning a single bisection $(B, \overline{B})$, outputs a collection $\mathcal{C}$ of constantly many pairs of disjoint sets $(X,Y)$ with $|X| = |Y|$, where $X \cup Y$ may not be the whole vertex set $V(H)$. The matching player then finds a perfect matching between $X$ and $Y$ for every $(X,Y) \in \mathcal{C}$. Under this setting, we show the following results.

\begin{theorem}\label{thm:cutplayerpoly}
There is a near-linear time cut player strategy so that after $\OO(\log s)$ rounds, for any matching player strategy, every set with at most $s$ vertices in $H$ has expansion $\Omega(\frac{1}{\log s})$.
\end{theorem}

We also show how to improve this result if we are allowed exponential time cut player strategies. We do not use this result for approximating sparsest cut, but just obtain this as a result for the cut-matching game itself, generalizing the result of~\cite{krv07}. 

\begin{theorem}\label{thm:cutplayer}
There is a (exponential time) cut player strategy so that after $\OO(\log s)$ rounds, for any matching player strategy, every set with at most $s$ vertices in the graph $H$ has expansion $\Omega(1)$.
\end{theorem}

Finally, we also show how to use this version of cut-matching game to approximate small sparse cuts. We will need this for the slightly more general case of terminal expansion.

\begin{theorem}\label{thm:cutmatching}
Given a parameter $\phi'$, a graph $G$ and a set of terminals $T$, one can find in almost-linear time either (a) a $\phi'$-terminal sparse cut or (b) conclude that every set with $\leq s$ terminals is $\frac{\phi'}{\log^2 s}$-terminal-expanding.

\end{theorem}

The next theorem is for the vertex version.

\begin{theorem}\label{thm:cutmatchingvertex}
Given a parameter $\phi' \leq \frac{1}{2}$, a graph $G$ and a set of terminals $T$, one can find in almost-linear time either (a) a $\phi'$-terminal vertex sparse cut or (b) conclude that every set with $\leq s$ terminals is $\frac{\phi'}{\log^2 s}$-terminal-vertex expanding. Formally, for every set $L$ with at most $s$-terminals, we have $|N_G(L)| \geq \frac{\phi'}{\log^2 s} |(L \cup N_G(L)) \cap T|$.

\end{theorem}

The main goal of the rest of this section will be to prove theorems \ref{thm:cutplayerpoly}, \ref{thm:cutplayer}, \ref{thm:cutmatching} and \ref{thm:cutmatchingvertex}.

\subsection{Cut player}
The goal of this section will be to prove~\Cref{thm:cutplayerpoly} and \Cref{thm:cutplayer}. In this entire section, we will assume that $\rho$ is a large enough constant, which will be chosen based on the analysis - we do not attempt to optimize $\rho$.

The game starts similarly with the empty graph $H$ on the vertex set $V(H)$. Before we proceed, we note that as we will show later, our cut-matching game will run for at most $d \log s$ rounds for some constant $d$ not depending on $\rho$. Also, throughout this section, we indicate all dependencies on $\rho$ in big-O and big-$\Omega$ notations: hence we will assume that $\rho$ can be chosen bigger than any constant hidden in asymptotic notations. Throughout the algorithm, after every round, we will ensure that the graph $H$ remains regular. 

First, we will need the following two results which give a Cheeger-type approximation for finding balanced sparse cuts. On a high level, both these algorithms either find a balanced low conductance cut in $H$, or find a small set that overlaps all low conductance cuts significantly.

\begin{theorem}[\cite{ov11}]{\footnote{The actual theorem is a little different, but this slightly more precise version follows from the proof.}}\label{thm:bal_sep}
Given parameters $b \in (0,\frac{1}{2})$, $\psi$ and a regular graph $H$, there exist constants $d$ and $d'$ (which may depend on $b$) and an algorithm that runs in time $\OTIL(\frac{m}{\psi})$ that with high probability either 
\begin{enumerate}
\item Outputs a $\Omega_b(1)$ balanced cut with conductance at most $d'\sqrt{\psi}$, or
\item Outputs a set $Y'$ with $|Y'| \leq \frac{bn}{4}$ such that $\frac{|Y' \cap Z|}{|Z|} \geq \frac{5}{8} - \frac{b}{4}$ for every set of vertices $Z$ with $|Z| \leq \frac{n}{2}$ that has \emph{conductance} smaller than $d\psi$.
\end{enumerate}
\end{theorem}

A similar but stronger result can be obtained easily, if we allow exponential time.

\begin{lemma}\label{lemma:bal_sep}
Given parameters $b \in (0,\frac{1}{2})$, $\psi$ and a regular graph $H$, there exists an algorithm that runs in exponential time that either 
\begin{enumerate}
\item Outputs a $\frac{b}{4}$-balanced cut with conductance at most $\psi$, or
\item Outputs a set $Y'$ with $|Y'| \leq \frac{b}{4}n$ such that $\frac{|Y' \cap Z|}{|Z|} \geq 0.6$ for every set of vertices $Z$ with $|Z| \leq \frac{n}{2}$ that has \emph{conductance} smaller than $\frac{\psi}{3}$.
\end{enumerate}
\end{lemma}
\begin{proof}
Check if there is a $\frac{b}{4}$-balanced cut with conductance at most $\psi$. If yes, return it. Otherwise, find the largest $b' < \frac{b}{4}$ for which there is a $b'$-balanced cut with conductance at most $\psi$. Suppose this cut is $(Y', \overline{Y'})$ with $|Y'| \leq |\overline{Y'}|$. Now we claim that for any $\leq \frac{\psi}{3}$ conductance cut $(Z, \overline{Z})$, with $|Z| \leq |\overline{Z}|$ it must be the case that $|Z \cap Y'| \geq 0.6|Z|$.

Suppose this is not the case. Consider the union $Z \cup Y'$. We will show that $Z \cup Y'$ has conductance at most $\psi$, contradicting the fact that $Y'$ is a most balanced cut with conductance at most $\psi$. First note that $|\delta_H(Z \cup Y')| \leq |\delta_H(Z)| + |\delta_H(Y')|$. We also have $\vol_H(Z \cup Y') = \vol_H(Y') + \vol_H(Z \setminus Y')$. Since $|Y' \cup Z| \leq \frac{bn}{4} + \frac{bn}{4} = \frac{bn}{2} \leq \frac{n}{2}$, we must have $\vol_H(Y' \cup Z) \leq \vol_H(V(H) \setminus (Y' \cup Z))$. Thus the conductance of $Z \cup Y'$ is  $$\frac{|\delta_H(Z \cup Y')|}{\vol_H(Z \cup Y')} \leq  \frac{|\delta_H(Z)| + |\delta_H(Y')|}{\vol_H(Z \setminus Y') + \vol_H(Y')} \leq \min\left(\frac{|\delta_H(Y')|}{\vol_H(Y')}, \frac{|\delta_H(Z)|}{\vol_H(Z \setminus Y')}\right)$$

The first term is at most $\psi$, since $Y'$ has conductance at most $\psi$. The second term is at most $\frac{\psi}{3\cdot 0.4} \leq \frac{\psi}{1.2} \leq \psi$, since $\vol_H(Z \setminus Y') \geq 0.4\vol_H(Z)$. This means that $Z \cup Y'$ has conductance at most $\psi$, a contradiction.\end{proof}

Before we describe the cut player strategy, we need a key subroutine which we describe next. Let $S \subseteq V(H)$ with $|S| \leq \frac{|V(H)|}{2}$ be any set of vertices of $H$, which is $\frac{1}{\rho^{101}}$-sparse in $H$.
We now describe a procedure, which we call {\textsc{ImproveCut}}, that when given a sparse cut in $H$, uses a single maxflow call to obtain a cut which is sparse within $S$, for any such $S$. We describe this procedure first by using exact $s$-$t$ min-cut for the sake of clarity: it will subsequently be clear that one can use a slightly relaxed notion of min-cuts as well. This is necessary to get a near-linear time cut-player, since the best known exact $s$-$t$ max-flow/min-cut algorithm~\cite{chen22} runs in almost-linear time.

\begin{algorithm}[H]\label{alg:improvecut}
\caption{{\textsc{ImproveCut}}}
\begin{algorithmic}[!htbp]
\State \IP A cut $(W,\overline{W})$ in $H$ that is $\Omega(1)$ balanced. Also, $\phi_H(W) \leq  \frac{1}{100\rho^3}$.
\State
\OP  A cut $(Q, \overline{Q})$ that is $\Omega(1)$ balanced such that $|\delta_{H[S]}(Q \cap S)| \leq \frac{|S|}{100\rho}$ for every $S$ that is $\frac{1}{\rho^{101}}$-sparse in $H$. 

\\

\State Consider the graph obtained by gluing the edges $\delta_H(W)$ to the induced subgraph $H[W]$. In this graph, for every edge $e = \{u,v\} \in \delta_H(W)$ with $u \in W$,  create a new vertex $x_e$. Remove the edge $\{u,v\}$ and add the edge $\{u,x_e\}$. Denote the resulting graph as $H'$.\\

\State We set up a flow problem on the modified graph $H'$. Add two new vertices $s,t$ and add the edges $(s, x_e)$ for each vertex $x_e$, each with capacity $1$. Add edges $(v,t)$ for every $v \in W$, each with capacity $\frac{1}{\rho^2}$. Each edge of $E(H')$ has capacity $1$.

\State Find a maximum $s$-$t$ flow and a corresponding min-cut $M$, where $M$ is the set of vertices reachable from $s$ in the min-cut. Let $R = M \setminus (\{x_e \mid e \in \delta_H(W)\} \cup \{s\}$)
and let $\overline{Q} = \overline{W} \cup R$, ${Q} = V(H) \setminus \overline{Q}$.

\State \Return $(Q, \overline{Q})$

\end{algorithmic}
\end{algorithm}

\begin{lemma}
The procedure \IC{} is sound. That is, $(Q,\overline{Q})$ is $\Omega(1)$ balanced and \\$\delta_{H[S]}(Q \cap S) \leq \frac{|S|}{100\rho}$ for every $S$ that is $\frac{1}{\rho^{101}}$-sparse in $H$.

\end{lemma}

\begin{proof}

Notice that since $\phi_H(W) \leq \frac{1}{100\rho^3}$ the number of edges in the cut $\delta_H(W)$ is at most $\frac{n}{100\rho^3}$. Hence the max flow is always at most $\frac{n}{100\rho^3}$, and thus the number of edges $(v,t)$ which can be cut for some $v \in V(H')$ is at most $\frac{n}{100\rho}$. This means that both $Q$ and $V(H) \setminus Q$ have $\Omega(n) - \frac{n}{100\rho} = \Omega(n)$ vertices.

 Suppose for contradiction that $Q$ is not sparse in $S$. In particular, suppose that $|\delta_{H[S]}(Q \cap S)|  \geq \frac{|S|}{100 \rho}$. Observe that we can assume that the min-cut $M$ does not cut any edges of the form $(x_e,u)$ where $u \in W$, for otherwise, we can construct another mincut $M'$ which cuts the edge $(s, x_e)$ instead of the edge $(x_e,u)$. This means that $|\delta_{H[S]}(Q)| = |\delta_{H[S]}(M)|$. Since the min-cut is saturated in any max-flow, the edges in $\delta_{H[S]}(M)$ must all be saturated in any max-flow. If $|\delta_{H[S]}(M)| \geq \frac{|S|}{100\rho}$, since the total capacity of the edges from $S$ to the vertex $t$ (the edges $(v,t)$ for $v \in S$) is at most $\frac{|S|}{\rho^2}$, there is an excess flow of $\frac{|S|}{100\rho} - \frac{|S|}{\rho^2} \geq \frac{|S|}{200\rho}$ which must leave the set $S \cap Q$. However $|\delta_H(S)| \leq \frac{|S|}{\rho^{101}}$, and hence this is impossible. Thus we conclude $|\delta_{H[S]}(Q \cap S)| = |\delta_{H[S]}(M)| \leq \frac{|S|}{100 \rho}$.\end{proof}

 To obtain an algorithm with nearly linear running time,  we modify \IC{} to  use the notion of fair minimum cuts of Li et al.~\cite{li2023near} instead of an exact min-cut algorithm. An $s$-$t$ $\alpha$-fair min cut in the graph $H$ is an $s$-$t$ cut $Z$ so that there exists a feasible $s$-$t$ flow that saturates $\frac{1}{\alpha}$ fraction of each cut edge $\delta_H(Z)$. The result of~\cite{li2023near} shows that there exists an algorithm that runs in time $\OTIL(\frac{m}{\epsilon^3})$ and computes a $(1 + \epsilon)$- fair $s$-$t$ min-cut $Z$. It is easy to check that the above proof still works by replacing the exact minimum cut $M$ with the fair min-cut $Z$ for a small enough choice of $\epsilon$ and similarly guarantees $|\delta_{H[S]}(Z)| \leq \frac{|S|}{100\rho}$.

Equipped with this result, we are now ready to describe the cut player strategy in each round. Since our strategies for both the near-linear time cut player and the exponential time cut player are similar, we describe both of them together.

\begin{algorithm}[H]
\caption{Cut player strategy}
\begin{algorithmic}[!htbp]

\State \textbf{Input:} A regular graph $H$ with degree $\Delta$.
\State \textbf{Output:} A collection $\mathcal{C}$ of constantly many pairs of disjoint vertices $(X,Y)$ with $|X| = |Y|$.


\State We start by setting $\psi = \frac{1}{10^4d'^2\rho^8\Delta^2}$ or $\psi = \frac{1}{10^4d'^2\rho^8\Delta}$, $b = \frac{1}{100}$ and invoking  \Cref{thm:bal_sep} or \Cref{lemma:bal_sep} on the graph $H$ (according to whether we use near-linear time or exponential time respectively).
\If{\Cref{thm:bal_sep} or \Cref{lemma:bal_sep}  return a set $Y'$ such that $\frac{|Y' \cap Z|}{|Z|} \geq 0.6$ for each $Z$ that has conductance at most $d\psi$ or $\frac{\psi}{3}$ respectively}
\State $Q \leftarrow Y'$

\ElsIf{\Cref{thm:bal_sep} or \Cref{lemma:bal_sep} return a $\Omega_b(1)$ balanced cut $(W, \overline{W})$}
\State $(Q,\overline{Q}) \leftarrow \IC(W,\overline{W})$
\State Suppose that  $\alpha n = |Q| \leq |\overline{Q}|$. Find a covering $\mathcal{R}$ of $\overline{Q}$ into $\OO(1)$ many sets of size $|Q|$ each as follows - partition $\overline{Q}$ into sets of size $|Q|$ except for one set which may have size less than $Q$, and then arbitrarily extend the last set to a set of size $|Q|$.
\State Add into $\mathcal{C}$ the pair of sets $(Q, R)$ for each $R \in \mathcal{R}$. Also add an arbitrary partition $(P_1, P_2)$ of $V(H) \setminus (R \cup Q)$ with $|P_1| = |P_2|$ into $\mathcal{C}$ (This part is just to make sure $H$ remains regular). 
\EndIf

\State (In both cases) Arbitrarily extend the set $Q$ to obtain a bisection $(B, \overline{B})$ with $Q \subseteq B$ and add it to the collection $\mathcal{C}$.

\end{algorithmic}
\end{algorithm}

Notice that the collection $\CC$ consists of pairs of sets $(X,Y)$ with $|X| = |Y|$. Some of these pairs are bisections with $X \cup Y = V(H)$, while some are not. Broadly, the goal of the cut-player is to simultaneously ``make progress'' towards making any arbitrary set $S \subseteq V(H)$ with $|S| \leq \frac{V(H)}{2}$ expanding. This can be achieved in one of two ways: by either including in $\CC$ a cut that is unbalanced with respect to $S$, so that the matching player makes $S$ expanding in this round. Alternatively, this can be done by including in $\CC$ a cut that is balanced and sparse with respect to $S$: then one can use a potential analysis similar to~\cite{krv07} to show that this cannot keep happening, so that $S$ becomes expanding after only a few iterations. The next lemma formalizes this notion and shows that the collection $\CC$ indeed does satisfy such a property.

\begin{lemma}\label{lemma:collection}
The collection $\mathcal{C}$ output by the cut player satisfies one of the following two conditions.
\begin{enumerate}
 \item For every subset $S$ of $V(H)$ with $|S| \leq \frac{|V(H)|}{2}$ that is $\frac{1}{\rho^{101}}$-sparse in $H$, we have
 that at least one of two following statements holds.
 
 \begin{itemize}
     \item $\mathcal{C}$ includes a bisection $(B', \overline{B'})$ such that $B' \cup \overline{B'} = V(H)$ so that there is a subset $W \subseteq B'$ such that $W$ is balanced and sparse  with respect to $S$. More precisely, the cut $(W, \overline{W})$ in $H$ is such that (a) $|W \cap S|, |W' \cap S| = \Omega(|S|)$ and (b) $W$ is  $\frac{1}{100}$- sparse in $H[S]$. Concretely, $|\delta_{H[S]}(W)| \leq \frac{1}{100} \min(|W \cap S|, |W' \cap S|)$.
     \item There exists a pair of sets $(X,Y) \in \mathcal{C}$ such that $||X \cap S| - |Y \cap S|| = \Omega(|S|)$, so that any perfect matching of $X$ to $Y$ must add at least $\Omega(|S|)$ edges across $(S, V(H) \setminus S)$. 
     
 \end{itemize}

 \item $\mathcal{C}$ has exactly one pair of sets $(B, \overline{B})$ which is a bisection, and $B$ contains a set $Y' \subseteq V(H)$ with $|Y'| \leq \frac{|V(H)|}{2}$ such that $\frac{|Y' \cap Z|}{|Z|} \geq 0.6$ for every set of vertices $Z$ that has \emph{conductance} smaller than $\frac{1}{\rho^{100} \Delta^2}$ or $\frac{1}{3\rho^{100} \Delta}$, depending on whether we use the exponential time or near-linear time algorithm, where $\Delta$ is the degree of the current (regular) graph $H$.

\end{enumerate}

\end{lemma}

\begin{proof}
If~\Cref{lemma:bal_sep} or~\Cref{thm:bal_sep} returns a set $Y'$ as in above, then note that this set overlaps by at least $60\%$ every set $Z$ with conductance at most $\frac{1}{3\rho^{100} \Delta}$ if we use the exponential time algorithm or every set $Z$ with conductance at most $\frac{1}{\rho^{100} \Delta^2}$ if we use the near-linear time algorithm, and we are done as this satisfies condition $2$.

Otherwise, in both cases, we obtain a cut $W$ which is $\Omega(1)$ balanced and has conductance at most $\frac{1}{100 \rho^3\Delta}$. Since the maximum degree in the graph $H$ is $\Delta$, it must be the case that this cut has sparsity at most $\frac{1}{100\rho^3}$.

Recall that in this case, the cut player uses the \IC{} procedure to obtain another cut $(Q, \overline{Q})$.
Now if $|Q \cap S| \geq 0.6|S|$, it follows that the bisection $(B, \overline{B})$ with $Q \subseteq B$ satisfies  $||B \cap S| - |\overline{B} \cap S|| \geq \Omega(|S|)$, and hence we satisfy condition $1$.

Similarly, suppose that $|\overline{Q} \cap S| \geq (1 - \frac{\alpha}{3})|S|$. Notice that since $|Q| = \alpha n$ and $|\overline{Q}| = (1 - \alpha)n$, the number of $R \in \mathcal{R}$ is at most $\frac{1 - \alpha}{\alpha} + 1 = \frac{1}{\alpha}$. Then it follows that there exists an $R \in \mathcal{R}$ such that $|R \cap S| \geq {(1 - \frac{\alpha}{3})} \alpha |S| > \frac{\alpha}{2} |S|$. Since in this case $|Q \cap S| \leq \frac{\alpha}{3}|S|$  it follows that the pair of sets $(Q,R)$ satisfies condition $1$.

 Finally, if none of the above happens, it must be the case that both $|Q \cap S|, |\overline{Q} \cap S| = \Omega(|S|)$. Thus the cut $(Q, \overline{Q})$ is both balanced and sparse inside $S$ since \IC{} guarantees that $|\delta_{H[S]}(Q \cap S)| \leq \frac{|S|}{100\rho} \leq \frac{1}{100} \min(|W \cap S|, |W' \cap S|)$.  In this case, the bisection $(B, \overline{B})$ added by the cut player with $Q \subseteq B$ satisfies condition $1$.\end{proof}

 Broadly, our next lemma shows that if we keep finding bisections $(B, \overline{B})$ that contain such balanced sparse cuts $(Q, \overline{Q})$ inside $S$ for too many rounds, then the matching player must add a matching that makes the set $S$ $\Omega(\frac{1}{\rho^{101}})$-expanding in $H$. The proof is similar to the potential analysis in~\cite{krv07}.

\begin{lemma}\label{lemma:potential}
Let $S$ be a set of size $s$ in $H$. Then, if the cut player keeps finding bisections $(B, \overline{B})$ such that there is a cut $(Q,\overline{Q})$ with $Q \subseteq B$ in $H$ that is both balanced and $\frac{1}{100}$-sparse with respect to $S$ for more than $\OO(\log s)$ rounds, there must exist a round in which at least $\Omega(\frac{s}{\rho^{101}})$ edges were added by the matching player across $(S, V(H) \setminus S)$.
Formally, \IC{} can find a cut $(Q,\overline{Q})$ satisfying $|Q \cap S|, |\overline{Q} \cap S| = \Omega(s)$ and $|\delta_{H[S]}(Q)| \leq \frac{1}{100} \min(|Q \cap S|, |\overline{Q} \cap S|)$ for at most $\OO(\log s)$ rounds before the matching player adds a matching that makes $S$ $\Omega(\frac{1}{\rho^{101}})$-expanding.
\end{lemma}

\begin{proof}

As in~\cite{krv07}, we will consider a random walk. This time, however, the random walk will only involve vertices from $S$. For $u,v \in S$, let $p(u,v,t)$ denote the amount of mass which was initially at $u$ that is currently at $v$ after $t$ rounds of the cut-matching game. Initially, $p(u,v,0) = 1$ if and only if $u = v$, and $0$ otherwise. We define the random walk as follows: let $M_t$ be the matching  (corresponding to the bisection $(B, \overline{B})$) output by the matching player in round $t$. For every vertex $a \in S$ that is matched in $M_t$ to another vertex $b \in S$, we average the distributions of $a$ and $b$. That is, we set $p(u,a,t) = \frac{p(u,a,t-1) + p(u,b,t-1)}{2}$, and set $p(u,b,t) = \frac{p(u,a,t-1) + p(u,b,t-1)}{2}$ for every vertex $u \in S$.

We now define a potential function based on these probability distributions - it is simply the sum of the entropies of the distributions induced by each vertex $u \in S$. Formally, we define

$$\Phi(t) = \sum_{u \in S} \sum_{v \in S} -p(u,v, t) \log p(u,v,t)$$

Clearly, $\Phi(0) = 0$, and $\Phi(t) \leq s \log s$ for any $t$ since the entropy of a distribution on $s$ outcomes is at most $\log s$, and $\Phi(t)$ is the sum of $s$ such entropies. We now show that the potential increases by a lot in each round, provided that there is $Q \subseteq B$ such that the cut $(Q, \overline{Q})$ is $\Omega(1)$ balanced and $\frac{1}{100}$-sparse with respect to $S$.

\begin{lemma}
$\Phi(t) - \Phi(t-1) \geq \Omega(s)$ for any $t \geq 1$.

\end{lemma}

\begin{proof}
Fix a round $t$. Recall that the matching player, in this round, finds a perfect matching $M_t$, between vertices of the bisection $(B, \overline{B})$. In particular, $M_t$ matches vertices of $Q$ to some vertices of $\overline{Q}$, saturating $Q$. We may assume that $M_t$ matches at most $\frac{s}{\rho^{100}}$ vertices of $S$ to vertices outside $S$, otherwise $S$ is already $\Omega(\frac{1}{\rho^{101}})$-expanding after this round, and we are done. Also recall that $|Q \cap S| = |Q \cap \overline{S}| = \Omega(s)$. Without loss of generality we assume $|Q \cap S| \leq |\overline{Q} \cap S|$. Let $W^* = Q \cap S$, and let $w^* = |W^*|$. We know that the sparsity of the cut $W^*$ in the graph $H[S]$ is $\phi^{H[S]}(W^*) \leq \frac{1}{100}$.

For a vertex $u \in W^*$, let us denote the total mass that was initially on $u$ that is, at the beginning of round $t$, either (i) outside $W^*$  or (ii) On some vertex $v \in W^*$ which is matched to some vertex outside $S$ in $M_t$, by $q_u(t)$. Informally, this mass is the ``problematic mass'' - it does not help us increase the potential. We first claim that $\sum_{u \in W^*}{q_u(t)} \leq \frac{w^*}{200} + \frac{s}{\rho^{100}} $. This follows from the fact that $W^*$ is a $\frac{1}{100}$-sparse cut in $H[S]$. Since in every round, at most $\frac{1}{2}$ mass can go out of $W^*$ through each edge of $W^*$, the total mass that went outside $W^*$ in the rounds prior to $t$ is at most $\frac{w^*}{200}$. Also,  there are at most $\frac{s}{\rho^{100}}$ vertices of $W^*$ which are matched to some vertex outside $S$ by $M_t$. The total mass on these vertices is at most $\frac{s}{\rho^{100}}$, since the algorithm maintains the invariant that every vertex has total mass $1$. Therefore it must be the case that $\sum_{u \in W^*}{q_u(t)} \leq \frac{w^*}{200} + \frac{s}{\rho^{100}} \leq \frac{w^*}{100} $, where the last inequality follows from the fact that $w^* = \Omega(s)$.

By averaging, there exist at least $\frac{w^*}{2}$ vertices $u$ for which $q_u(t) \leq \frac{1}{25}$. Fix such a vertex $u$. We say that a vertex $v \in W^*$ is ``good'' for  $u$ if (a) $v$ is matched to some vertex $M_t(v) \in S$ and (b) $p(u, v) >  2p(u,M_t(v))$. We now show that the total mass of $u$ that is on the bad vertices must be small.

There are two types of ``bad'' vertices - the ones which violate condition (a) and the ones which violate condition (b). Let us denote by $B_1$ the set of vertices $v \in W^*$ matched to some vertex outside $S$ by $M_t$ , and by $B_2$ the set of vertices $v \in W^*$ which are matched to some vertex $M_t(v) \in S$ but having  $p(u,v) \leq 2p(u, M_t(v))$. Observe that $p(u,M_t(v))$ contributes to $q_u(t)$, since $M_t(v)$ is outside $W^*$. This means that the total mass on the vertices $B_1 \cup B_2$ can be bounded by $2q_u(t)$, simply from the definition of $q_u(t)$. The total mass of $u$ inside $W^*$ is at least $1 - q_u(t) \geq 1 - \frac{1}{25}$. Out of these, at most $2q_u(t) \leq \frac{2}{25}$ mass is on the ``bad'' vertices. This implies that the total mass on the ``good'' vertices is at least $1 - \frac{3}{25} = \frac{22}{25}$.

Finally, for every matching edge matching $a,b \in S$ such that $2p(u,a) > p(u,b)$, the entropy of the distribution of $u$ must increase at least by $\Omega(p(u,a))$. This is easy to prove but is already proved in~\cite{krv07} so we skip its proof here. Since the total mass on the good vertices is $\Omega(1)$, the entropy of the distribution for $u$ increases by $\Omega(1)$ after this round. Thus, the overall increase in potential is at least $\Omega(s)$ since there are at least $\Omega(s)$ vertices $u$ for which $q_u(t) \leq \frac{1}{25}$.\end{proof}

Since the potential increases by $\Omega(s)$ in each round and the total potential is at most $s\log s$, such a cut $(Q,\overline{Q})$ can only be found for $\OO(\log s)$ many rounds. \end{proof}
We are now in a position to prove Theorems~\ref{thm:cutplayer} and~\ref{thm:cutplayerpoly}.

\begin{proof}{Proof of \Cref{thm:cutplayer}, \Cref{thm:cutplayerpoly}}

The above analysis implies that after $\OO(\log s)$ rounds, for any set $S$ of size at most $s$ if \IC{} keeps returning cuts which are balanced and sparse with respect to $S$, then the matching player must add $\Omega(\frac{|S|}{\rho^{101}})$ edges across $(S, V(H) \setminus S)$. Then after these $\OO(\log s)$ many rounds, there are three cases for any such set $S$ based on the properties of the cut player guaranteed in~\Cref{lemma:collection}. We will assume that after these $\OO(\log s)$ rounds, $|\delta_H(S)| \leq \frac{|S|}{\rho^{101}}$, for otherwise we are done (for both~\Cref{thm:cutplayer} and~\Cref{thm:cutplayerpoly}) since $S$ is already expanding (recall that $\rho$ is a constant).

The first case is when in every round, the cut-player returns a cut which is balanced and sparse with respect to $S$. In this case, the matching player must have added $\Omega(\frac{|S|}{\rho^{101}})$ edges across $(S, V(H) \setminus S)$ in some round. Therefore if this happens, we are done, as after this round $S$ would have become $\Omega(\frac{1}{\rho^{101}})$-expanding.

The second case is when in some round, the collection $\mathcal{C}$ output by the cut player includes a pair of sets $(X,Y)$ with $||X \cap S| - |Y \cap S|| \geq \Omega(|S|)$. This again forces the matching player to add $\Omega(|S|)$ edges across $(S, V(H) \setminus S)$ in this round, which makes $S$ expanding. Finally, in the third case, the cut player could return a set $Y'$ which overlaps all $\frac{1}{\rho^{100} \Delta^2}$ (or $\frac{1}{3\rho^{100} \Delta}$ if exponential time) conductance cuts by more than $0.6$ fraction where $\Delta$ is the current degree of the graph $H$. We analyze this case in the next paragraph.

 First, we analyze the exponential time cut player: we will show that $S$ must be $\frac{1}{\rho^{101}}$-expanding after the last round in which the cut player finds the set $Y'$. Suppose $S$ is $\frac{1}{\rho^{101}}$-sparse just before this last round. Then its conductance is at most $\frac{1}{\rho^{101} \Delta}$. Recall that the set $Y'$ overlaps by 60\% every set $S$ which has conductance at most $\frac{1}{3\rho^{100} \Delta} \geq \frac{1}{\rho^{101} \Delta}$. Thus we must have $|Y' \cap S| \geq 0.6|S|$. But the cut player, in this round, outputs a bisection $(B, \overline{B})$ such that $Y' \subseteq B$. This means that $S$ becomes $\Omega(1)$-expanding after this round, and this concludes the proof of~\Cref{thm:cutplayer}.

 We proceed similarly for the near-linear time cut player. Suppose that $S$ is $\frac{1}{\rho^{101} \Delta}$-sparse before the last round. Then $S$ has conductance at most $\frac{1}{\rho^{101} \Delta^2}$. Recall that the set $Y'$ overlaps by 60\% every set $S$ which has conductance at most $\frac{1}{\rho^{100} \Delta^2} \geq \frac{1}{\rho^{101} \Delta^2}$. Thus we must have $|Y' \cap S| \geq 0.6|S|$. But the cut player, in this round, outputs a bisection $(B, \overline{B})$ such that $Y' \subseteq B$. This means that $S$ becomes $\Omega(1)$-expanding after this round. Since we never run the game for more than $\OO(\log s)$ rounds, and in each round the matching player can increase the degree of each vertex by only a constant, it follows that $\Delta = \OO(\log s)$, and therefore $S$ must be $\Omega(\frac{1}{\rho^{101} \log s})$-expanding. Finally, since $\rho$ is just a (large enough) constant which we can set accordingly using our analysis, we obtain our result in~\Cref{thm:cutplayerpoly}. This concludes the proof.\end{proof}

\subsection{Matching player: fast algorithms for \SC{} and \VSC}

Next, we describe the matching player strategy that helps us approximate sparse cuts. Here, we are given a graph $G$ together with a parameter $\phi'$, and a set of terminals $T$. The graph $H$ that the cut-matching game operates on now has the vertex set $V(H) = T$.

\begin{algorithm}\label{alg:matching}
\caption{Matching player strategy}
\begin{algorithmic}
\State Given a collection $\mathcal{C}$ of constantly many pairs of disjoint sets, for each pair $P_1, P_2 \subseteq V(H)$ with $|P_1| = |P_2|$ given by the cut player, find either
\State \begin{enumerate}
 \item a perfect matching $M$ of vertices of $P_1$ to $P_2$ that can be flow embed in $G$ with congestion $\frac{1}{\phi'}$ or 
 
 \item a $\phi'$-terminal sparse cut in $G$.

\end{enumerate}

\end{algorithmic}
\end{algorithm}

Next, we show that the matching player strategy can be implemented using a single max flow. This is very similar to the matching player in~\cite{krvactual}.

\begin{lemma}
The matching player strategy can be implemented using a single max flow call.
\end{lemma}

\begin{proof}
We setup a flow problem as follows. We create two new vertices $s,t$ and add edges from $s$ to each vertex of $P_1$ with capacity $1$. Similarly, we add edges from $t$ to each vertex of $P_2$ with capacity $1$. We retain each edge of $G$ with capacity $\frac{1}{\phi'}$. We now find a maximum flow between $s$ and $t$. For simplicity, we will assume that $\frac{1}{\phi'}$ is an integer (this assumption is also standard in previous cut-matching games).

If the maximum flow is less than $|P_1|$, the minimum cut in this graph will correspond to a $\phi'$-sparse cut in $G$. This can be seen as follows.  Suppose $(W, \overline{W})$ is the cut in $G$ corresponding to the $s-t$ min cut. First note that $W, \overline{W} \neq \emptyset$ as the min-cut is strictly has weight strictly less than $|P_1|$. Observe that in any max-flow the flow is saturated across the minimum cut. But the total flow out of the set $W$ is at most $|W \cap T|$. It follows that $|W \cap T| \geq \frac{|\delta_G(W)|}{\phi'}$. One can show similarly that $|\overline{W} \cap T| \geq \frac{|\delta_G(\overline{W})|}{\phi'}$. It follows that $W$ is a $\phi$-terminal sparse cut in $G$.

Otherwise, the flow path decomposition of the maximum flow  gives rise to a perfect matching $M$. We return this matching $M$. Using the almost-linear time max-flow algorithm of~\cite{chen22}, we can accomplish this in time $m^{1 + o(1)}$.\end{proof}

We now proceed to prove \Cref{thm:cutmatching}.

\begin{proof}[Proof of \Cref{thm:cutmatching}]
We run the cut-matching game with the matching player as described above, with parameter $s$, and the graph $H$ initialized as the empty graph with the vertex $V(H) = T$. If ever the matching player returns a $\phi'$-terminal sparse cut, then we are done. Otherwise, the cut player guarantees than in $\OO(\log s)$ rounds the graph $H$ is a $s$-small set $\Omega(\frac{1}{\log s})$-expander. But since each matching flow embeds in $G$ with congestion $\frac{1}{\phi'}$, it must be the case that $H$ itself flow embeds in $G$ with congestion $\frac{\OO(\log s)}{\phi'}$. Now consider a set $S$ with at most $s$-terminals in $G$, and let $S_T = S \cap T$. Since $S_T$ is $\Omega(\frac{1}{\log s})$-expanding in $H$, it means that $|\delta_H(S_T)|  = \Omega(\frac{|S_T|}{\log s})$. But the edges $\delta_H(S_T)$ flow embed in $G$ with congestion $\frac{\OO(\log s)}{\phi}$. This means that one can send a flow of value  $\Omega(\frac{|S_T|}{\log s})$ from $S$ to outside $S$ in $G$ with congestion $\frac{\OO(\log s)}{\phi}$. This implies that $S$ must be $\Omega(\frac{\phi}{\log^2 s})$-terminal expanding in $G$.\end{proof}

Next, we prove our result for the vertex terminal sparse cuts using a simple modification to the matching player. The matching player strategy is almost exactly the same as before, except that for every pair $P_1, P_2$ given the cut player we want to find a perfect matching $M$ vertices $P_1$ to $P_2$ that embeds in $G$ with \emph{vertex congestion} $\frac{1}{\phi'}$ or a $\phi'$-(vertex) terminal sparse cut. Accordingly, we change the flow problem so that every \emph{vertex} has capacity $\frac{1}{\phi'}$. 

If the maximum flow is less than $|P_1|$, the minimum (vertex) cut in this graph will correspond to a $\phi'$-terminal sparse vertex cut in $G$. This can be seen as follows. Suppose $(L',C',R')$ is the vertex min-cut in $G$ corresponding to the $s-t$ vertex min cut. First note that $L', R' \neq \emptyset$ as the min-cut has weight strictly less than $|P_1|$. Observe that in any max-flow the flow is saturated across the minimum cut. But the total flow out of the set $L' \cup C'$ is at most $|(L' \cup C') \cap T|$. It follows that $|(L' \cup C') \cap T| \geq \frac{|C'|}{\phi'}$. One can show similarly that $|(R' \cup C') \cap T| \geq \frac{|C'|}{\phi'}$. This shows that $(L',C',R')$ is $\phi'$- terminal sparse.

Otherwise, the flow path decomposition of the maximum flow gives rise to a perfect matching $M$. We return this matching $M$. Again using the almost-linear time max-flow algorithm of~\cite{chen22} we can accomplish this in time $m^{1 + o(1)}$.

\begin{proof}[Proof of~\Cref{thm:cutmatchingvertex}]
We run the cut-matching game with the matching player as described above and the parameter $2s$ instead of $s$. Again, the graph $H$ is initialized as the empty graph with the vertex $V(H) = T$. If ever the matching player outputs a $\phi'$-terminal sparse vertex cut, we are done. Otherwise, after $\OO(\log s)$ rounds, the cut-matching game certifies that $H$ is a $2s$-small set $\frac{1}{\log s}$-expander, and $H$ embeds in $G$ with vertex congestion $\OO(\frac{\log s}{\phi'})$.

Now consider a vertex cut $(L,C,R)$ in $G$, with $|L \cap T| = s$. Suppose this cut is not $\Omega(\frac{\phi'}{\log^2 s})$-terminal expanding. In particular, we have $|C \cap T| \leq |C| \leq \OO(\frac{\phi'}{\log^2 s})|(L \cup C) \cap T| \leq \frac{1}{2} |(L \cup C) \cap T|$ since $\phi' \leq \frac{1}{2}$. This gives $|C \cap T| \leq |L \cap T|$ and hence $|(L \cup C) \cap T| \leq 2s$. By the guarantee of the cut-matching game, $(L \cup C) \cap T$ must be $\Omega(\frac{1}{\log s})$-edge expanding in $H$. This means $|\delta_H((L \cup C) \cap T)| \geq \Omega(\frac{s}{\log s})$. But the edges $\delta_H((L \cup C) \cap T)$ embed in $G$ with vertex congestion $\OO(\frac{\log s}{\phi'})$. This means that we can send $|\delta_H((L \cup C) \cap T)| = \Omega(\frac{s}{\log s})$ units of flow from $L \cup C$ to the rest of the graph with vertex congestion only $\OO(\frac{\log s}{\phi'})$, and this in turn means that ${|C|} \geq \OO(\frac{\phi'}{\log^2 s})|(L \cup C) \cap T|$, which shows that $(L,C,R)$ must be $\Omega(\frac{\phi'}{\log^2 s})$-terminal expanding, which is a contradiction.\end{proof}

We are now ready to prove~\Cref{thm:cutmatchingedgelogk} and~\Cref{thm:cutmatchingvertexlogk}.

\begin{proof}[Proof of~\Cref{thm:cutmatchingedgelogk}]
Suppose that there is a $\phi = \frac{k}{s}$-sparse cut $S$ of size $k$ that separates $s$ vertices. We first guess $\phi$, $\log k$ upto factor $2$. Let $\phi' = D\phi\log^2 k$ for some large constant $D$ to be fixed later. Next, construct a $(\epsilon, D\phi'\log^2 k)$ sample set $T$ with $\epsilon = \frac{1}{100}$, and a constant $D$ which will be chosen later. Note that $T$ has size $ \min\{n,\Theta(n\phi')\} = \min\{n, \lambda n\phi'$\} for some constant $\lambda$. We assume that $\lambda nD\phi' \log^2 k < n$, for otherwise we must have $s = \OO(k\log^2 k)$ and we can simply run our cut-matching game with the terminal set as the entire vertex set, so that~\Cref{thm:cutmatching} will give us an $\OO(\log^2 s)$-approximation which is already an $\OO(\log^2 k)$-approximation. It follows that $|S \cap T| = s' =  \Theta(\lambda Dk \log^2 k)$. Thus $S$ has terminal sparsity $\OO(\frac{1}{\lambda D\log^2 k})$. Running the cut-matching game algorithm of~\Cref{thm:cutmatching} then gives us an $\OO(\frac{1}{\lambda D\log^2 k} \log^2(s')) = \OO(\frac{1}{\lambda D})$-terminal sparse cut, call it $S'$. Choose $D$ large enough so that $S'$ is $\frac{1}{10\lambda} \leq \frac{n}{10|T|}\phi'$-terminal sparse, which in turn means that $S'$ must satisfy the sample set condition. By the guarantee of sample sets, it follows that $|S'| = \Omega(\frac{n}{|T|} |S' \cap T|)$. Since $S'$ is $\frac{n}{10|T|}\phi'$-terminal sparse, it follows that this cut must be $\OO(\phi') = \OO(\phi \log^2 k)$-sparse.\end{proof}

\begin{proof}[Proof of~\Cref{thm:cutmatchingvertexlogk}]
Suppose that there is a $\phi = \frac{k}{k + s}$-sparse cut $(L, C, R)$ with $|R| \geq |L| = s$ and $|C| = k$ of size $k$ that separates $s$ vertices. Let $\epsilon = \frac{1}{100}$. We first guess $\phi$, $k$, $s$ upto factor $2$. Let $\phi' =  D\phi (\log^2 k + \log^2 \log n\phi)$ for some large constant $D$ which will be chosen later. Next, construct a $(\epsilon, \phi')$ (vertex) sample set $T$ with $\epsilon = \frac{1}{100}$  (Notice that this can be done even when we guess $s,\phi'$ upto factor $2$, we skip this minor detail). Note that $T$ has size $\min\{n,\OO(n\phi' \log n\phi')\} = \min\{n, \lambda n\phi' \log {n\phi'}\}$, say. We assume that $\lambda n\phi' \log {n\phi'} < n$, for otherwise we have $s =  \OO(k(\log^2 k + \log^2 \log{n\phi})\log {n\phi'})$, and by running our cut-matching game with the vertex set as the terminal set,~\Cref{thm:cutmatchingvertex} would give us an $\OO(\log^2 s) = \OO(\log^2 k + \log^2 \log {n}{\phi})$-approximation. Therefore it follows that $|L \cap T| = s' =  \Theta(\lambda kD(\log^2 k + \log^2 \log n\phi) \log n\phi')$. Thus $S$ has terminal sparsity $\OO(\frac{1}{\lambda D(\log^2 k + \log^2 \log n\phi)\log n\phi'})$.

Now we run the algorithm from~\Cref{thm:cutmatchingvertex} to obtain another vertex cut $(L',C',R')$ with $|R' \cap T| \geq |L' \cap T|$ which is $\psi$-terminal sparse where $$\psi = \OO(\frac{1}{\lambda D(\log^2 k + \log^2 \log {n\phi})\log n\phi'} \log^2 |L \cap T|) = \OO(\frac{1}{\lambda D \log n\phi'}).$$ Choose $\lambda$ large enough so that $\psi \leq \frac{\epsilon^2}{10D\log(n\phi')} \leq \frac{n}{10|T|}\phi'$. Let $s' = |L' \cap T|$ and $s'' = \frac{n}{|T|}s'$. 

There are two cases. Fix a constant $\alpha > 1$ which will be chosen later. First, suppose no component in $G \setminus C'$ has size $>n - \frac{s''}{\alpha}$. Then it is clear that we can partition $V(G) \setminus C'$ into $A \cup B$ such that both $A$ and $B$ have $\Omega(s'')$ vertices.

Now suppose there is indeed a component $X$ which has size $> n - \frac{s''}{\alpha}$. We will obtain a contradiction. Clearly, it must be the case that $G \setminus (C' \cup X)$ has size $< \frac{s''}{\alpha}$. Also either $L' \subseteq V(G) \setminus (C' \cup X)$ or $R' \subseteq V(G) \setminus (C' \cup X)$, and hence $V(G) \setminus (C' \cup X)$ contains at least $s'$ terminals. But then since $G \setminus (C' \cup X)$ is the set of all but one component after removing a vertex cut, and further $G \setminus (C' \cup X)$ is $\psi \leq \frac{\epsilon^2}{10d} \leq \frac{n}{10|T|}\phi'$-terminal sparse, $G \setminus (C' \cup X)$ must satisfy the sample condition. But this means $|V(G) \setminus (C' \cup X) \cap T| \leq 2 \frac{|T|}{n}\frac{s''}{\alpha} \leq \frac{2s'}{\alpha}$, which is a contradiction for $\alpha > 2$.

Thus we must in fact have that every component in $G \setminus C'$ has size $> n - \frac{s''}{\alpha}$, and hence we obtain a partition of $V(G) \setminus C'$ into $A \cup B$ such that both $A$ and $B$ have $\Omega(s'')$ vertices. It follows that since the vertex cut $(A,C,B)$ is $\OO(\frac{n}{10|T|}\phi')$-terminal sparse, it must be $\OO(\phi')$-sparse. Hence we obtain an $\OO(\frac{\phi'}{\phi}) = \OO(\log^2 k + \log^2 \log {n\phi})$-approximation.\end{proof}

\section{\SSVE}\label{sec:ssve}

In this section we study the approximability of \SSVE{} and prove~\Cref{thm:DKS,thm:DSH_SSVE,thm:fpt}.

\subsection{Hardness}
We start by showing two hardness results for \SSVE. The first is a reduction from \DKS{} to \SSVE{} which shows $n^\epsilon$ hardness for some $\epsilon > 0$ assuming standard hardness results for \DKS{}. The second is a reduction from \DSH{}, a generalization of \DKS, which shows a strong hardness of approximation result, that \SSVE{} is hard to approximate beyond a factor $2^{o(k)}$ even using FPT algorithms, that is, algorithms running in time $f(k)n^{\OO(1)}$. This is in sharp contrast to \SSE, for which we obtained an $\OO(\log k)$-approximation (\Cref{thm:logksse}).

\begin{proof}[Proof of~\Cref{thm:DKS}]
Given an instance $(G,k)$ of \DKS{}, consider the graph $H$ with vertex set $L \cup R$. The set $R$ contains a vertex $x_w$ for every vertex $w \in V(G)$, and the set $L$ contains a vertex $y_{uv}$ for every edge $e = \{u,v\} \in E(G)$. For every vertex $y_{uv}$, add the edges $(x_u, y_{uv})$ and $(x_v, y_{uv})$ to $E(H)$. Finally, add a clique on the vertices $R$ to $E(H)$.

First, suppose that the \DKS{} objective is $\geq \ell$. Let $R' \subseteq R$ be the vertex set corresponding to the \DKS{} solution in $G$, and $L'$ be the set corresponding to the edges inside $G[R']$. Observe that $|L'| \geq \ell$. Clearly, $(L', R', V(H) \setminus (L' \cup R'))$ is a vertex cut of size $k$ in $H$ which separates $\ell$ vertices.

Now suppose there is a vertex cut in $H$ cutting $k\tau$ vertices which separates at least $\ell \tau$ vertices for some $\tau > 1$. Observe that we may assume that such a cut only cuts vertices of $R$, since the vertices of $R$ form a clique. Let $L' \subseteq L$ be the set of vertices separated. Sample each vertex in the cut with probability $\frac{1}{2\tau}$ into a set $W \subseteq R$. It follows that the expected number of vertices separated from $L$ after deleting $W$ is at least $\frac{\ell}{4\tau^2}$ since each edge has a $\frac{1}{4\tau^2}$ chance that both its endpoints are in $W$. Using standard concentration inequalities, it must be the case that with constant probability, we obtain a cut in $H$ cutting $\leq k$ right vertices that separates $\Omega(\frac{\ell}{\tau})$ left vertices.

We now show that if we have a $(\beta, \gamma)$ bi-criteria approximation algorithm for \SSVE{}, then we obtain an $\OO(\beta^2\gamma)$ randomized approximation algorithm for \DKS{}. First, construct the graph $H$ as above. Guess the optimal value $opt$ for the \DKS{} instance. One can assume $opt = \Omega(k)$, since one can always find $k$ vertices that induce at least $k/2$ edges. 
Run the $(\beta,\gamma)$-approximation to \SSVE{} with set size $opt$. If it returns a set with size $< opt$, re-run the algorithm till the total number of vertices separated is at least $opt$. Since in each iteration the obtained cut must have sparsity at most $\beta \cdot \frac{k}{k + opt}$, it follows that we obtain a set of size $\OO(\tau \cdot opt)$ which is separated from the rest of the graph by a cut of size at most $\OO(\beta k\tau)$ for some $\gamma \geq \tau > 1$. We now use random sampling as in above to obtain a set of size $k$ vertices that separates $\Omega(\frac{opt}{\beta^2\tau})$ vertices. It follows that we obtain an $\OO(\beta^2 \tau \leq \beta^2 \gamma)$-approximation for \DKS. Setting $\gamma = \OO(1)$ and $\beta = \sqrt{\alpha}$, we obtain the hardness for \SSVE.\end{proof}

We next obtain a strong hardness of approximation in terms of $k$ based on \DSH{} defined below. 

\begin{definition}[\DSH{}]
Given a bipartite graph $G$ with vertex set $L \cup R$ as the bipartition, find a set $S$ of $k$ vertices in $R$ such that maximum number of vertices in $L$ have neighbours only in $S$.   
\end{definition}

Notice that this definition is equivalent to the standard definition in terms of hyperedges: Each vertex of $L$ represents a hyperedge, and each vertex of $R$ represents a hypergraph vertex. Let an $(\alpha, \beta)$-approximation algorithm for \DSH{} be the one that returns $S \subseteq R$ with $|S| \leq \beta k$ that contains the neighbor set of at least $(1/\alpha) \cdot (|S|/k) \cdot opt$ left vertices (hyperedges), so the ratio between the number of hyperedges and the number of vertices is at least $(1/\alpha)$ times that of the optimal solution.

\begin{lemma}\label{lemma:DSH}
An $f(k) \poly(n)$ time 
$(\alpha, \beta)$-approximation algorithm for \SSVE{} implies an 
$f(k) \poly(n)$ time $(\frac{\alpha }{1 - (\alpha - 1)(k/s)}, \frac{\alpha \beta}{1 - (\alpha - 1)(k/s)})$-approximation for \textsc{Densest} $k$-\textsc{Sub}-\textsc{Hypergraph} for the instance when the optimal has $k$ vertices inducing $s$ hyperedges. 

\end{lemma}

\begin{proof}
The reduction is almost the same as the reduction from \DKS{}. Given an instance ($G$, $k$, $L$, $R$) of \DSH{} with vertex set $L \cup R$, add a clique among the vertices $R$. Call the resulting graph $H$.

Suppose there is a \DSH{} solution that separates $s$ vertices. It is clear that $H$ has a vertex cut $(L', C', R')$ with $|R'| \geq |L'| \geq s$ and $|C'| = k$. Now suppose there is an algorithm that finds a vertex cut $(L'',C'',R'')$ in $H$ with $|R''| \geq |L''|$,  $|L''| \leq \beta s$, and $\frac{|C''|}{|L''| + |C''|} \leq \frac{\alpha \cdot k}{k + s}$.
Again, it is clear that we may assume that $C''$ only cuts vertices of $R$, and that $L'' \subseteq L$. 
Since $|C''| \leq \frac{\alpha k}{k + s} |L''| / (1 -  \frac{\alpha k}{k + s}) = \frac{\alpha k}{s - (\alpha - 1)k} |L''| \leq k \cdot \frac{\alpha \beta}{1 - (\alpha - 1)(k/s)}$, we obtain a $(\frac{\alpha}{1 - (\alpha - 1)(k/s)}, \frac{\alpha \beta}{1 - (\alpha - 1)(k/s)})$-approximation for \DSH.\end{proof}

Therefore, when $\alpha k = o(s)$, an $(\alpha, \beta)$-approximation for \SSVE{} implies a $(\alpha(1 + o(1)), \alpha \beta(1 + o(1)))$-approximation for \DSH. 

Manurangsi, Rubinstein, and Schramm~\cite{mrs20} introduced the Strongish Planted Clique Hypothesis (SPCH) and proved that assuming the SPCH, there is no $f(k) \poly(n)$-time algorithm for \DSH{} that finds $k$ vertices that induce at least $(1/\alpha(k)) opt$ hyperedges, for any function $\alpha(k) = 2^{o(k)}$. We slightly generalize their result to show that there is no $f(k) \poly(n)$-time $(\alpha(k), \beta(k))$-approximation algorithm for \DSH{} for any $\alpha(k) = 2^{o(k)}, \beta(k)$, and $f(k)$.

To formally state the SPCH and our result, let $\mathcal{G}(n,1/2)$ be the distribution of an Erd\H{o}s-Renyi random graph, and $\mathcal{G}(n,1/2, n^{\delta})$ be the distribution of a graph where a $n^{\delta}$-size random clique is planted in an Erd\H{o}s-Renyi random graph. 

\begin{hypothesis} [Strongish Planted Clique Hypothesis]
There exists a constant $\delta \in (0, 1/2)$ such that no $n^{o(\log n)}$-time algorithm $\mathcal{A}$ can satisfy both of the following.
\begin{itemize}
    \item (Completeness) $\Pr_{G \sim \mathcal{G}(n, 1/2, n^{\delta})}[\mathcal{A}(G) = 1] \geq 2/3.$
    \item (Soundness) $\Pr_{G \sim \mathcal{G}(n, 1/2)}[\mathcal{A}(G) = 1] \leq 1/3.$
\end{itemize}
\end{hypothesis}

\begin{theorem}\label{thm:DSH}
Assuming the SPCH, for any function $f(k), \beta(k)$, and $\alpha(k) = 2^{o(k)}$, there is no $f(k) \cdot \poly(n)$-time $(\alpha(k), \beta(k))$-approximation algorithm for \DSH{} when $s = 2^k$. 
\end{theorem}

Since $k \cdot \alpha(k) = o(s)$, together with \Cref{lemma:DSH}, it rules out $(\alpha'(k), \beta'(k))$-approximation for \SSVE{} for any $\alpha'(k) = 2^{o(k)}$ and $\beta'(k)$ as well, finishing the proof of \Cref{thm:DSH_SSVE}. 

\begin{proof}
Given a value of $k$ (as a parameter growing arbitrarily slower than $n$) and $\delta$,
we give a reduction from SPCH to \DSH.
Given an instance $G = (V_G, E_G)$ of SPCH with $n$ vertices, our reduction produces a hypergraph $H = (V_H, E_H)$.
Our reduction is parameterized by $\ell, r, N$, where $r = r(k) = o(k)$ and $\ell = o(\log n / r)$ will be determined later (so $\ell$ will be greater than any function of $k$), and $N = 100kn^{(1 - \delta)\ell}$. 
The hypergraph $H$ will be $r$-uniform. 
The reduction is as follows. 

    \begin{itemize}
        \item $|V_H| = N$. Each vertex $s \in V_H$ corresponds to a random subset of $V_G$ where each $v \in V_G$ is independently included with probability $\ell / n$. 

        \item A $r$-set $\{ s_1, \dots, s_r \} \subseteq V_H$ becomes a hyperedge if $\cup_{i \in [r]} s_i$ forms a clique in $G$. 

        \item We will let $N = 100kn^{(1-\delta)\ell}$. 

        \item Reduction time: at most $N^r = n^{O(\ell r)}$. 
        
    \end{itemize}

    For completeness, if $G$ contains a clique $C$ of size $n^{\delta}$, one vertex $s_i$ of $H$, is a subset of $C$ with probability $n^{(\delta-1) \ell}$. The expected number of such vertices is $100k$, so with a probability at least $0.9$, the number of such vertices is at least $k$. In $H$, these $k$ vertices induce $\binom{k}{r}$ hyperedges by construction.

    For soundness, assume that there exists $S \subseteq V_H$ such that 
    $t := |S| \leq \beta k$ and $S$ induces at least $(1/\alpha) (|S|/k) \binom{k}{r}$ hyperedges, where $\alpha = \alpha(k)$ and $\beta = \beta(k)$ are the functions specified in the theorem. 
    We will show that such $S$ cannot exist with high probability. 

    Consider a graph $G'$ whose vertex set is $V_H$ and $(s, s')$ is an edge if they belong to a same hyperedge of $H$. 
We want to say that each vertex in $s \in S$ has a large degree in $G'[S]$. (This is an analogue of Observation 6 of~\cite{mrs20}.)

\begin{claim}
There exists $S' \subseteq S$ such that the degree of every $s \in S'$ in $G'[S']$ is at least $\frac{k}{2(\alpha k)^{1/r}}$. 
\end{claim}
\begin{proof}
        Initially, the {\em hypergraph density} of $H[S]$, which is simply the number of hyperedges divided by the number of vertices, is at least $(1/\alpha) (|S|/k) \binom{k}{r} / |S| \geq \binom{k}{r}/(\alpha k)$. Greedily, if removing $s \in S$ (and all its incident hyperedges) does not decrease the hypergraph density, do it until there is no such vertex.

        Now we want to show that every $s \in S$ has a large degree. Consider $s \in S$ whose degree in $G'[S]$ is $d$. Then, the number of hyperedges of $H[S]$ containing $s$ is at most $\binom{d}{r - 1} \leq \binom{d}{r}$.
        (Note that a hyperedge induces a $r$-clique in $G'$ by construction.)         
        Since removing $s$ from $S$ will decrease the hypergraph density we have,
        \begin{align*}
        &\ \frac{\binom{k}{r}}{\alpha k} < \binom{d}{r} \\
        \Rightarrow &\ \frac{(k/2)^r}{\alpha k} < d^r 
        \qquad \mbox{(assuming }r \leq k/10 \mbox{ that will be ensured later)} 
        \\
        \Rightarrow &\ d > \frac{k}{2
        (\alpha k)^{1/r}},
        \end{align*}
        which proves the claim.\end{proof}
        
For notational simplicity, redefine $S$ to be the set satisfying $|S| \leq \beta k$ and $\deg_{G'[S]}(s) \geq \frac{k}{2(\alpha k)^{1/r}}$ for every $s \in S$. Now, we show that in the graph $G$, the vertices in $(\cup_{s \in S} s)$ must have many edges. In order to show this, we use the following lemma~\cite{mrs20} directly.
        \begin{lemma}[Lemma 7 of~\cite{mrs20}]
            If $N \leq 1000\ell n^{(1-\delta)\ell}$ and $20 \leq \ell$, for any $M \subseteq V_H$ with $|M| \leq n^{0.99\delta} / \ell$, we have $|\cup_{s \in M} s| \geq 0.01\delta|M|\ell$.
            \label{lemma:mrs_lem7}
        \end{lemma} 
        (Note that this is only over the  randomness of selecting $V_H$ and not the randomness of $G$.)

Let $T = \cup_{s \in S} s \subseteq V_G$, and consider the set of edges that should appear in $G[T]$. For each $v \in T$, fix an arbitrary $s \in S$ that contains $v$, and let $s_1, \dots, s_d$ be the neighbors of $s$ in $G'[S]$ (so that $d \geq k/(2(\alpha k)^{1/r})$). 
        By Lemma~\ref{lemma:mrs_lem7}, $| \cup_{i \in [d]} s_i | \geq 0.01 \delta \ell k/(2(\alpha k)^{1/r})$, which implies that the degree of $v$ in $G[T]$ is at least $0.01 \delta \ell k/(2(\alpha k)^{1/r})$. 
        Therefore, we can conclude that the total number of edges that must appear in $G[T]$ is at least $0.01 |T| \delta \ell k / (4(\alpha k)^{1/r})$.
        The probability of all these edges appearing is $2^{-0.01 |T| \delta \ell k / (4(\alpha k)^{1/r})}
        \leq 2^{-0.0001 \delta^2 |S| \ell^2 k / (4(\alpha k)^{1/r})}$, again using Lemma~\ref{lemma:mrs_lem7} for $T$. 

        We union bound over all choices of $S$ ($N^{|S|}$ choices for a given size) and all choices for $G'[S]$ ($2^{|S|^2} \leq N^{|S|}$ choices since $N$ is larger than any function of $|S|$.
        \begin{align*}
        & \sum_{p=1}^{\beta k}
        N^{2p} \cdot 
        2^{-0.0001 \delta^2 p \ell^2 k / (4(\alpha k)^{1/r})} \\
        \leq & \sum_{p=1}^{\beta k}
        \bigg( N^{2} \cdot 
        2^{-0.0001 \delta^2  \ell^2 k / (4(\alpha k)^{1/r})} \bigg)^p \\
        = & \sum_{p=1}^{\beta k}
        \bigg( 2^{2(1-\delta)\ell \log (100kn) -0.0001 \delta^2  \ell^2 k / (4(\alpha k)^{1/r})} \bigg)^p \\
        = & \sum_{p=1}^{\beta k}
        \bigg( 2^{2(1-\delta) \log (100kn) -0.0001 \delta^2  \ell k / (4(\alpha k)^{1/r})} \bigg)^{\ell p}.
        \end{align*}

        We set the parameters $\ell$ and $r$ depending on $k, \alpha, \beta, f, \delta$. As long as $\ell k / (4(\alpha k)^{1/r}) > 1000000 \log n / \delta^2$, the above probability becomes at most $0.01$. Choose $\ell = (1000000/\delta^2)\log n / (r \cdot g(k)),$ where $g$ is a (slowly) growing function to be determined shortly, which implies that 
    \[
    \ell k / (4(\alpha k)^{1/r}) > 1000000 \log n / \delta^2 \Leftrightarrow 
    k/(4(\alpha k)^{1/r}) > rg(k)
    \Leftrightarrow 
    \alpha k < (k/(4rg(k))^r. 
    \]
    Letting $r = k / (g(k))$ will ensure that as long as $\alpha k < 2^{k/(g(k))}$, the desired set $S$ does not arise from $G(n, 1/2)$ with probability at least $0.99$. Since $\alpha = 2^{o(k)}$, one can always choose a growing function $g$ satisfying this. 

    Therefore, an $(\alpha, \beta)$-approximation $\mathcal{A}$ for \DSH{} 
    implies an algorithm that correctly decides whether a given graph is a random graph or a planted graph with probability at least $0.9$. 
    In terms of running time, since $H$ has size at most $O(N^{r})$, if $\mathcal{A}$ runs in time $f(k) \cdot N^{c \cdot r}$ for some absolute constant $c$, then its running time in terms of planted clique is 
    \[
    f(k) n^{(1 - \delta)\ell r}
    \leq 
    f(k) n^{1000000((1 - \delta)/\delta^2)\log n / g(k)},
    \]
    so for large enough $k$, the existence of $\mathcal{A}$ refutes the SPCH.\end{proof}

\subsection{Algorithm}

For our next result, we prove~\Cref{thm:fpt} which shows that we can match this lower bound for \SSVE{} using the treewidth reduction technique of Marx et al.~\cite{msr13}, under a mild assumption. Previously, this technique was used by van Bevern et al.~\cite{bamo15} to show the fixed-parameter tractability of vertex bisection when the number of connected components in the optimal bisection is a constant.

For our result, we will need the following theorem from~\cite{msr13}.

\begin{theorem}[Treewidth reduction, see Lemma 2.11 from ~\cite{msr13}]\label{thm:twred}
Given a graph $G$ and two vertices $s$ and $t$, one can in time $\OO(f(k)(n + m))$ find a set of vertices $U$, such that the graph $H$ obtained from $G$ by contracting each connected component of $V(G) \setminus U$ has treewidth at most $2^{\OO(k^2)}$. Further, $U$ contains every inclusion-wise minimal $s$-$t$ separator of size at most $k$. 
\end{theorem}

\begin{proof}[Proof of~\Cref{thm:fpt}]

There are two cases. For the first, assume that every component in $G \setminus C$ has size at most $\frac{3n}{4}$. In this case, $C$ is a balanced separator, and we can simply find a balanced separator of size $k$ using directly the result of~\cite{fm06} in time $2^{\OO(k)} n^{\OO(1)}$.

For the second case, when at least one component in $G \setminus C$ has size at least $\frac{3n}{4}$ - 
notice that without loss of generality, we can assume that $R$ is connected, for otherwise we can move all except the largest component in $R$ to $L$, and obtain a cut which is sparser, contradicting the maximality of $(L,C,R)$.

Group the connected components $C_1, C_2....C_{\ell}$ of $G[L]$ by their neighbourhood in $C$. Observe that the number of groups is at most $2^k$. It follows that there exists at least one group of components, say $\AA$ so that $|\bigcup_{A \in \mathcal{A}} A| \geq \frac{|L|}{2^k}$ and each $A \in \AA$ has the same neighbourhood in $C$. Consider an arbitrary component $A \in \AA$. Let $s$ be a vertex of $A$ and $t$ be a vertex of $R$. Consider the inclusion-wise minimal separator $T \subseteq C$ that separates $s$ from $t$. Clearly, $T$ must separate $A$ from $R$. Since every component in $\AA$ has the same neighbourhood in $C$, it follows that $T$ in fact separates each component of $\AA$ from $R$.

We now proceed as follows. First, we guess the vertices $s \in A$ and $t \in R$. Next, we invoke \Cref{thm:twred} to obtain the set $U$ and the graph $H$ with treewidth at most $2^{\OO(k^2)}$. We assign a weight to each vertex of $H$: every vertex of $U$ has weight $1$, and every vertex corresponding to a connected component of $G \setminus U$ has weight equal to the number of vertices in that component. It is now clear that the problem reduces to finding a set with weight at least $\frac{|L|}{2^k}$ and at most $|L|$, that is separated from the rest of the graph by a separator of weight $k$ in $H$. Using Theorem 4 from~\cite{bamo15}, this can be done in time $2^{\OO(\textbf{tw}(H))} \poly(n)$, completing the proof. \end{proof}

Note that~\Cref{thm:fpt} also gives a $2^k$-approximation in time $2^{2^{\OO(k^2)}}\poly(n)$ for \SSVE, when we are promised that the optimal cut is maximal. We remark that maximality is a natural assumption when we want to find the sparsest cut in the graph.

\section{Applications}\label{sec:applications}

\subsection{Multicut mimicking networks}\label{sec:mmn}

In this section, we show an improvement in the size of the best multicut mimicking network that can be computed in polynomial time. This follows directly from the reduction in~\cite{wahlstrom22} to \SSE{}. We start by defining multicut mimicking networks.

\begin{definition}
Given a graph $G$ and a set of terminals $T$, let $G'$ be another graph with $T \subseteq V(G')$. We say that $G'$ a multicut mimicking network for $G$ if for every partition $\mathcal{T} = T_1 \cup T_2 \cup \ldots T_s, s \geq 1$  of $T$, the size of a minimum multiway cut for $\mathcal{T}$ is the same in $G$ and $G'$. 
\end{definition}

In other words, a multicut mimicking network preserves the size of the minimum multiway cut for every partition of the terminal set (We recall that a multiway cut for a terminal set is a cut that separates all terminals from each other). Let us denote the total number of edges incident on vertices of $T$ by $cap_G(T)$, and let $cap_G(T) = k$.

\begin{theorem}\cite{wahlstrom22}]\label{thm:mmn_sse}
Let $A$ be an approximation algorithm for \SSE{} with a bi-criteria approximation
ratio of $(\alpha(n,k), \OO(1))$, where $n$ is the number of vertices and $k$ is the number of edges cut in the optimal cut. Let $(G, T)$ be a terminal
network with $cap_G(T) = k$. Then there is a multicut mimicking network for $G$ of size $k^{\OO(\alpha(n,k) \log k)}$ and it can be found in randomized polynomial time.
\end{theorem}

Since at the time no $\polylog(k)$-approximation for \SSE{} was known,~\Cref{thm:mmn_sse} by itself could only be used to show the existence of a multicut mimicking network.~\cite{wahlstrom22} however used the approximation algorithm for \SSE{} by Bansal et al.~\cite{minmax} who gave an approximation of $\OO(\sqrt{\log n \log {\frac{n}{s}}})$, where $s$ is the bound on the set size and obtained a multicut mimicking network of size $k^{\OO(\log^3 k)}$ in polynomial-time using a more careful analysis. 

From our results in~\Cref{section:sse}, we obtain an algorithm with $\alpha(n,k) = \OO(\log k)$. Clearly,~\Cref{thm:mmn} follows from~\Cref{thm:mmn_sse} and the fact that $\alpha(n,k) = \OO(\log k)$. This improves the result to $k^{\OO(\log^2 k)}$.

This improves the best known kernels for many problems which are a consequence of multicut mimicking networks as considered in~\cite{wahlstrom22}. In particular, {{\sc Multiway Cut}} parameterized by the cut-size and {{\sc Multicut}} parameterized by the number of pairs and the cut-size together, admit $\OO(k^{\OO(\log^2 k)})$ size kernels in polynomial time.

\subsection{Min-max graph partitioning using weighted sample sets}
In min-max graph partitioning, we are given a graph $G$ and an integer $r$. The goal is to partition the graph into $r$ components of size $\frac{n}{r}$ each such that the maximum cut-size among the $r$ components is minimized. Bansal et. al~\cite{minmax} gave an $\OO(\sqrt{\log n \log r})$-approximation, and~\cite{hc16} obtained an $\OO(\log \frac{n}{r})$-approximation using their approximation for \SSE. We show that we can obtain an $\OO(\log opt)$-approximation for min-max graph partitioning, where $opt$ is the min-max objective. This follows from an $\OO(\log opt)$-approximation for Weighted-$\rho$-unbalanced cut.~\cite{minmax} shows that an $\alpha$-approximation for Weighted-$\rho$-unbalanced cut essentially implies an $\alpha$-approximation for min-max graph partitioning, so henceforth we will restrict our attention to Weighted-$\rho$-unbalanced cut.

\begin{definition}[Weighted-$\rho$-unbalanced cut \cite{minmax}] Given a graph $G$, a weight function $y: V(G) \rightarrow \mathbb{Q^{+}}$, parameters $\tau, \rho \in (0,1)$, find a set $S$ with $|S| \leq \rho n$ such that $y(S) \geq \tau y(V(G))$ such that $\delta_G(S)$ is minimum possible, or report that no such set exists.

\end{definition}

Let us denote the minimum possible $\delta_G(S)$ for some set $S$ satisfying the above by $opt$. We now define an $(\alpha, \beta, \gamma)$-approximation for this problem as in~\cite{minmax}.

\begin{definition}[Approximate Weighted-$\rho$-unbalanced cut] Given $(G,y,\tau,\rho)$, suppose that there exists a set $S$ with $|S| \leq \rho n$ such that $y(S) \geq \tau y(V(G))$ and $\delta_G(S) = opt$. Find a set $S'$ with $|S'| \leq \beta\rho n$ such that $y(S') \geq \frac{\tau y(V(G))}{\gamma}$ such that $|\delta_G(S')| \leq \alpha \;opt$.
\end{definition}

\begin{theorem}\label{thm:unbalanced}
There is a $(\alpha = \OO(\log opt), \beta = \OO(1), \gamma = \OO(1))$-approximation for weighted-$\rho$- unbalanced cut in polynomial time.      
\end{theorem}

Before we prove this theorem, we will need a generalization of sample sets as defined in~\Cref{def:sample_set}. Given a graph $G$ and a measure $\mu$ on the vertices, we define the \emph{$\mu$-sparsity} of a set $W$ as $\phi_\mu(S) = \frac{|\delta_G(S)|}{\mu(S)}$. Given a weight function $w: V(G) \rightarrow \mathbb{Q^{+}}$, we write $w(X) = \sum_{v \in X} w(v)$ for any subset $X \subseteq V(G)$.

\begin{definition}\label{def:sample_set_mu}
Given a graph $G = (V(G), E(G))$ with a measure $\mu$ on the vertices, parameters $\epsilon, \phi'$, an $(\epsilon, \phi')$ (edge) sample set for $G$ is a set of terminals $T \subseteq V(G)$ with a weight function $\alpha$ which assigns a weight $\alpha_t \geq 1$ to each terminal $t$. Suppose that $\mu(V(G)) = K$. Consider a  set of vertices $W \subseteq V(G)$ which  satisfies either (a) $\phi_\mu(W) \leq \phi'$ or (b) $\phi_\alpha(W) \leq \frac{K}{10\alpha(T)}\phi'$. Then we must have $|\frac{K}{\alpha(T)}\alpha(W \cap T) - \mu(W)| \leq \epsilon \mu(W)$.
\end{definition}

\begin{lemma}\label{lemma:sample_set_mu}
There is an $(\epsilon,\phi')$ edge sample set for every graph $G$ of total weight \\ $\alpha(V(G)) = \alpha(T) = \Theta\left(\mu(V(G)\frac{\phi'}{\epsilon^2}\right)$. Further, such a set can be computed in polynomial time.
\end{lemma}

\begin{proof}
Our proof is similar to the construction of the determinstic sample set in~\cite{fm06} for edge separators. We will use the notion of Steiner-$t$-decomoposition introduced in~\cite{fm06}.
\begin{lemma}[Follows from Lemma 5.2,  ~\cite{fm06} and Claim 3.3~\cite{flspw18}]
Given any graph $G$, a measure $\mu'$ on the vertices, and a value $1 \leq t \leq n$. Suppose every vertex $v \in V(G)$ satisfies $\mu'(v) \leq t$. Then one can compute in polynomial time a partition $V_1, V_2 \ldots V_{\ell}$ of $V(G)$ and a partition of the edge set $E_1, E_2 \ldots E_{\ell}, E'$ such that
\begin{enumerate}
    \item $G[E_i]$ is connected for each $i \geq 1$
    \item $V(G[E_i]) = V_i \cup \{u\}$ for some vertex $u$.
    \item $\mu'(V_i) \leq 2t$ for all $i \in [\ell]$.
    \item $\mu'(V_i) \geq t$ for all $i > 1$.
\end{enumerate}
\end{lemma}

We will refer to each vertex set $V_i$ as a \emph{bag} in the decomposition.

We now compute a sample set as follows. Fix $t = \frac{\epsilon}{100\phi'}$. First, for any vertex $v$ with $\mu(v) > t$, which we shall henceforth call a \emph{high vertex}, pick $v$ into the sample set $T$, and assign it a weight $\alpha_v = \lfloor \frac{\mu(v)}{t\epsilon} \rfloor$. Let $V_h$ denote the set of high vertices. Define a new measure $\mu'$ as follows.
  \[
    \mu'(v) = \left\{\begin{array}{lr}
        \mu(v), & \text{if } v \notin V_h\\
        0 & \text{otherwise.} \\
        \end{array}\right\}
  \]
Compute a Steiner-$t$-decomposition for $\mu'$ with $t = \frac{\epsilon}{100\phi'}$. Recall that $\mu(V(G)) = K$, and let $\mu'(V(G)) = K'$. We now add to the sample set $T$ as follows: arbitrarily pick $1$ vertex from the bag $V_i$ and assign it a weight of $\lfloor \frac{\mu'(V_i)}{t\epsilon} \rfloor$ for each $i > 1$.

\begin{lemma}
The set $T$ is a $(4\epsilon, \phi')$ sample set for all $\epsilon \in (0,\frac{1}{100})$.
\end{lemma}

\begin{proof}
We first show that every bag is represented proportional to its measure in the terminal set. Observe that from every bag $V_i$, we pick one vertex with weight $\frac{\mu'(V_i)}{t\epsilon} \geq \lfloor \frac{\mu'(V_i)}{t\epsilon} \rfloor \geq \frac{\mu'(V_i)}{t\epsilon} - 1$. Since $t \leq \mu'(V_i)$ for all $i > 1$, this means that $\alpha(V_i) \in [\frac{\mu'(V_i)}{t\epsilon}-\epsilon\frac{\mu'(V_i)}{t\epsilon}, \frac{\mu'(V_i)}{t\epsilon}]$ for all $i > 1$. Similarly, we also obtain $\alpha(v_h) \in  [\frac{\mu(v_h)}{t\epsilon}-\epsilon\frac{\mu(v_h)}{t\epsilon}, \frac{\mu(v_h)}{t\epsilon}]$ for every vertex $v_h \in V_h$.

Adding and noting that $\mu'(V_1) \leq 2t$  we obtain $\alpha(T \setminus V_h) \in [\frac{(K' - 2t)}{t\epsilon}(1 - \epsilon), \frac{K'}{t\epsilon}]$. Similarly, we get $\alpha(V_h) \in [\frac{(K-K')}{t\epsilon}(1 - \epsilon), \frac{K - K'}{t\epsilon}]$. Together, we obtain $\alpha(V(G)) = \alpha(T) \in [\frac{K - 2t}{t\epsilon}(1 - \epsilon), \frac{K}{t\epsilon}]$. We know that $t = \frac{\epsilon}{100 \phi'} \leq \frac{\epsilon}{4} K$, which in turn gives ${K-2t} \geq (1 - \frac{\epsilon}{2})K$, and hence we have $\alpha(T) \in [(1 - \frac{\epsilon}{2})(1 - \epsilon)\frac{K}{t\epsilon}, \frac{K}{t\epsilon}]$ or equivalently, $\alpha(T) \in [(1 - 2\epsilon)\frac{K}{t\epsilon}, \frac{K}{t\epsilon}]$. This means that  $\frac{K}{\alpha(T)}\alpha(V_i) \in [\mu'(V_i){(1 - \epsilon)}, \frac{\mu'(V_i)}{1-2\epsilon}]$ and hence $\frac{K}{\alpha(T)}\alpha(V_i) \in [\mu'(V_i){(1 - \epsilon)}, {\mu'(V_i)}{(1+3\epsilon)}]$ for small enough $\epsilon$, for all $i > 1$. Similarly we also obtain $\frac{K}{\alpha(T)} \alpha(v_h) \in [\mu(V_h)(1 - \epsilon), \mu(V_h)(1 + 3\epsilon)]$ for each $v_h \in V_h$.

First, consider the case where we are given a set $W$ with $\phi_\mu(W) \leq \phi'$.  Observe that the number of edges cut, $\delta_G(W)$ is at most $\phi'\mu(W)$. Mark a bag $V_i$ of the decomposition as ``bad'' if there is a cut edge which has both endpoints in $E_i$. Additionally, we always mark the bag $V_0$ as bad. Clearly, the number of bad bags is at most $\phi'\mu(W) + 1$. The total $\mu$-measure in these bags excluding the vertices in $V_h$ is at most $2t\phi'\mu(W) + t\leq \frac{\epsilon\mu(W)}{4}$, since $t = \frac{\epsilon}{100\phi'} \leq \frac{\epsilon \delta_G(W)}{100\phi'} \leq \frac{\epsilon \mu(W)}{100}$. Thus at least $(1 - \frac{\epsilon}{4})\mu(W)$ measure of $W$ is either in good bags or the vertices $V_h$. Observe that by the definition of a good bag, if $W$ intersects a good bag $V_i$, it must include the whole bag $V_i$, since there is a way to reach every vertex of the bag using the edges $E_i$, and none of the edges in $E_i$ are cut edges.

Let $G'$ be the set of vertices of the good bags and the vertices $V_h$ contained in $W$. From the preceeding argument, we know that $\mu(G') \geq (1 - \frac{\epsilon}{4})\mu(W)$.
Recall that every good bag $V_i$ satisfies $\frac{K}{\alpha(T)}\alpha(V_i) \in [\mu'(V_i){(1 - \epsilon)}, {\mu'(V_i)}{(1+3\epsilon)}]$, and every vertex $v_h \in V(h)$ 
satisfies $\frac{K}{\alpha(T)} \alpha(v_h) \in [\mu(V_h)(1 - \epsilon), \mu(V_h)(1 + 3\epsilon)]$. It follows that $W$ must satisfy $\frac{K}{\alpha(T)} \alpha(W) \geq (1 - \epsilon)(1 - \frac{\epsilon}{4})\mu(W) \geq (1 - 2\epsilon)\mu(W)$. On the other hand,  we have $\frac{K}{\alpha(T)}\alpha(G') \leq \mu'(V_i)(1 + 3\epsilon)$ and  $\frac{K}{\alpha(T)} \alpha(v_h) \leq \mu(V_h)(1 + 3\epsilon)$. Let $B$ be the union of vertices in all the bad bags. Observe that by construction we must have $\alpha(B) \leq \frac{\mu(B)}{t\epsilon}$. Noting that $\alpha(T) \geq (1- 2\epsilon)\frac{K}{t\epsilon}$, this gives $\frac{K}{\alpha(T)}\alpha(B) \leq \frac{\mu(B)}{1- 2\epsilon} \leq \frac{\mu(W)\epsilon}{4(1 - 2\epsilon)} \leq \frac{\mu(W)\epsilon}{2}$ for small enough $\epsilon$. Together, we obtain $\frac{K}{\alpha(T)}\alpha(W) \leq \frac{K}{\alpha(T)}\alpha(G')  + \frac{K}{\alpha(T)}\alpha(B) + \frac{K}{\alpha(T)} \alpha(V_h) \leq \mu(W)(1 + 3\epsilon + \frac{\epsilon}{2}) \leq \mu(W)(1 + 4\epsilon)$.

Now consider the case when we are given a set $W$ with $\phi_\alpha(W) \leq \frac{\epsilon^2}{200}$ (Note that $\frac{\epsilon^2}{200} \geq \frac{K}{10\alpha(T)}\phi'$). Again, we have $\delta_G(W) \leq \phi_T(W)\alpha(W) \leq \frac{\epsilon^2}{200}\alpha(W)$, and this upper bounds the number of bad bags. Define $G'$ and $B$ as before. Recall that each bag $V_i$ satisfies $\alpha(V_i) \leq \frac{\mu'(V_i)}{t\epsilon} \leq \frac{2}{\epsilon}$. This means that the total weight of terminals in bad bags $\alpha(B) \leq \frac{2}{\epsilon} \cdot \frac{\epsilon^2}{200}\alpha(W) \leq \epsilon \alpha(W)$. Thus we must have $\alpha(W \cap B) \leq \epsilon \alpha(W)$. This implies that $\alpha(G')\geq (1 - \epsilon)\alpha(W)$. Since $G'$ is a collection of bags (and high vertices), as in the previous analysis, we must have $\frac{K}{\alpha(T)}\alpha(G') \in [\mu(G')(1 - \epsilon), \mu(G')(1 + 3\epsilon)]$. This gives 
$$\mu(W) \geq \mu(G') \geq \frac{K}{\alpha(T)}\frac{\alpha(G')}{1 + 3\epsilon} \geq \frac{K}{\alpha(T)}(1 - \epsilon) \frac{\alpha(W)}{1 + 3\epsilon}$$

It follows that
$$\mu(W) \geq \frac{K}{\alpha(T)}\alpha(W)(1 - 4\epsilon) \geq \frac{K}{2\alpha(T)} \alpha(W) \geq \frac{K}{\alpha(T)} \frac{100}{\epsilon^2}\delta_G(W).$$

Finally, note that $\frac{K}{\alpha(T)} \geq t\epsilon$, so that we get $\mu(W) \geq \frac{1}{\phi'}\delta_G(W)$, and hence $W$ is $\phi'$-sparse. Since we already proved the condition for any $\phi'$-sparse set, this concludes the proof of the lemma.\end{proof}

As shown already, we have $\alpha(T) \leq \frac{K}{t\epsilon} = \OO(\frac{K\phi'}{\epsilon^2}) = \OO(\frac{\mu(V(G))\phi'}{\epsilon^2}).$\end{proof}

\begin{proof}[Proof of \Cref{thm:unbalanced}]
We will again combine our notion of sample sets with the LP-rounding framework. First, suppose that we are given a graph $G$ together with a set of terminals $T$ and integral weights $\alpha$ on the terminals, with $\alpha(v) \geq 1$ whenever $\alpha(v) \neq 0$ for any vertex $v$. Now further suppose that there exists a set $S \subseteq V(G)$ with $|S| = s$, $\alpha(S) = \ell$ and $|\delta_G(S)| = opt$. Then we design an algorithm that finds a set $S'$ with $|S'| \leq 10 s$, $\alpha(S') \geq \frac{\ell}{10}$ and $|\delta_G(S')| \leq \OO(opt \log \ell)$. The proof combines the ideas behind Theorem 13 of~\cite{hc16} and Theorem 2.1 of~\cite{minmax}. We postpone the proof to the appendix (see \Cref{thm:sse}).

Now we proceed as follows. Let $(S, V(G) \setminus S)$ denote the optimal partition where $S$ has size $s \leq \rho n$. First, we guess the size $s$ by trying all choices from $1$ to $\rho n$ (In fact is sufficient to guess $s$ up to factor $2$). Similarly, we also guess the value $y(S)$, suppose it is equal to $\ell$. It follows that $S$  is $\phi = \frac{opt}{\ell}$-sparse with respect to the measure $y$. Choose $\epsilon = \frac{1}{100}$ and construct an $(\epsilon, \phi' = C\phi \log opt)$ sample set $T$ with total weight $\alpha(T) = \Theta(y(V(G))\phi \log opt)$ where $C$ is a large enough constant that will be chosen later.
 The sample set guarantee implies that $\alpha(S) = \Theta({Copt \log opt})$. Using the result from~\Cref{thm:sse}, it follows that we obtain a set $S'$ with $|S'| \leq 10 s'$, $\alpha(S') \geq \frac{\alpha(S)}{10}$ with $|\delta_G(S')| \leq \OO(\log opt) opt$. This set has terminal sparsity $\phi_\alpha(S) = \OO(\frac{1}{C})$. Now we choose $C$ large enough so that $\phi_T(S') \leq \frac{\epsilon^2}{200} \leq \frac{K}{10\alpha(T)}\phi'$.

Then $S'$ must satisfy the sample condition. Using the guarantee of sample sets, since $\alpha(S') \geq \frac{C}{10} opt \log opt$ it must be the case that $y(S') \geq \Omega(\ell)$. This completes the proof.\end{proof}

\section{Discussion and Open Problems}\label{sec:open}

\paragraph{Fast $O(\polylog k)$-approximate \SSE.} We show $O(\polylog k)$-approximation algorithms for \SSE{} that run in polynomial time based on LPs in \Cref{sec:sselp}. Via a new cut-matching game, we speed up the running time to almost linear time to solve the easier problem, \SC{}, where we do not require the output set to be small (see \Cref{sec:cutmatching}). A natural question is whether there is an almost-linear time $O(\polylog k)$-approximation algorithm for \SSE{}.

\paragraph{$O(\polylog k)$-approximate \VSC.}
Can we remove the $\poly(\log\log n\phi)$ term in \Cref{thm:logkvertex,thm:cutmatchingvertexlogk} for approximating \VSC{} and obtain $O(\polylog k)$ approximation algorithms? We have shown a fundamental barrier of our current approach based on sample sets in \Cref{subsec:sample_set_lower_bound}.

\bibliographystyle{plain}
\bibliography{references}

\appendix

\section{Appendix}\label{sec:appendix}

In this section, we extend an $\OO(\log s)$-approximation for \SSE{} to the terminal version and vertex terminal version.

\begin{theorem}\label{thm:sse}

Suppose we are given a graph $G$ and a set of terminals $T \subseteq V(G)$, and integral weights $x$ such that $x(V(G)) = K$,  $x$ is non-zero only on vertices of $T$, and $x(v) \geq 1$ for each $v \in T$. Suppose there exists a set with $|S| = s$, $x(S) = \ell$ and $|\delta_G(S)| = opt$. Then we can find a set $Y$ satisfying $|Y| \leq 10 s$, $3\ell \geq x(Y) \geq \frac{\ell}{10}$ with $|\delta_G(Y)| \leq \OO(\log \ell) opt$.

\end{theorem}

\begin{proof}

We start with the LP in~\cite{minmax} with some modifications. Let us fix $|T| = t$.

\begin{center}
\noindent\fbox{

  \begin{minipage}{0.95\textwidth}
  $$ \min \sum_{e=\{u,v\} \in E(G)} d(u,v)$$
  $$\sum_{u \in V(G)} \min\{d(u,v) , y_v\} \geq (n - s)y_v\;\forall v \in V(G)\;\;\;\text{\ldots(1)}$$
  $$\sum_{u \in V(G)}
  x(u) \min\{d(u,v) , y_v\} \geq (K -\ell)y_v\;\forall v \in V(G) \;\;\;\text{\ldots(2)}$$
  $$\sum_{v \in V(G)} x(v) y_v \geq \ell $$
  $$\sum_{v \in V(G)} y_v \leq s$$
  $$d(u,v) \leq d(u,w) + d(w,v)\; \forall u,v,w \in V(G)$$
  $$d(u,v) = d(v,u)\; \forall u,v \in V(G)$$
  $$d(u,v) \geq y_u - y_v\;\forall u,v \in V(G)$$
  $$ d(u,v), y_v \geq 0\; \forall u,v \in V(G)$$

  \end{minipage}
 
      }
\end{center}

It is clear that the LP is a relaxation - to see this, given the integral solution $S$, set $y_v = 1$ for each $v \in S$, and $0$ otherwise. Set $d(u,v) = 1$ if $u \in S$ and $v\notin S$ or $u \notin S$ and $v \in S$, and $d(u,v) = 0$ otherwise. The LP value is at most $opt$, where $opt$ denotes the size of the optimal cut $|\delta_G(S)|$.

\textbf{The rounding:} Having obtained the optimal values $d(u,v)$ from the optimal LP solution, we now describe our rounding procedure. Our goal will again be to produce an LP-separator as introduced in~\cite{minmax}, but with the improved approximation ratio $\OO(\log \ell)$. Before describing our rounding scheme, we begin with a crucial observation that will be used multiple times.

\begin{observation}\label{obs:ball_minmax}
For any vertex $v \in V(G)$ and $r < 1$, let $Ball_r(v) = \{u \in V(G) \mid d(u,v) \leq ry_v\}$. Then for every $v \in V(G)$, we must have $|Ball_r(v)| \leq \frac{s}{1- r}$ and $x(Ball_r(v)) \leq \frac{\ell}{1-r}$.
\end{observation}
\begin{proof}
Let $|Ball_r(v)| = z$. The LP spreading constraint (1) for vertex $v$ implies that we must have $z \cdot ry_v + (n - z) \cdot y_v \geq (n- s)y_v$. This in turn means that $z \leq \frac{s}{1-r}$.

Similarly, let $x(Ball_r(v)) = z'$. The second spreading constraint (2) for vertex $v$ implies that we must have $z' \cdot ry_v + (K - z') \cdot y_v \geq (K-\ell)y_v$ which gives $z' \leq \frac{\ell}{1-r}$.\end{proof}

For a set $X \subseteq V(G)$, we define $f(X) = x(X) -  \frac{\ell}{10s}|X| - \frac{\ell}{200LP\log \ell}|\delta_G(X)|$ where $LP$ is the optimal LP value.

\begin{algorithm}

\caption{Rounding the LP}
\begin{algorithmic}
\State $Y \leftarrow \emptyset$
\While {$|Y| \leq s$ and $x(Y) \leq \frac{\ell}{4}$}

\For {$j = 1$ to $\OO(n \log n)}$

\State Pick a threshold $\delta \in [0,1]$ uniformly at random.
\State Let $X = \{v \in T \mid y_v \in [\delta, 2\delta]\}$ and $\pi: [|X|] \rightarrow X$ be a random permutation of the elements of $X$. Let $U\leftarrow \emptyset$
\State Pick $r \in [0.05,0.1]$ uniformly at random.
\For{$i = 1,2 \ldots |X|$}
\State Let $C = \{v \in V(G) \setminus U \mid d(\pi(i),v) \leq {ry_{\pi(i)}}\}$
\State $U \leftarrow U \cup C$
\State $\SS \leftarrow \SS \cup \{C\}$
\EndFor
\State  Assign to $C$ one cluster of $\SS$ with probability $\frac{1}{n}$, return empty set with probability $1 - \frac{|\SS|}{n}$
\If {$f(C) > 0$}
\State $Y = Y \cup C$
\State \textbf{break}
\EndIf
\EndFor
\EndWhile

\end{algorithmic}
\label{alg:rounding_minmax}
\end{algorithm}

The algorithm runs for multiple iterations. \Cref{alg:rounding_minmax} describes the rounding scheme. We start with $Y = \emptyset$. In every iteration, as long as $|Y| \leq s$ and $x(Y) \leq \frac{\ell}{4}$, we proceed with the next iteration to find a cluster $C$  that satisfies $f(C) > 0$. We then set $Y = Y \cup C$, and repeat.

Here we note that at an intermediate step, once the set $Y$ is removed, the LP solution for the vertices $V(G) \setminus Y$ remains  ``almost-feasible'' for the graph $G \setminus Y$. More precisely, every constraint is still satisfied, except potentially the constraint $\sum_{v \in V(G \setminus Y)} x(v) \geq \ell$. However since $x(Y) \leq \frac{\ell}{4}$, it must be the case that $\sum_{v \in V(G \setminus Y)} x(v) \geq \frac{3\ell}{4}$. Therefore we will use this weaker constraint in our analysis at an intermediate step.

Next, we describe the analysis of the rounding algorithm for each iteration, where we find a cluster $C$.

Every time we pick a vertex $\pi(i)$ and separate out a cluster $C$, we will call $\pi(i)$ the \emph{center} of the cluster $C$. Also, we will say that two vertices $u$ and $v$ are separated while considering $\pi(i) \in T$ if exactly one of $u,v$ is in the cluster with center $\pi(i)$, and neither $u$ nor $v$ is in any cluster whose center is $\pi(1), \pi(2) \ldots \pi(i-1)$.  

\begin{lemma}
In each iteration of the outer for loop, the output set $C$ satisfies the following properties.
\begin{enumerate}
    \item For every vertex $v \in T$, we have $v \in C$ with probability at least $\frac{y_v}{2n}$. For any vertex $v \in V(G)$, $v \in C$ with probability at most $\frac{2y_v}{3n}$.
    \item $|C| \leq 2s$ and $x(C) \leq 2\ell$.
    \item The expected number of edges cut, $E[|\delta_G(C)|] \leq \frac{\OO(\log \ell)LP}{n}$.
\end{enumerate}
\end{lemma}

\begin{proof}
Clearly, for any $v \in T$ $y_v \in [\delta, 2 \delta]$ whenever $\delta \in [\frac{y_v}{2}, y_v]$ which happens with probability at least $\frac{y_v}{2}$. This means that with probability $\frac{y_v}{2}$, $v$ is in some cluster $W \in \SS$. But each cluster is picked with exactly probability $\frac{1}{n}$, and hence it must be the case that $v$ is in $C$ with probability at least $\frac{y_v}{2n}$. Now we show the other part. Suppose $v \in V(G)$ is in some cluster $W$ and $w$ is the center of the cluster $W$. We have $d(u,w) \leq 0.1y_w$. Using the LP constraints, it is easy to see that $0.9y_w \leq y_v \leq 1.1y_w$. But $y_w \in [\delta, 2\delta]$, and hence $0.9 \delta \leq y_v \leq 2.2 \delta$, or equivalently we must have $\frac{y_v}{2.2} \leq \delta \leq \frac{y_v}{0.9}$. But this happens with probability at most $\frac{2y_v}{3}$. Again, since the cluster containing $v$ is picked with probability exactly $\frac{1}{n}$, the probability that $v$ is in the chosen cluster is at most $\frac{2y_v}{3n}$.

(2) directly follows from \Cref{obs:ball_minmax} with $r = 0.1$.

Finally, we prove (3). Consider an edge $e = (u,v)$. This edge can be cut only if $u$ or $v$ is in some cluster $C$ of $\SS$. Without loss of generality suppose that $u$ is in some such cluster $C$ (the other case is symmetric). Let $w$ be the center of cluster $C$. 

As noted above, we must have $\frac{y_u}{2.2} \leq \delta \leq \frac{y_u}{0.9}$. Thus, the only values of $\delta$ for which the edge $\{u,v\}$ can be an edge in $\delta_G(C)$ must satisfy $\frac{y_u}{2.2} \leq \delta \leq \frac{y_u}{0.9}$ or $\frac{y_v}{2.2} \leq \delta \leq \frac{y_v}{0.9}$.  

Now, fix a value of $\delta$ such that either $\frac{y_u}{2.2} \leq \delta \leq \frac{y_u}{0.9}$ or $\frac{y_v}{2.2} \leq \delta \leq \frac{y_v}{0.9}$.

Let $X' \subseteq X$ be the set of vertices $w$ of $X$ that satisfy $d(w,u) \leq 0.1y_w $ or $d(w,v) \leq 0.1y_w$. Informally, these are the only vertices that can separate $u$ or $v$ from each other, thereby cutting the edge $\{u,v\}$. We show first that $|X'| \leq 10\ell$. Suppose to the contrary that $|X'| > 10\ell$. It is clear that at least $5\ell + 1$ vertices $w \in X'$ that satisfy one of the two inequalities $d(w,u) \leq 0.1y_w$ or $d(w,v) \leq 0.1y_w$. Without loss of generality, we assume that $5\ell + 1$ vertices $w \in X$ satisfy $d(w,u) \leq 0.1y_w$. Let $X''$ be the set of such vertices $w$. Fix some $w' \in X''$. 
By the triangle inequality, for every $w \in X''$ we have $y_w + 0.1y_w \geq y_u \geq y_w - 0.1y_w = 0.9y_w$. Thus $d(w,u) \leq \frac{y_u}{9}$. Also $d(w',u) \leq \frac{y_u}{9}$. By the triangle inequality on $d$, we must have $d(w',w) \leq d(w,u) + d(w,w') \leq \frac{2y_u}{9} \leq \frac{2\cdot1.1y_{w'}}{9} \leq \frac{y_w}{3}$. But this means $|Ball_{\frac{1}{3}}(w') \cap T| \geq 5\ell$, which in turn means $x(Ball_{\frac{1}{3}}(w')) \geq 5\ell$ which contradicts \Cref{obs:ball_minmax}.

Order the vertices $w$ of $X'$ in increasing order of $\frac{\min\{d(u,w), d(v,w)\}}{y_w}$, and let $Q = (t_1, t_2 \ldots t_{|X'|})$ denote this sequence. Notice that if some $t_i$ is considered before $t_j$ for some $i < j$ in $\pi$, then the edge $(u,v)$ will already be cut/removed in one cluster when $t_i$ is considered and hence we cannot cut this edge while considering $t_j$.

Also, once we obtain the clusters, the edge $(u,v)$ is cut only if we pick a cluster containing either $u$ or $v$, which happens with probability at most $\frac{2}{n}$.
\begin{align*}
    Pr[u,v\; \text{separated} | \delta]\ &\leq \frac{2}{n} \sum_{i = 1}^{|X'|} Pr[u,v\; \text{separated while considering $t_i$}]\\
                                  &\leq \frac{2}{n} \sum_{i = 1}^{|X'|} Pr[t_i \text{ appears before $t_1, t_2$} \ldots t_{i-1}\\
                                  &\text{and} \min\{d(u,t_i), d(v,t_i)\} \leq Cy_{t_i} \leq \max\{d(u,t_i), d(v, t_i)\}]\\
                                  &\leq \frac{2}{n} \sum_{i = 1}^{|X'|} \frac{|d(u,t_i) - d(v, t_i)|}{0.05y_{t_i}i}\\
                                  &\leq \frac{2}{n}  \sum_{i = 1}^{|X'|} \frac{d(u,v)}{0.05\delta i} \\
                                  &\leq \frac{2}{n} 20\log \ell \cdot \frac{d(u,v)}{\delta}
\end{align*}
Here, we used the fact that $|d(u,t_i) - d(v, t_i)| \leq d(u,v)$. The last step follows since $|X'| \leq 10\ell$. Finally, as observed before, the edge $\{u,v\}$ can be separated only when $\frac{y_u}{2.2} \leq \delta \leq \frac{y_u}{0.9}$ or $\frac{y_v}{2.2} \leq \delta \leq \frac{y_v}{0.9}$. We now obtain
\begin{align*}
Pr[u,v\; \text{separated}] &\leq \int_{\frac{y_u}{2.2}}^{\frac{y_u}{0.9}}Pr[u,v\; \text{separated} | \delta] Pr[\delta] d \delta + \int_{\frac{y_v}{2.2}}^{\frac{y_v}{0.9}}Pr[u,v\; \text{separated} | \delta] Pr[\delta] d \delta \\
&\leq 20 \log \ell \cdot \frac{2}{n}\left ( \int_{\frac{y_u}{2.2}}^{\frac{y_u}{0.9}}\frac{d(u,v)}{\delta} d \delta + \int_{\frac{y_v}{2.2}}^{\frac{y_v}{0.9}} \frac{d(u,v)}{\delta}d\delta \right) \\
&\leq \frac{100}{n} \log \ell \cdot d(u,v).
\end{align*}\end{proof}

The rest of the analysis is similar to~\cite{minmax}, however we present it again for completeness' sake. For a set $X \subseteq V(G)$, recall that we defined $f(X) = x(X) -  \frac{\ell}{10s}|X| - \frac{\ell}{200LP\log \ell}|\delta_G(X)|$ where $LP$ is the optimal LP value. Then if $C$ is the set returned by the algorithm, we must have $E[f(C)] = E[x(X)] - \frac{\ell}{200LP\log \ell}E[|\delta_G(C)|] - \frac{\ell}{10s}E[|X|] \geq  \sum_{v \in T}\frac{y_v}{2n} x(v)  - \frac{\ell}{200LP \log \ell}\frac{50 \log \ell LP}{n} - \frac{\ell}{10sn} \frac{2s}{3} \geq \frac{\ell}{100n}$, where we use the fact that $\sum_{v \in T} x(v) y_v  \geq \frac{3\ell}{4}$, and the fact that $E[|X|] \leq \sum_{v}\frac{2y_v}{3n} \leq \frac{2s}{3n}$. Also note that $f(C) \leq  2\ell$. It follows that $E[f(C)] \leq 2\ell \cdot Pr[f(C) > 0]$, which means $Pr[f(C) > 0] \geq \frac{E[f(C)]}{2\ell} \geq \frac{1}{200n}$. Thus repeating the algorithm $\OO(n\log n)$ times, we are guaranteed to find a set with $f(C) > 0$ with high probability. But $f(C) > 0$ means that $C$ satisfies  $\delta_G(C) \leq \OO(\log \ell) \frac{x(C)}{\ell} LP \leq \OO(\log \ell) \frac{x(C)}{\ell} OPT$. Also $f(C) > 0$ implies $x(C) \geq \frac{\ell}{10s}|C|$.

Recall that the algorithm keeps finding such a set $C$ and sets $Y = Y \cup C$ till we obtain $|Y| \geq s$ or $x(Y) \geq \frac{\ell}{4}$. If $|Y| \geq s$, since $f(C) > 0$ in each iteration, it is clear that $x(Y) \geq \frac{\ell}{10s} s = \frac{\ell}{10}$. Note that we must have $|Y| \leq 4s$, since $|C| \leq 2s$ for every set $C$ that we find.

Finally, note that $x(Y) \leq \frac{\ell}{4} + 2\ell = 3\ell$ since $x(C) \leq 2\ell$ for each $C$ that we find. This gives $|\delta_G(Y)| \leq \OO(\log \ell) \frac{x(Y)}{\ell} LP \leq \OO(\log \ell) LP \leq \OO(\log \ell) opt$ as desired.\end{proof}

The next theorem adapts the above theorem to the vertex version to obtain~\Cref{thm:lpvertex}. Here we only need a simpler version for our results, so we do not have the weight function $x$. Yet another marked difference is that we no longer return a small set. For convenience, we recall~\Cref{thm:lpvertex}.

\thmlpvertex*
\begin{proof}

The proof will essentially be a modified algorithm based on the edge version. Now we have an additional set of variables $b_v$ which indicate if vertex $v$ is in the cut. Thus, for the canonical integral solution, we have $y_v = 1$ for every vertex $v \in L$ and $b_v = 1$ for every vertex $v \in C$, and $d(u,v) = 1$ if $u \in L$ and $v \notin L$, or $u \in R$ and $v \notin R$.

\begin{center}
\noindent\fbox{

  \begin{minipage}{0.95\textwidth}
  $$ \min \sum_{v \in V(G)} b_v$$
  $$\sum_{u \in T} \min\{d(u,v) , y_v\} \geq (t - s)y_v\;\forall v \in T$$

  $$\sum_{v \in T} y_v \geq s$$
  $$d(u,v) \leq d(u,w) + d(w,v)\; \forall u,v,w \in V(G)$$
  $$d(u,v) \leq d(u,w) + b_v\; \forall u,v,w \in V(G), (w,v) \in E(G)$$
  $$d(u,v) = d(v,u)\; \forall u,v \in V(G)$$
  $$d(u,v) \geq y_u - y_v\;\forall u,v \in V(G)$$
  $$ d(u,v), y_v \geq 0\; \forall u,v \in V(G)$$

  \end{minipage}
       }
\end{center}

It is clear that the LP is a relaxation and the LP value is at most $opt$, where $opt$ denotes the size of the optimal cut ($opt = |C|$).

\textbf{The rounding:} Having obtained the optimal values $d(u,v),b_v,y_v$ from the optimal LP solution, we now describe our rounding procedure. Our goal will again be to produce an LP-separator as introduced in~\cite{minmax}, but with the improved approximation ratio $\OO(\log s)$. Before describing our rounding scheme, we begin with a crucial observation that will be used multiple times.

\begin{observation}\label{obs:ball}
For any vertex $v \in T$ and $r < 1$, let $Ball_r(v) = \{u \in V(G) \mid d(u,v) \leq ry_v\}$. Then for every $v \in T$, we must have $|Ball_r(v) \cap T| \leq \frac{s}{1- r}$.
\end{observation}
\begin{proof}
Let $|Ball_r(v) \cap T| = z$. The LP spreading constraint for vertex $v$ implies that we must have $z \cdot ry_v + (t - z) \cdot y_v \geq (t- s)y_v$. This in turn means that $z \leq \frac{s}{1-r}$.\end{proof}

\begin{algorithm}

\caption{Rounding the LP}
\begin{algorithmic}
\State $C \rightarrow \emptyset$.
\For{$j = 1$ to $\OO(n\log n)$}
\State Pick a threshold $\delta \in [0,1]$ uniformly at random.

\State Let $X = \{v \in T \mid y_v \in [\delta, 2\delta]\}$ and $\pi: [|X|] \rightarrow X$ be a random permutation of the elements of $X$. Let $U\leftarrow \emptyset$
\State Pick $r \in [0.05,0.1]$ uniformly at random.
\For{$i = 1,2 \ldots |X|$}
\State Let $C = \{v \in V(G) \setminus U \mid d(\pi(i),v) \leq {ry_{\pi(i)}}\}$
\State $U \leftarrow U \cup C$
\State $\SS \leftarrow \SS \cup \{C\}$
\EndFor

\If{\text{$|U \cap T|< \frac{|T|}{2}$}}
\Return $U' \leftarrow U$
\Else
\State Divide the clusters of $\SS$ into two groups, $\SS_1$ and $\SS_2$, such that each group has $\geq \frac{|(U \cap T)|}{3}$ terminals.
\State \Return the group of vertices $U'$ with at most $\frac{|T|}{2}$ terminals.
\EndIf
\If{$f(U') > 0$}
\State \textbf{break}
\EndIf
\EndFor

\end{algorithmic}
\label{alg:rounding}
\end{algorithm}

\Cref{alg:rounding} describes the rounding scheme. For a set $X \subseteq V(G)$ we define $f(X) = |X \cap T| - \frac{s}{2000LP\log s}|N_G(X)|$ where $LP$ is the optimal LP value.

Everytime we pick a vertex $\pi(i)$ and separate out a cluster $C$, we will call $\pi(i)$ the \emph{center} of the cluster $C$. Also, we will say that a vertex $u$ is separated while considering $\pi(i) \in T$ if some neighbour of $u$ is in $C$ but $u \notin C$, and neither $u$ nor its neighbours are in any cluster whose center is $\pi(1), \pi(2) \ldots \pi(i-1)$.  

\begin{lemma}
In each iteration, the output set $U'$ satisfies the following properties.
\begin{enumerate}
    \item $E[U' \cap T] \geq \frac{s}{6}$.
    \item The expected number of vertices cut, $E[|N_G(U')|] \leq {\OO(\log s)LP}$.
\end{enumerate}
\end{lemma}

\begin{proof}
Clearly, for any $v \in T$ $y_v \in [\delta, 2 \delta]$ whenever $\delta \in [\frac{y_v}{2}, y_v]$ which happens with probability at least $\frac{y_v}{2}$. This means that with probability at least $\frac{y_v}{2}$, $v$ is in some cluster $W \in \SS$. It follows that $E[U' \cap T] \geq \frac{1}{3} \sum_v \frac{y_v}{2} \geq \frac{s}{6}$.

Next, we prove (2). Consider a vertex $u$. First, suppose that $b_u \geq \frac{\delta}{100}$. For a given vertex $u$, the probability that this happens is at most $100b_u$. Now suppose that $b_u < \frac{\delta}{100}$, and let us consider the probability that the vertex $u$ is cut while considering the ball around the terminal vertex $w$.

In order to cut $u$ while growing the ball from $w$, there must be a neighbour $u'$ of $w$ inside the ball of $w$ of radius $0.1y_w$. Hence we must have $d(u,w) \leq 0.1y_w + b_w \leq 0.1y_w + \frac{\delta}{100} \leq 0.1y_w + 0.1y_w \leq 0.2y_w$. This means, using the triangle inequalities in the LP constraints, we must have $0.8y_w \leq y_u \leq 1.2y_w$. But $y_w \in [\delta, 2\delta]$, and hence $0.8 \delta \leq y_u \leq 2.4 \delta$, or equivalently
 $\frac{y_u}{2.4} \leq \delta \leq \frac{y_u}{0.8}$. Thus, the only values of $\delta$ for which the vertex $u$ can be cut must satisfy $\frac{y_u}{2.4} \leq \delta \leq \frac{y_u}{0.8}$.

Now, fix a value of $\delta$ such that either $\frac{y_u}{2.4} \leq \delta \leq \frac{y_u}{0.8}$.

Let $X' \subseteq X$ be the set of vertices $w$ of $X$ that satisfy $d(w,u) \leq 0.2y_w $ . Informally, these are the only vertices that can separate $u$, thereby adding $u$ as a cut vertex. We show first that $|X'| \leq 5s$. Suppose to the contrary that $|X'| > 5s$. Then at least $5s + 1$ vertices $w \in X'$ that satisfy  $d(w,u) \leq 0.2y_w$. Fix some $w' \in X'$. 
By the triangle inequality, for every $w \in X'$ we have $y_w + 0.2y_w \geq y_u \geq y_w - 0.2y_w = 0.8y_w$. Thus $d(w,u) \leq \frac{y_u}{4}$. Also $d(w',u) \leq \frac{y_u}{4}$. By the triangle inequality on $d$, we must have $d(w',w) \leq d(w,u) + d(w,w') \leq \frac{2y_u}{4} = \frac{y_u}{2} \leq 0.6y_w$. But this means $|Ball_{0.6}(w')| \geq 5s$, which contradicts \Cref{obs:ball}.

Order the vertices $w$ of $X'$ in increasing order of $\frac{d(u,w)}
{y_w}$, and let $Q = (t_1, t_2 \ldots t_{|X'|})$ denote this sequence. Notice that if some $t_i$ is considered before $t_j$ for some $i < j$ in $\pi$, then the vertex $u$ will already be cut/removed in one cluster when $t_i$ is considered and hence we cannot cut this edge while considering $t_j$.

\begin{align*}
    Pr[u\;\text{cut} | \delta]\ &\leq \sum_{i = 1}^{|X'|} Pr[u \text{\;cut while considering $t_i$}]\\
                                  &\leq  \sum_{i = 1}^{|X'|} Pr[t_i \;\text{appears before $t_1, t_2$} \ldots t_{i-1}\\
                                  &\text{and } d(u,t_i) - b_u \leq Cy_{t_i} \leq d(u,t_i)]\\
                                  &\leq \sum_{i = 1}^{|X'|} \frac{b_u}{0.05y_{t_i}i}\\
                                  &\leq  100\log s \cdot \frac{b_u}{\delta}
\end{align*}
The last step follows since $|X'| \leq 10s$. Finally, as observed before, the vertex $u$ can only be cut when $\frac{y_u}{2.4} \leq \delta \leq \frac{y_u}{0.8}$. We now obtain
\begin{align*}
Pr[u\;\text{cut}] &\leq \int_{\frac{y_u}{2.4}}^{\frac{y_u}{0.8}}Pr[u\;\text{cut} | \delta ] Pr[\delta] d \delta  \\
&\leq 100 \log s \cdot  \int_{\frac{y_u}{2.4}}^{\frac{y_u}{0.8}}\frac{b_u}{\delta} d \delta \\
&\leq 200\log s \cdot b_u.
\end{align*}\end{proof}

Finally, this above expression only bounds the probability that $u$ is cut assuming that $b_u \leq \frac{\delta}{100}$. To account for the other case, as discussed above, we note that the probability that $b_u \geq \frac{\delta}{100}$ is at most $100b_u$. Thus the total probability that $u$ is cut is at most $300\log s \cdot b_u$. The rest of the analysis is similar to~\cite{minmax}, however, we present it again for completeness' sake. For a set $X \subseteq V(G)$, recall that we defined  $f(X) = |X \cap T| - \frac{s}{2000LP\log s}|N_G(X)|$ where $LP$ is the optimal LP value. Then if $U'$ is the set returned by the algorithm, we must have $E[f(U')] = E[U' \cap T] - \frac{s}{2000LP\log s}E[N_G(U')] \geq \frac{s}{6} - \frac{3s}{20} \geq \frac{s}{60}$. Also note that $f(U') \leq n$. It follows that $E[f(U')] \leq n \cdot Pr[f(U') > 0]$, which means $Pr[f(U') > 0] \geq \frac{E[f(U')]}{n} \geq \frac{s}{60n}$. Thus repeating the algorithm $\OO(n\log n)$ times, we are guaranteed to find a set with $f(U') > 0$ with high probability. But $f(U') > 0$ means that $U'$ satisfies  $|N_G(U')| \leq \OO(\log s) \frac{LP}{s}|U' \cap T| \leq \OO(\log s)\frac{k}{s}|U' \cap T|$. Note that since we assumed $s = \Omega(k \log s)$, we must have $|N_G(U')| \leq \frac{1}{2}|U' \cap T|$. Finally, we have $|(U' \cup N_G(U')) \cap T| \leq \frac{1.5|T|}{2}$ and thus $|(V(G) \setminus (U' \cup N_G(U')) \cap T| \geq \frac{|T|}{4} = \Omega(U' \cap T)$. It follows that $U'$ must be $\OO(\frac{|N_G(U')|}{|(U' \cup N_G(U')) \cap T|}) \leq \OO(\log s)\phi$-terminal sparse.\end{proof}

\end{document}